\documentclass[lettersize,journal]{IEEEtran}
\usepackage{amsmath, amsfonts}
\usepackage{array}
\usepackage[caption=false, font=normalsize, labelfont=sf, textfont=sf]{subfig}
\usepackage{textcomp}
\usepackage{stfloats}
\usepackage{url}
\usepackage{verbatim}
\usepackage{graphicx}
\def\BibTeX{{\rm B\kern-.05em{\sc i\kern-.025em b}\kern-.08em
    T\kern-.1667em\lower.7ex\hbox{E}\kern-.125emX}}
\usepackage{balance}

\usepackage{epsfig}

\usepackage{amssymb}
\usepackage{amsthm}

\usepackage{cite}
\usepackage[dvipsnames]{xcolor}

\usepackage{tabstackengine}
\setstackEOL{\cr}
\usepackage{booktabs}
\usepackage[utf8]{inputenc}
\usepackage{lscape}
\usepackage{subcaption}
\usepackage{longtable}

\usepackage{makecell}

\usepackage{algorithm}
\usepackage{algorithmic}

\usepackage[disable
]{todonotes}

\usepackage{rotating}

\usepackage{parnotes}

\usepackage[suppress
]{color-edits}
\addauthor[Yoram]{yb}{red}
\addauthor[Berk]{bi}{blue}

\usepackage[breaklinks=true, bookmarks=false]{hyperref}
\usepackage{enumitem} 

\usepackage{jabbrv}

\usepackage{wrapfig}

\newcommand{\Real}{\mathbb{R}}

\newcommand{\xb}{\mathbf{x}}

\newcommand{\commentout}[1]{}

\newtheorem{theorem}{Theorem}

\newtheorem{lemma}[theorem]{Lemma}

\newtheorem{innercustomgeneric}{\customgenericname}
\providecommand{\customgenericname}{}
\newcommand{\newcustomtheorem}[2]{%
  \newenvironment{#1}[1]
  {%
   \renewcommand\customgenericname{#2}%
   \renewcommand\theinnercustomgeneric{##1}%
   \innercustomgeneric
  }
  {\endinnercustomgeneric}
}

\newcustomtheorem{customthm}{Theorem}
\newcustomtheorem{customlemma}{Lemma}

\makeatletter
\newcommand*{\shifttext}[2]{%
  \settowidth{\@tempdima}{#2}%
  \makebox[\@tempdima]{\hspace*{#1}#2}%
}
\makeatother

\begin{document}
\bstctlcite{IEEEexample:BSTcontrol}

\title{RED-PSM: Regularization by Denoising \\ 
of {Factorized Low Rank} Models\\
for Dynamic Imaging}

\author{Berk Iskender$^1$\footnotemark $\quad$ Marc L. Klasky$^2$ $\quad$ Yoram Bresler$^1$\\
\small{$^1$University of Illinois at Urbana-Champaign,
IL, USA $\quad$ $^2$Los Alamos National Laboratory, NM, USA}\\
}

\maketitle
\makeatletter{
\renewcommand*\@makefnmark{}
\footnotetext{This research was supported in part by Los Alamos National Labs under Subcontract No. 599416/CW13995.}\makeatother}

\begin{abstract}
Dynamic imaging addresses the recovery of a time-varying 2D or 3D object at each time instant using its undersampled measurements. In particular, in the case of dynamic tomography, only a single projection at a single view angle may be available at a time, making the problem severely ill-posed. We propose an approach, RED-PSM, which combines for the first time two powerful techniques to address this challenging imaging problem. The first, are non-parametric factorized low rank models, also known as partially separable models (PSMs), which have been used to efficiently introduce a low-rank prior for the spatio-temporal object. The second is the recent \textit{Regularization by Denoising (RED)}, which provides a flexible framework to exploit the impressive performance of state-of-the-art image denoising algorithms, for various inverse problems. We propose a partially separable objective with RED and a computationally efficient and scalable optimization scheme with variable splitting and ADMM. Theoretical analysis proves the convergence of our objective to a value corresponding to a stationary point satisfying the first-order optimality conditions. Convergence is accelerated by a particular projection-domain-based initialization. We demonstrate the performance and computational improvements of our proposed RED-PSM with a learned image denoiser by comparing it to a recent deep-prior-based method known as TD-DIP. Although the main focus is on dynamic tomography, we also show {performance advantages} of RED-PSM in a cardiac dynamic MRI setting.
\end{abstract}


\section{Introduction}
\IEEEPARstart{T}{ime-varying} or dynamic tomography involves the reconstruction of a dynamic object using its projections acquired sequentially in time. The problem arises in micro tomography \cite{maire201620}, myocardial perfusion imaging \cite{10.1093/ehjci/jew044}, thoracic CT \cite{keall2004acquiring}, imaging of fluid flow processes \cite{scanziani2018situ, o2011dynamic} and dynamic imaging of material samples undergoing compression\cite{patterson2016situ, patterson2020synchrotron}. Also, it is closely related to the dynamic MRI (dMRI) problem, which typically arises in cardiac imaging \cite{salerno2017recent}.

Dynamic tomography is a challenging ill-posed inverse problem: since the measurements are inconsistent due to the evolving object, traditional reconstruction algorithms lead to significant artifacts. As discussed further below, although numerous methods have been proposed to address this problem, they suffer from various limitations, or provide less than satisfactory reconstruction quality, especially in the scenario of main interest in this paper, when at each time instant only one projection is available.

\subsection{Proposed Approach}
\label{sec:intro_prop_method}
We propose an approach, RED-PSM, which combines for the first time two powerful techniques to address this challenging imaging problem. 
The first technique, are {non-parametric} {\textit{factorized low-rank}} object models, 
{popularized under the name} \emph{partially separable models (PSMs)} \cite{liang2007spatiotemporal}.
{PSM was} 
combined with low-rank matrix recovery in \cite{Haldar2010, zhao2010low, goud2010real}, and used since then to represent or motivate a \textit{low-rank object prior} in {tens of works on dynamic imaging.}
The second technique is the recent \textit{Regularization by Denoising (RED)} \cite{romano2017little}, which provides a flexible framework to exploit the impressive performance of state-of-the-art image denoising algorithms, for various inverse problems.

We propose a {non-convex} partially separable objective with RED for {learning-regularized low-rank spatio-temporal object recovery} and a computationally efficient and scalable optimization scheme with {bi-convex} ADMM.

Theoretical analysis proves the convergence of our objective to a value corresponding to a stationary point satisfying the first-order optimality conditions. Convergence is accelerated by a particular projection-domain-based initialization.

We demonstrate the performance and computational improvements of our proposed RED-PSM with a learned image denoiser by comparing it to a recent deep-prior-based method known as TD-DIP \cite{yoo2021time}, and to spatial and spatiotemporal total variation (TV) regularized versions of PSM {with a bi-convex objective}.

Although the main focus is on dynamic tomography, we also show the performance advantages of RED-PSM in a cardiac dynamic MRI setting.

RED was originally inspired by the Plug-and-Play (PnP) approach \cite{Venkatakrishnan2013}, which was the first to exploit powerful denoisers to represent the prior on the object in inverse problems but without explicitly defining the regularizer. 
Instead, in PnP the denoiser is incorporated by replacing a proximal mapping.
Both PnP and RED have shown strong empirical performance in various applications. A recent survey \cite{kamilov2023plug} provides a detailed discussion of the main results and different applications of PnP and the relation to RED. For the algorithmic variations PnP-ADMM and PnP-ISTA there are convergence results to unique fixed points. Moreover, for PnP-ISTA, similar to our result, convergence is shown to a stationary point of a possibly non-convex objective when the denoiser is an MMSE estimator \cite{xu2020provable}. This result requires the prior to be non-degenerate, i.e., not lie on a lower dimensional manifold, which, unfortunately, would be violated by the low-rank bilinear representation in our PSM.
We believe that a PnP-based method similar to RED-PSM can be formulated for our problem, but a proof of convergence to a stationary point may not directly follow, and a different approach and analysis may be needed.

Our motivation for using the original RED formulation in RED-PSM rather than more recent formulations and analyses (e.g., as in \cite{cohen_regularization_2021} and the references therein) {{are}} its simple gradient expression, and explicit regularizer expression facilitating implementation and theoretical analysis. 
The objective is to demonstrate the large improvement in performance that this simple learned regularizer can bring to the PSM scheme, and the improvement over much more complex and slower state-of-the-art methods. 
Recent variations on RED such as in \cite{cohen_regularization_2021}, and PnP methods (as covered in \cite{kamilov2023plug}) have improved over the original methods in reconstruction quality in various static imaging problems and/or in theoretical guarantees. 
However, theoretical guarantees for the case of a non-convex data fidelity term have been provided recently for only one such non-convex scenario - phase-retrieval \cite{xue_convergence_2022}. 
As this is the first work to extend RED to spatio-temporal imaging with a PSM, and to provide a convergence guarantee for the method in the non-convex bilinear scenario, we chose to limit our attention in this paper to the original RED version \cite{romano2017little, reehorst2018regularization}.

\subsection{Previous Work}
{We limit our discussion to dynamic CT and dMRI, although similar or analogous methods have been or can be applied in other dynamic imaging modalities. The same can be said of the method proposed in this paper.}

\subsubsection{Dynamic MRI (dMRI)}
{In spite of major advances in hardware and signal acquisition methods, MRI has remained a relatively slow modality, posing significant challenges for dynamic imaging. This, coupled with the high flexibility in the acquisition strategy, has motivated extensive work on dMRI, with hundreds of papers on the subject. We highlight some key developments in the algorithm-based dMRI techniques.}

{Fundamentally, algorithmic approaches to dMRI are based on identifying and exploiting redundancy in the dynamic object, which enable its recovery from an incomplete, and usually sparse set of samples acquired in the joint k-space and time domain -- the $k$-$t$-space. Many of the proposed methods have extensions to parallel imaging, which introduces additional such redundancy by additional hardware.}

{A large set of approaches rely on modeling the support of the object in the domain dual to the $k$-$t$-domain --  the ($x$-$f$) domain (the space -- temporal frequency domain). These approaches may be classified into three groups. 
The first involves generic modeling of the ($x$-$f$) support as in reduced field-of-view (FOV) techniques \cite{Hu1994,Brummer2004,Madore1999}. The second involves sparse and unknown support modeling in a transform domain, and sampling and recovery using the methods of compressed sensing \cite{Lustig2006,Gamper2008,jungImprovedBLASTSENSE2007,Jung2009}. The third involves adaptive sparse support modeling  \cite{Zhao2001,Aggarwal2002,aggarwalPARADIGM2008,SharifPARADISE2010}, where the support information is adapted to the dynamics of the imaged slice, and the theory of time-sequential sampling \cite{Willis1997,Willis1997a} is used to design an optimal $k$-$t$ sampling and reconstruction scheme.  Other related approaches model prior information such as motion periodicity and related harmonic structure \cite{Liang1994,Liang1997}, or the object’s spatial and temporal correlation functions \cite{Tsao2003}.}

{{Since its introduction,} the partially-separable model (PSM) \cite{liang2007spatiotemporal}, with a {non-parametric} {factorized} low-rank form separating the spatial and temporal structure, has been used {extensively in many works} {on dMRI} to represent the underlying spatio-temporal object. 
{A group of methods {recover the PSM in two steps: (a) estimate the temporal subspace using data (e.g., a navigator signal) sampled at low spatial, but high temporal resolution; and (b) estimate the spatial subspace using the entire acquired data set. Such methods include} 
e.g., \cite{liang2007spatiotemporal, pedersen2009k, brinegar2010improving, haldar2011low, zhao2011further}, {with followup versions} promoting sparsity of the object in a transform domain (e.g., \cite{zhao2012image}). 
Other PSM approaches {use a low-rank matrix recovery formulation} to estimate spatial and temporal components jointly {(free of conditions or additional data acquisition enabling a separate reconstruction of the temporal subspace)}, e.g., \cite{Haldar2010, zhao2010low}), also with {versions} promoting sparsity \cite{goud2010real, lingala2013blind}.}}

{{Manifold models have also been considered for dMRI. Recent \ybedit{such} work \cite{djebra2022manifold} surveys previous related works, and proposes a new method that is related to PSM. The method}} partitions the dynamic object into temporal subsets each lying on a low-dimensional manifold, approximated by its tangent subspace. As a result, the object is represented by a temporally partitioned PSM. A simple sparsity penalty is used for the spatial representation. The spatial basis functions and linear transformations aligning the temporal subsets are determined by optimizing a non-convex objective using an ADMM algorithm. {{The subset temporal partitioning and temporal bases are estimated using an initial reconstruction obtained from training data. In some applications, such data may not be available.}}

{Another set of methods, {including \cite{lingala2011accelerated, Majumdar2011, majumdar2013non, majumdar2015learning},} promotes low-rankness {of the entire spatio-temporal object matrix or of patches thereof \cite{Trzasko2011,Trzasko2013}  implicitly by using the nuclear norm or Schatten-p functional with $p<1$,} {often promoting sparsity in a transform domain as well.} Methods combining motion estimation and compensation with other priors are also used widely for dMRI, e.g., \cite{lingala2014deformation, chen2014motion, yoon_motion_2014, tolouee2018nonrigid}.}

{A different approach in dMRI {decomposes}  the object into 
the sum of low-rank and sparse components, {with the low-rankness} encouraged implicitly using nuclear or a Schatten-p functional \cite{tremoulheac2014dynamic, chiew2015k, otazo2015low, ravishankar2017low}. A recent related fast method \cite{babu2023fast} decomposes the object representation into {the sum of} three components: the mean signal; a low-rank PSM; {and a residual sparse in the Fourier domain.}} Although benchmark competing methods are either faster or more accurate on some of the diverse data sets used for the comparison in \cite{babu2023fast}, the method provides the lowest reconstruction error and runtimes when these quantities are averaged over these data sets.

{{Recent approaches to dMRI reconstruction have used deep neural networks (DNNs). A majority of these methods, such as
\cite{Schlemper2018,Qin2019,Biswas2019,hauptmann2019real,Kuestner2020,Huang2021,Ke2021,Ke2021a,Sandino2021,Ghodrati2021,Kuestner2021,Qi2021,Hammernik2021,Wang2022}, are supervised, requiring a large set of high-resolution
artifact-free images for network training. However,
acquiring such training data can be difficult for some  dMRI applications, and may be impossible for many other dynamic imaging applications, where clinicial dMRI strategies for freezing motion such as ECG synchronization and breath-hold by a cooperative patient, are not available.}}

{{The need  for a training data set is overcome by} 
recent object-domain deep image prior (DIP)-based {DNN algorithms for dMRI} \cite{yoo2021time, zou2021dynamic}, {which are unsupervised}. Providing impressive results, DIP-based algorithms such as \cite{yoo2021time} suffer from overfitting and usually require handcrafted early stopping criteria during the optimization of generator network parameters. To overcome these drawbacks, \cite{zou2021dynamic} includes regularization constraining the geodesic distances between generated objects. However, this requires the computation of the Jacobian of the generator network at each iteration of the update of the weights, significantly increasing the computation and run time.}

{A different unsupervised approach \cite{ahmed2022dynamic} combines factorized low rank (i.e., partially separable) {with} generative models.} Although it combines the PSM with the recent DIP framework, this method has the following limitations: (i) The spatial generator is an artifact removal network taking the full-sized spatial basis functions as input. Since this prevents a patch-based implementation,  it may be computationally problematic for high-resolution 3D+temporal settings; (ii) The CNN {structural} prior for natural images {that is used in the spatial generator} may not be useful as a prior for the individual spatial basis functions, since the least-squares optimal spatial basis functions are the left singular vectors of the complete object, and as such may not have the structure of natural images; (iii) As with the other DIP-based methods, if the additional penalties on the generator parameters are insufficient, this method can be prone to overfitting. {{These various limitations are overcome by the proposed approach.}}

\subsubsection{Dynamic Tomography (dCT)}

{CT is faster than MRI, thanks to acquiring an entire projection at one view angle at once. However, because of the angular scanning required in almost all systems, in spite of significant advances in hardware, {dCT} is still a relatively slow modality. It is also the problem that motivated much of our work reported in this paper.}
 
{In perhaps the earliest algorithmic approaches to dCT, the time-sequential sampling problem in d-CT was formulated and studied in \cite{willis1990tomographic}, and \cite{willis95, willis95_2} provided a solution using time-sequential sampling of bandlimited spatio-temporal signals including an optimal view angle order for the scan and theoretical guarantees for unique and stable reconstruction. However, the approach is limited by its bandlimitedness assumptions.}
 
{A class of algorithms for dCT {are based on modeling} the time variation by a motion field, and {motion-compensated reconstruction.} {Many of} these algorithms, e.g., \cite{ma11081395, zang2018space, jailin2021projection, capostagnoDeformableMotionCompensation2021c}  alternate between estimating the motion field and the time-varying object. {Some algorithms, e.g., \cite{burger2017variational} use the optical flow (mass continuity)} PDE. A different method \cite{hahnUsingNavierCauchyEquation2022} {estimates the motion field separately by evolving a linear elastic deformation by the Navier-Cauchy PDE, assuming known boundary evolution of the object and initial density distribution.} However, these methods assume the total mass {or density} to be preserved over time, which may be limiting assumptions. For instance, {flow of a contrast agent into the imaged volume, or} imaging a fixed slice of a 3D time-varying object under compression {with cross-slice motion} may violate either of these assumptions. As a {hardware-based, problem-specific} approach, optical tracking of fiducial marks {has been used} for motion correction \cite{brombalMotionArtifactsAssessment2021}.}

{{Various methods approach dCT reconstruction as Bayesian estimation. These include methods} \cite{georgeTimeresolvedCardiacCT2008a, butalaTomographicImagingDynamic2009, hakkarainenUndersampledDynamicXRay2019, caoUndersampledDynamicTomography2022} that employ a state-space formulation and use Kalman filter techniques {to approximate an MMSE estimate}. {A different Bayesian method} \cite{mohan2015timbir} incorporates a spatiotemporal Markov random field object prior, and models measurement imperfections. An iterative algorithm {is used to compute a MAP estimate.} An interlaced projection acquisition scheme is also proposed.}

{Several dCT methods {motivated by} compressed sensing {use} a total-variation penalty along both spatial and temporal coordinates {as in} \cite{luckaEnhancingCompressedSensing2018c}, while other methods, e.g., \cite{bubbaSparseDynamicTomography2020} propose efficient sparse representations for the {dynamic object.} {Incorporating} low-rank modeling \cite{gao2011robust} decomposes {the object} into {low-rank+sparse} components, {promoting the low rankness} implicitly using the nuclear norm. The method {relies on the object being approximately} static for groups of projections, enabling an approximate reconstruction from each such group.}

{The PSM model has been introduced into dynamic tomography \cite{iskender2022dynamic, iskender2022dynamic_long} by carrying the PSM to the projection domain using the harmonic representation of projections, and estimating the spatial and temporal basis functions jointly. The work provides a theoretical analysis of uniqueness and stability and the choice of time-sequential angular scan scheme for the problem. Despite the advantages of this approach, its performance is still limited by the null space of the measurement operator.}

{{To help address the underdetermined problem, the method of \cite{iskender2023factorized}} improves over \cite{iskender2022dynamic, iskender2022dynamic_long} by combining a PSM with  basic spatial regularization {by total variation.} {Like \cite{iskender2022dynamic, iskender2022dynamic_long} this method too} estimates the temporal and spatial PSM components jointly. This makes it applicable to the time-sequential acquisition scenario considered in this paper{, where only a single projection is available per time frame.} It imposes low rankness as a soft constraint, using a hybrid {biconvex} objective with a data fidelity term and a TV regularizer expressed in terms of the unfactorized object, and a penalty for object-PSM mismatch. It {uses} \cite{iskender2022dynamic, iskender2022dynamic_long} {for fast initialization.} The method \cite{iskender2023factorized} can be considered a precursor to the method of this paper, which combines {a hard rank constrained PSM with a sophisticated learned spatial regularizer,} and provides theoretical convergence analysis.}

{Recent methods \cite{zhiArtifactsReductionMethod2019, huangUnetbasedDeformationVector2020, madestaSelfcontainedDeepLearningbased2020,  zhangDeepLearningbasedMotion2023} introduce deep learning {into} {dCT}. {These methods require an initial reconstruction from the subsets of data from which the motion can be estimated accurately. However, {without a periodicity assumption,} this is not applicable to the scenarios in this paper, where only a single projection is available at a time. {Another method} \cite{majee2021multi} {for cone-beam dCT} proposes a multi-agent consensus equilibrium \cite{buzzard2018plug}-based technique using separate deep denoisers for axial, sagittal, and coronal slices of the object.}}

{{A} deep-learning-based method \cite{lozenski2022memory} {for dynamic photoacoustic tomography (PACT)}  introduced a neural field (NF) to represent the dynamic object with {fewer} parameters, {combining it} with a simple total variation (TV)-based regularizer. Despite providing an efficient representation for the dynamic object, the method outperforms its comparison benchmarks only in some of the projection-per-frame settings.}

\subsection{Contributions}
\label{sec:contributions}
\textbf{1.} {To the best of our knowledge, RED-PSM is the first PSM-based approach to dynamic imaging that 
incorporates a {particular}
(RED-based \cite{romano2017little}) spatial prior {that is} learned {and} pre-trained.}

\textbf{2.} {We are not aware of any prior work that uses RED with an explicit low-rankness constraint. In RED-PSM we achieve parsimonious representation for the dynamic object using PSM, {{and} reduce the data requirements further by incorporating the RED prior.}}
   
\textbf{3.}
Unlike supervised {learning-based} methods for spatio-temporal imaging \cite{Schlemper2018,Qin2019,Biswas2019,hauptmann2019real,Kuestner2020,Kuestner2021,Huang2021,Ke2021,Ke2021a,Sandino2021,Ghodrati2021,Qi2021,Hammernik2021,Wang2022,shi2015convolutional,wang2022spatial,ginio2023efficient} that learn a spatio-temporal model from training data, RED-PSM does not require ground-truth spatio-temporal training data, which is often not available or expensive to acquire. Thus RED-PSM offers the best of both worlds: a learned spatial model using readily available \textit{static} training data, and unsupervised single-instance adaptation to the spatio-temporal measurement.
    
\textbf{4.} A novel and effective ADMM algorithm for the resulting new minimization problem enforces PSM as a hard constraint and incorporates the learned RED spatial regularizer.
   
\textbf{5.} The method is supported by theoretical analysis:
we show convergence of the proposed PSM-based objective to a stationary point, and do so for a highly non-trivial learned regularizer. To date, the convergence of RED has been analyzed in the {nonlinear measurement} setting only for the phase retrieval problem \cite{xue_convergence_2022}. 
Our analysis of RED-PSM is the first convergence result of RED in the bilinear (bi-convex), non-convex scenario.
This convergence guarantee is of practical importance.
First, it provides a mathematical justification to terminate the algorithm after sufficiently many iterations. Second, it eliminates concerns about the solution degrading after too many iterations as in the case of DIP-based methods.
    
\textbf{6.} Compared to a recent DIP-based \cite{yoo2021time} algorithm, RED-PSM achieves better reconstruction accuracy with orders of magnitude faster run times.

\textbf{7.} To improve and speed up the empirical convergence of RED-PSM, we use a particular fast projection-domain PSM initialization scheme \cite{iskender2022dynamic, iskender2022dynamic_long} for the spatial and temporal basis functions. The accelerated and reliable convergence with this initialization scheme is important for applications that require fast turn-around between acquisition and reconstruction.
   
\textbf{8.} A version of the approach with a patch-based regularizer is shown to provide almost equivalent reconstruction accuracy. This makes the proposed method conveniently scalable to high-resolution 3D or 4D settings.

An earlier and partial version of this work, missing, among other things, the detailed discussion of previous work, convergence analysis, and some of the experimental results, was presented at a conference \cite{iskender2023red}.

\section{Problem Statement}
\label{sec:prob_statement}
\begin{figure}[b]
    \centering
    \includegraphics[width=0.9\linewidth]{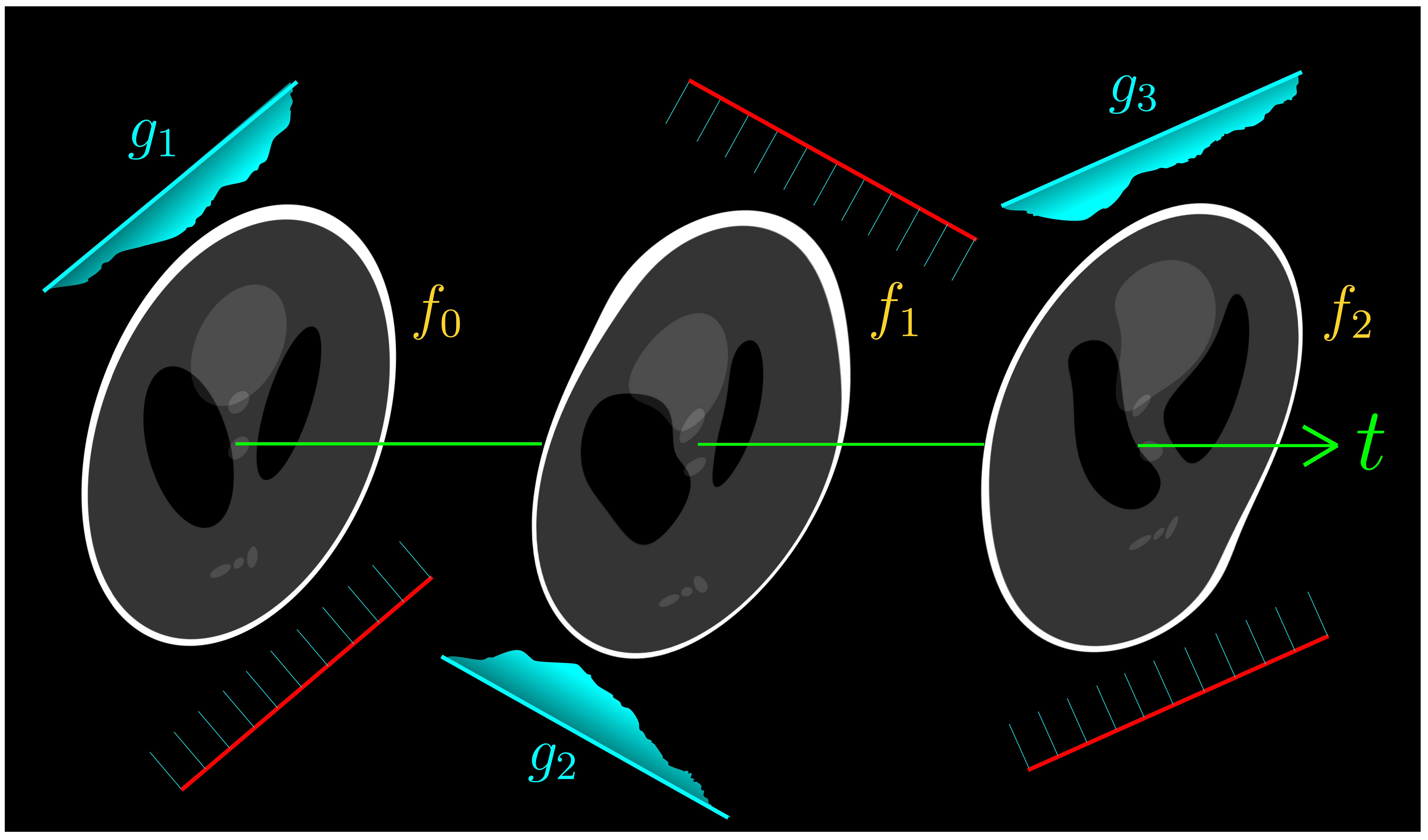}
    \caption{\small Imaging geometry for time-varying tomography of the object $f_t$ with single measurement $g_t$ at each time instant for $t \in \{0,1,2\}$.}
    \label{fig:problem_statement}
\end{figure}

In a 2D setting, illustrated in Figure \ref{fig:problem_statement}, the goal in the dynamic tomography problem is to reconstruct a time-varying object $f(\xb,t)$, $\xb \in \Real^2$ vanishing outside a disc of diameter $D$, from its projections  
\begin{align*}
    g(\cdot,\theta,t) = \mathcal{R}_\theta \{f(\mathbf{x},t)\}
\end{align*}
obtained using the Radon transform operator $\mathcal{R}_\theta$ at angle $\theta$. Considering time-sequential sampling, in which only one projection is acquired at each time instant, and sampling uniform in time, the acquired measurements are
\begin{equation}
\label{eq:sampled_proj}
    \{g(s,\theta_p, t_p)\}_{p=0}^{P-1}, \,\, \forall s, t_p = p\Delta_t,
\end{equation}
where $s$ is the offset of the line of integration from the origin (i.e., detector position), and $P$ is the total number of projections (and temporal samples) acquired. The sampling of the $s$ variable is assumed fine enough and is suppressed in the notation. The angular sampling scheme, the sequence $\{\theta_p\}_{p=0}^{P-1}$, with $\theta_p \in[0, 2 \pi]$, is 
a free design parameter.

Our objective in dynamic tomography is to reconstruct the underlying dynamic object  $\{f(\xb,t_p)\}_{p=0}^{P-1}$ from the time-sequential projections in \eqref{eq:sampled_proj}. 
The challenge is that because each projection belongs to a different object, the projections in \eqref{eq:sampled_proj} are inconsistent. Therefore, a conventional, e.g., filtered backprojection (FBP) reconstruction as for a static object results in significant reconstruction artifacts. This is to be expected, as the problem is severely ill-posed: an image with a $D$-pixels diameter requires more than $D$ projections for artifact-free reconstruction, whereas only one projection is available per time-frame in the time-sequential acquisition \eqref{eq:sampled_proj}.
{Several} dynamic tomography methods \cite{mohan2015timbir, zang2018space, majee2021multi} group temporally neighboring projections and assume the object is static during their acquisition. However, this reduces the temporal resolution, and any violation of this assumption (as in the case of time-sequential sampling) leads to mismodeling in the data fidelity terms. 
\section{Partially Separable Models (PSM)}
\label{sec:psm}
For spatio-temporal inverse problems such as dynamic MRI and tomography, the underlying object can be accurately represented using a partially-separable model (PSM), which effectively introduces a {{factorized}} low-rank prior to the problem. For dynamic tomography, a PSM can represent the imaged object $f$, or its full set of projections $g$. In this paper we use an object-domain PSM.

The representation of a dynamic object $f(\mathbf{x},t)$ by a $K$-th order partially separable model (PSM) is the series expansion
\begin{equation} \label{eq:parsep}
    f(\mathbf{x},t) = \sum_{k=0}^{K-1} \Lambda_{k}(\mathbf{x}) \psi_{k}(t).
\end{equation}
This model facilitates interpretability by separating the spatial structure from the temporal dynamics. Empirically, modest values of $K$ provide high accuracy in applications to MR cardiac imaging \cite{Haldar2010,ma2019dynamic,ma2021dynamic}. Theoretical analysis  \cite{iskender2022dynamic_long} shows that for a spatially bandlimited object undergoing a time-varying affine transformation (i.e, combination of time-varying translation, scaling, and rotation) of bounded magnitude, a low order PSM provides a good approximation.

As detailed later, as a standard choice in modeling temporal functions, we use a $d$-dimensional representation with $d \ll P$. 
A similar idea was used in \cite{zhao2010low} to represent the low temporal bandwidth of cardiac motion in dMRI. Together with the PSM for the object this leads to a significant reduction in the number of representation parameters for a spatio-temporal object with a $D$-pixel diameter and $P$ temporal samples: from $\approx PD^2$ to $\approx KD^2 + Kd$, with $K \ll P$. 
Because $d \ll D$, this corresponds to a factor of $P/K \gg 1$ compression of the representation, providing an effective parsimonious model for the spatio-temporal object $f$.

Also, by propagating the PSM object model to the projection domain, it enables quantitative analysis \cite{iskender2022dynamic_long} of the choice of optimal sequential projection angular schedule.
\section{Proposed Method: RED-PSM}
\label{sec:proposed_method}
The overall RED-PSM framework for recovering dynamic objects explained in this section is illustrated in Figure \ref{fig:framework}.
\begin{figure}
    \centering
    \includegraphics[width=\linewidth]{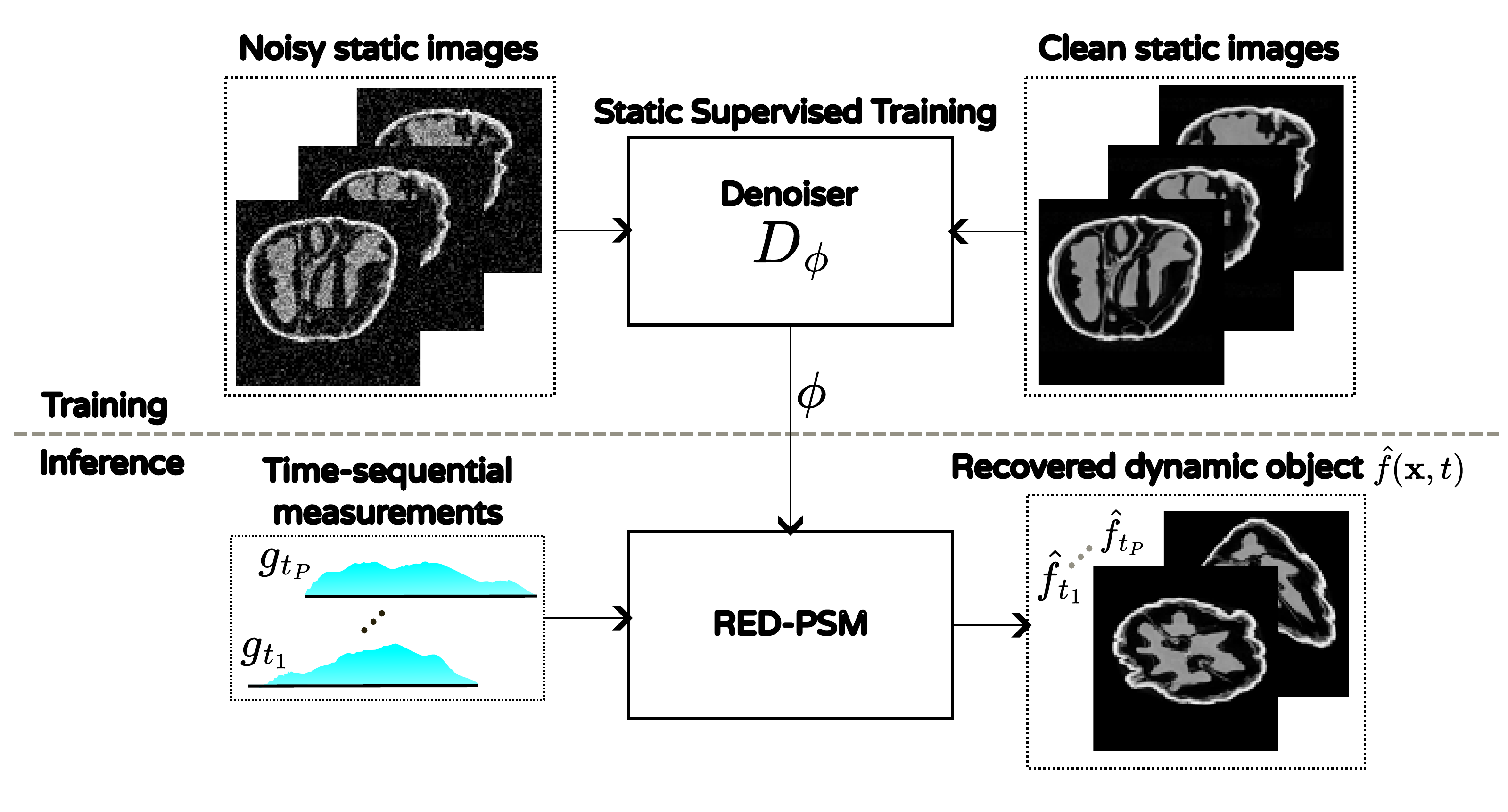}
    \setlength{\abovecaptionskip}{-16pt}
    \setlength{\belowcaptionskip}{-16pt}
    \caption{\small 
    The RED-PSM framework. The deep denoiser $D_\phi$ is trained on slices of  \textit{static} objects similar to the object of interest, and the learned spatial prior is used at inference time.}
    \label{fig:framework}
\end{figure}
\subsection{Variational Formulation}
\label{sec:var_form}
We use a discretized version of the PSM \eqref{eq:parsep} for the dynamic object, with the object $f(\cdot, t)$ at each time instant $t=0,1, \ldots, P-1$, represented by a $N\times N$-pixel image (a ``time frame", or ``snapshot"). 
Vectorizing these images to vectors $f_t \in \mathbb{R}^{N^2}$, the entire dynamic object is the $N^2 \times P$ matrix $f = [ f_0 \ldots f_{P-1}]$. %
{It will also be useful to extract individual time frames from $f$. Denoting  the $t$-th column of the $P \times P$ identity matrix by $e_t$, we have $fe_t= f_t$, i.e., $e_t$ extracts the $t$-th column of $f$. Applying the PSM model, we assume $ f = \Lambda\Psi^T \in \mathbb{R}^{N^2 \times P}$, where} the columns $\Lambda_k$ and $\Psi_k$ of $\Lambda \in \mathbb{R}^{N^2 \times K}$ and $\Psi \in \mathbb{R}^{P \times K}$ are the discretized spatial and temporal basis functions for the PSM representation, respectively.

Assuming that the x-ray detector has $N$ bins, the projection of the object at time $t$ is $g_t = g(\cdot, \theta_t, t) = R_{\theta_t} f_t \in \mathbb{R}^{N}$, where the measurement matrix $R_{\theta_t} \in \mathbb{R}^{N \times N^2}$ computes the projection at view angle $\theta_t$.

We formulate the recovery of $f$ as the solution $\hat{f}= \hat{\Lambda}\hat{\Psi}^T$ to the following variational problem
\begin{align}
\label{eq:variational_PSM_reg}
  (\hat{\Lambda}, \hat{\Psi})=\arg &\min_{\Lambda, \Psi}  \sum_{t=0}^{P-1} \big( \big\| R_{\theta_t} \Lambda \Psi^T e_t
  - g_t \big\|_2^2 + \lambda \rho(\Lambda\Psi^T   {e_t}) \big)
  \nonumber \\ &
  + \xi \|\Psi\|_F^2 + \xi\|\Lambda\|_F^2 \quad   \text{where}
  \quad \Psi = UZ.
\end{align}

The first term is the data fidelity term measuring the fit between available undersampled measurements $g_t$ of the true object and the measurements obtained from the estimated object $f = \Lambda\Psi^T \in \mathbb{R}^{N^2 \times P}$. 
The second term with weight $\lambda > 0$ is a \emph{spatial} regularizer injecting relevant spatial prior to the problem. It is applied to the PSM $\Lambda\Psi^T$ column by column, that is, to individual temporal image frames. The last two terms with weight $\xi > 0$ prevent the bilinear factors from growing without limit.
\footnote{Inclusion of these Frobenius norm {terms} also happens to lead to connections with nuclear norm, but, as discussed in detail in Section~\ref{sec:conv_analysis}, our problem formulation (even if the RED regularizer is disregarded) is \emph{not} equivalent to a formulation encouraging a low-rank solution by a nuclear norm or other Schatten norm penalty such as used in \cite{lingala2011accelerated, majumdar2013non, majumdar2015learning}. Instead, the factorized form $f = \Lambda\Psi^T$ (as in other PSM-based matrix recovery works such as \cite{Haldar2010, zhao2010low}) introduces this rank constraint explicitly.}

Finally, the identity $\Psi = UZ$ is an implicit \emph{temporal} regularizer that restricts the temporal basis functions $\Psi$ to a $d$-dimensional subspace of $\mathbb{R}^{P}$ spanned by a fixed basis $U \in \mathbb{R}^{P \times d}$ 
where $d \geq K$. The action of $U$ on columns of $Z$ may be interpreted as interpolation from $d$ samples to the $P$-sample temporal basis functions. In practice, we incorporate this identity by explicit substitution (reparametrization of $\Psi$ in terms of the free variable $Z\in \mathbb{Z}^{d \times K}$) into the objective, and the minimization in \eqref{eq:variational_PSM_reg} over $\Psi$ is thus replaced by minimization over $Z$. 
This reduces the number of degrees of freedom in $\Psi$ to a fixed number $dK$, independent of the number $P$ of temporal sampling instants. 
For notational conciseness, we do not display this constraint/reparametrization in the sequel, but it is used throughout.

\subsection{Incorporating Regularization by Denoising}
\label{sec:red}
For the {spatial regularizer $\rho(\cdot)$} we consider ``Regularization by Denoising (RED)" \cite{romano2017little}.
RED proposes a recovery method using a {denoising operator $D: \mathbb{R}^{N^2} \rightarrow \mathbb{R}^{N^2}$} in an explicit regularizer of the form
\begin{align}
\label{eq:RedReg}
    \rho(f_t) = \frac{1}{2}f_t^T(f_t - D(f_t)).
\end{align}
Recent studies using this regularizer provide impressive performances for various \textit{static} image reconstruction tasks including high-dimensional cases \cite{sun2019block} {and phase retrieval \cite{metzler2018prdeep}}. {This regularizer was also combined with a DIP-based fidelity term in \cite{mataev2019deepred}. However, we avoid this due to disadvantages related to speed, overfitting, and convergence guarantees.}

While it provides significant flexibility for the type of denoisers that can be used, RED still requires $D$ to be differentiable and locally homogeneous, and to satisfy the passivity condition $\| D(f_t) \| \leq  \| f_t \|$, for the theoretical analysis of RED to apply.
\footnote{While many powerful denoisers have been demonstrated to satisfy these conditions in \cite{romano2017little}, 
recent work \cite{reehorst2018regularization} provides another framework to explain the good performance of RED with denoisers not satisfying them.}

{Next, we consider the optimization in \eqref{eq:variational_PSM_reg}.} For the conventional variational formulation
\begin{align}
  \hat{f}_t =  \arg\min_{f_t} 
  \| R_{\theta_t}f_t - g_t \|_2^2 + \lambda \rho(f_t), \nonumber
\end{align}
an efficient choice are iterative algorithms \cite{romano2017little} that use the standard ``early termination" approach \cite{boyd2011distributed}, and only require a single use of the denoiser per iteration. However, the regularized PSM objective in \eqref{eq:variational_PSM_reg}
does not allow to propagate the RED updates on $f$ to the respective basis functions efficiently. To overcome this difficulty, we perform a bilinear variable splitting $f = \Lambda\Psi^T$ and obtain our final formulation
\begin{align}
\label{eq:hard_cnst_RED_obj}
    \min_{f, \Lambda, \Psi} &H(f, \Lambda, \Psi) \quad \text{s.t.} \quad f = \Lambda \Psi^T \\
\label{eq:obj_fidelity_psm}
\text{where} \,\, H(f, \Lambda, \Psi) = &\sum_{t=0}^{P-1} \Big( \Big\| R_{\theta_t} \Lambda \Psi^T e_t - g_t \Big\|_2^2 + \lambda \rho(f e_t) \Big) \nonumber \\ &\quad  + \xi\|\Lambda\|_F^2 + \xi\|\Psi\|_F^2.
\end{align}
Since the PSM is enforced as a hard constraint, the estimated object $f$ is constrained to have $\text{rank}(f) \leq K$. Problem \eqref{eq:hard_cnst_RED_obj} is non-convex even if $\rho$ is convex, because of the product between unknowns $\Lambda$ and $\Psi$.

We propose an algorithm based on ADMM to solve \eqref{eq:hard_cnst_RED_obj}. %
{To this end,} we form the augmented Lagrangian %
{in the scaled form}, \cite{gabay1976dual, eckstein1992douglas}
\begin{align}
\label{eq:aug_lag_2}
\mathcal{L}_\beta[\Lambda, \Psi, f; \gamma] = &\sum_t \Big( \Big\| R_{\theta_t} \Lambda \Psi^T e_t - g_t \Big\|_2^2 + \lambda \rho(f e_t) \Big) \nonumber \\
+ \xi\|\Psi\|_F^2 + \xi\|\Lambda\|_F^2& - %
{\frac{\beta}{2}\| \gamma \|_F^2} + \frac{\beta}{2}\| \Lambda \Psi^T - f + %
{\gamma}\|_F^2,
\end{align}
where $\gamma \in \mathbb{R}^{P \times N^2}$ represents the scaled dual variable associated with the constraint $f=\Lambda \Psi^T$ and $\beta > 0$ is the penalty parameter.

Then, ADMM can be used to solve \eqref{eq:aug_lag_2} as in Algorithm \ref{alg:admm_psm}.

\begin{algorithm}[H]
\caption{RED-PSM}\label{alg:admm_psm}
\textbf{input:} $\Lambda^{(0)}$, $\Psi^{(0)}$, $\gamma^{(0)}$, $f^{(0)} = \Lambda^{(0)} \Psi^{(0)T}$, $\beta>0$, $\lambda>0$, $\xi>0$
\begin{algorithmic}[1]
\FOR{$i\in\{1,\ldots,I\}$}
    \STATE $\Lambda^{(i)} =
    \arg\min_{\Lambda} \{ \sum_t \| R_{\theta_t} \Lambda \Psi^{(i-1)T}e_t - g_t \|_2^2 $ \par $ + \frac{\beta}{2}\| \Lambda \Psi^{(i-1)T} - f^{(i-1)} + %
    {\gamma^{(i-1)}} \|_F^2 + \xi\|\Lambda\|_F^2\}$ \label{line:lambda_step_psm_fidelity_alg1}
    \smallskip
    
    \STATE $\Psi^{(i)} =
    \arg\min_{\Psi} \{ \sum_t \| R_{\theta_t}\Lambda^{(i)} \Psi^T e_t - g_t \|_2^2 $ \par $ + \frac{\beta}{2}\| \Lambda^{(i)} \Psi^T + %
    {\gamma^{(i-1)} - f^{(i-1)} } \|_F^2 + \xi\|\Psi\|_F^2 \}$ \label{line:psi_step_psm_fidelity_alg1}
    \smallskip

    \STATE %
    {$\forall t:\,\, f_t^{(i)} = \arg\min_{f_t} \{ \lambda \rho(f_t) $ \par $+ \frac{\beta}{2}\| (\Lambda^{(i)} \Psi^{(i)T} + \gamma^{(i-1)})e_t - f_t  \|_2^2 \}$} \label{line:f_step_psm_fidelity_alg1}
    \smallskip
    
    \STATE $\gamma^{(i)} = \gamma^{(i-1)} + %
    {\Lambda^{(i)} \Psi^{(i)T} - f^{(i)}}$ \label{line:dual_step_psm_fidelity_alg1}
\ENDFOR
\end{algorithmic}
\end{algorithm}

Line~\ref{line:f_step_psm_fidelity_alg1} in Algorithm~\ref{alg:admm_psm} then corresponds to the variational denoising for all $t$ of %
{the ``pseudo image frame"} %
$(\Lambda^{(i)} \Psi^{(i)T} + \gamma^{(i-1)})e_t$ with regularization $\lambda \rho(f_t)$. 
Instead of solving this denoising problem by running an iterative algorithm to convergence, we follow the RED approach \cite{romano2017little,reehorst2018regularization}, and for each $t$ replace the $f_t$ update in Step~\ref{line:f_step_psm_fidelity_alg2} by a single fixed-point iteration step using the approach of early stopping \cite{boyd2011distributed} (Sec. 4.3.2), taking advantage of the gradient rule
\begin{align}
\label{eq:red_grad_rule}
    \nabla \rho(f_t) = f_t - D_\phi(f_t)
\end{align}
where $D_\phi$ is the denoiser. 
This results in Algorithm~\ref{alg:admm_psm_red}, which requires only a single use of the denoiser per iteration of ADMM. Furthermore, as in Algorithm \ref{alg:admm_psm}, the modified Line 4 in Algorithm~\ref{alg:admm_psm_red} can be performed in parallel for all $t$.

\begin{algorithm}[htb]
\caption{RED-PSM with efficient $f$ step}\label{alg:admm_psm_red}
\small{\textbf{Notes:} Inputs, and Lines 1-3 and 5-6 are the same as Algorithm \ref{alg:admm_psm}. 
The $\smash{f}$ step is applied $\smash{\forall t}$.
}
\begin{algorithmic}[1]
\setcounter{ALC@line}{3}
    
    \STATE \normalsize{$\forall t:\,\, f_t^{(i)} =  \frac{\lambda}{\lambda + \beta} D_\phi(f_t^{(i-1)}) + %
    {\frac{\beta}{\lambda + \beta} \big( \Lambda^{(i)} \Psi^{(i)T} +  \gamma^{(i-1)} \big)e_t}$ \label{line:f_step_psm_fidelity_alg2}}
\vspace{-0.25cm}
\end{algorithmic}
\end{algorithm}
\vspace{-0.5cm}

\subsection{Regularization Denoiser}
\label{sec:reg_denoiser}
The regularization denoiser $D_\phi$ has a {deep neural network CNN (DnCNN)} \cite{zhang2017beyond} architecture and is trained in a supervised manner on a training set of 
{2D slices $f_i~\in~\mathbb{R}^{N^2}, i= 1, \ldots N$ of one or more \textit{static} objects similar to the object of interest}, assuming that such data will be available in the settings of interest. Thus, the RED steps are agnostic to the specific motion type. The training objective for the denoiser is
\begin{align}
\label{eq:red_denoiser_train}
\min_{\phi} \sum_i \| f_i - D_\phi (\tilde{f}_i) \|_F^2 \quad s.t. \quad \tilde{f}_i = f_i + \eta_i, \,\, \forall i,
\end{align}
where the injected noise $\eta_i \sim \mathcal{N}(0, \sigma_i^2 I)$ has noise level $\sigma_i \sim U[0, \sigma_{\max}]$ spanning a range of values, so that the denoiser learns to denoise data with various noise levels.

\subsection{Computational Cost}
\label{sec:complexity_analysis}

Space (memory requirements) and time (operation count) complexities for two variants of the RED-PSM bilinear ADMM algorithm are shown in Table~\ref{tab:space_time_comp}. The operation counts are for a single outer iteration. As the number of outer iterations typically has a weak dependence on the size of the problem, the scaling shown tends to determine the run time. Furthermore, thanks to its structure, the algorithm also offers many opportunities for easy parallelization, so that actual runtime can be proportionally reduced by allocating greater computational resources. See the Supplementary Material Section~\ref{sec:time_space_comp} for the detailed analysis of computational requirements.

The complexities for the proposed Algorithm~\ref{alg:admm_psm_red} are given in the first column {of Table~\ref{tab:space_time_comp}.} Space complexity is dominated by the %
{storage of $\Lambda$, and scales proportionally to image cross-section size in pixels $N^2$, and the order $K$ of the PSM. The time complexity is dominated by the  computations of the gradient with respect to $\Lambda$, and scales  proportionally to the size $N^2P$ of the spatio-temporal object,   PSM order $K$, and the number $M_i$ of inner iterations used to solve optimization subproblems  (Lines 2 and 3 in Algorithm~\ref{alg:admm_psm_red}).} 

Finally, in the second column of Table \ref{tab:space_time_comp}, we investigate the patch-based version of the proposed algorithm described in Section \ref{sec:patch-based_denoiser}. Given a patch size $N_B \ll N$ and stride $s \leq N_B$, this alternative %
{increases the operation count by $(\frac{N_B}{s})^2$ since it operates on overlapping patches. For example, for the settings in our experiments, $s=\frac{N_B}{2}$, increasing the operation count
by a factor of $4$. This is a modest price to pay, because this alternative also reduces the space complexity by a factor of $(\frac{N_B}{N})^2$. For example, for $N_B$=$8$ and $N$=$128$ as used in our experiments, this corresponds to a reduction by a factor of $144$ in space complexity. Thus, this variant of RED-PSM enables scaling for high-dimensional and high-resolution settings.

\begin{table}[hbtp!]
\footnotesize
\setlength{\tabcolsep}{2.5pt}
\renewcommand{\arraystretch}{1.05}
\centering
\begin{tabular}{@{}lccc@{}}
\toprule
\multicolumn{1}{c}{Complexity}&
\multicolumn{1}{c}{PSM-fidelity}&
\multicolumn{1}{c}{Patch-based}\\
\midrule
   Space & $O(KN^2)$ 
   & $O(KN_B^2)$ \\
   Time & $O(K N^2 P M_i)$ 
   & $O(K N^2 P M_i (\frac{N_B}{s})^2)$  \\
 \bottomrule
\end{tabular}
\setlength{\abovecaptionskip}{2pt}
\setlength{\belowcaptionskip}{-12.0pt}
\caption{\small Time and space complexities for two variants of the RED-PSM algorithm for a single outer iteration of the bilinear ADMM. $M_i$ is the number of inner iterations for each iteratively solved subproblem, and  $N_B \ll N$,  $s \leq N_B$.}
\label{tab:space_time_comp}
\end{table}

\section{Convergence Analysis}
\label{sec:conv_analysis}
In this section, we follow an approach %
{similar to recent work on ADMM for a bilinear model} \cite{hajinezhad2018alternating} to analyze convergence. We show that under mild technical conditions, the objective in Algorithm \ref{alg:admm_psm} 
is guaranteed to converge (with increasing number $I$ of iterations) to a value corresponding to a stationary point of the Lagrangian, that is, satisfying the necessary conditions for first order optimality. 

In practice, Algorithm \ref{alg:admm_psm_red} with the efficient $f$ step version, which we implemented and used in the experiments reported in Section \ref{sec:experiments}, has better run times, and rapid empirical convergence. However, its analysis requires additional steps, which are not particularly illuminating. Therefore we focus on the analysis of the nominal Algorithm \ref{alg:admm_psm}.
 
{In spite of the similarity, at a high level, to the problem and analysis in \cite{hajinezhad2018alternating}, our problem formulation and algorithm differ in several  aspects, which require modifications in the analysis and proof. In particular, different to \cite{hajinezhad2018alternating}, where the bilinear form only appears in the constraint, the bilinear form appears both in our objective function and constraint. {To satisfy strong convexity for the respective subproblems, our objective also uses the Frobenius norm terms for the PSM factors, instead of the interim proximal terms in \cite{hajinezhad2018alternating}. Moreover, to allow the efficient $f$-step in Algorithm \ref{alg:admm_psm_red}, we use a different order of updates than \cite{hajinezhad2018alternating}. Importantly, we account for the explicit RED regularization and its required properties in the proof of convergence.} Our analysis includes a proof (similar to \cite{wang2019global}) of the boundedness of the iterates to justify the existence of points of accumulation of the iterate sequence, which appears to have been inadvertently left out in \cite{hajinezhad2018alternating}. 
{However, different to \cite{wang2019global}, which has objectives with linear constraints,  our objective \eqref{eq:variational_PSM_reg} has a bilinear constraint.}}

To simplify the notation, we replace the {separate computation of the projections of {time $t$-frames in the data fidelity term} by using the operator $\bar{R}: \mathbb{R}^{N^2 \times P} \rightarrow \mathbb{R}^{N \times P}$ that computes the entire set of $P$ projections at view angles $\theta_t$ of the image frames at times $t$, $t=1, \ldots, P$ of dynamic image $f$, i.e, of each of its columns indexed by $t$, producing $g= \bar{R} f \in \mathbb{R}^{N \times P}$. When applied to the PSM, $\bar{R}$ performs {$R_{\theta_t} \Lambda \Psi^T e_t$} for each $t$. We also {aggregate the contribution} of the RED regularizer into $\bar{\rho}(f) \triangleq \sum_{t=0}^{P-1} \rho(f_t): \mathbb{R}^{N^2 \times P} \rightarrow \mathbb{R}$, and the denoiser into} $\bar{D}: \mathbb{R}^{N^2 \times P} 
\rightarrow \mathbb{R}^{N^2 \times P}$, which performs $D$ for each column of $f$ indexed by $t$.
}
Then, the Lagrangian function with dual variable $\gamma$ can be rewritten as 
\begin{align}
\label{eq:lag_1} 
\mathcal{L}[\Lambda, \Psi, &f; \gamma] = \| \bar{R}\Lambda\Psi^T - g \|_F^2 + \lambda %
{\bar{\rho}}(f) \nonumber \\
& + \xi\|\Lambda\|_F^2 + \xi\|\Psi\|_F^2 + %
{\beta}\langle \gamma, (\Lambda \Psi^T - f) \rangle ,
\end{align}
where the inner product is defined as $\langle A, B \rangle = \text{Tr}(A^T B)$.

The corresponding augmented Lagrangian is
\begin{align}
\label{eq:aug_lag_1}
&\mathcal{L}_\beta[\Lambda, \Psi, f; \gamma] = \| \bar{R}\Lambda\Psi^T - g \|_F^2 + \lambda %
{\bar{\rho}}(f) \\
& + \xi\|\Lambda\|_F^2 + \xi\|\Psi\|_F^2 + %
{\beta}\langle \gamma, (\Lambda \Psi^T - f) \rangle + \frac{\beta}{2}\| \Lambda \Psi^T - f \|_F^2. \nonumber
\end{align}

Then, we can state the subproblems with respect to each primal variable in Algorithm \ref{alg:admm_psm} as
\begin{align}
    S_\Lambda &= \| \bar{R}\Lambda\Psi^{(i-1)T} - g \|_F^2 + {\xi} \| \Lambda \|_F^2 \nonumber \\ &+ \frac{\beta}{2}\| \Lambda \Psi^{(i-1)T} - f^{(i-1)} + \gamma^{(i-1)} \|_F^2 \label{eq:SLambda} \\
    S_\Psi &= \| \bar{R}\Lambda^{(i)}\Psi^{T} - g \|_F^2 + {\xi} \| \Psi \|_F^2 \nonumber \\ &+ \frac{\beta}{2}\| \Lambda^{(i)}\Psi^{T} - f^{(i-1)} + \gamma^{(i-1)} \|_F^2, \text{ and} \label{eq:SPsi}\\
    S_f &= \lambda %
    {\bar{\rho}}(f) + \frac{\beta}{2}\| \Lambda^{(i)} \Psi^{(i)T} - f + %
    {\gamma^{(i-1)}} \|_F^2. \label{eq:Sf}
\end{align}

Algorithm \ref{alg:admm_psm} will be shown in Theorem \ref{thm:stat_pt_conv_thm} to converge to a stationary point of Problem \eqref{eq:hard_cnst_RED_obj}, which is defined below.

\textbf{Definition 1 (Stationary solution for Problem \eqref{eq:hard_cnst_RED_obj}).} The point $W^* = (\Lambda^*, \Psi^*, f^*, \gamma^*)$ is a stationary solution of the problem \eqref{eq:hard_cnst_RED_obj} if it satisfies the stationarity and primal feasibility conditions for the variables of the Lagrangian in \eqref{eq:lag_1}:
\begin{subequations}
\begin{align}
    \nabla_f\mathcal{L}(\Lambda^*, \Psi^*, f^*; \gamma^*) &= \lambda \nabla %
    {\bar{\rho}}(f^*) - %
    {\beta}\gamma^* = 0; 
    \label{eq:LagStationary_f} \\
    \nabla_\Lambda\mathcal{L}(\Lambda^*, \Psi^*, f^*; \gamma^*) &= (2\bar{R}^T(\bar{R}\Lambda^*\Psi^{*T} - g) + %
    {\beta}\gamma^*) \Psi^* \nonumber 
    \\ &\quad + %
    {2} \xi \Lambda^*= 0; 
   \label{eq:LagStationary_Lambda} \\
    \nabla_\Psi\mathcal{L}(\Lambda^*, \Psi^*, f^*; \gamma^*) &= (2\bar{R}^T(\bar{R}\Lambda^*\Psi^{*T} -g) + %
    {\beta}\gamma^*)^T\Lambda^{*} \nonumber \\ &\quad + %
    {2}\xi \Psi^* = 0; 
    \label{eq:LagStationary_Psi}\\
    f^* &= (\Lambda\Psi^T)^{(*)}.
    \label{eq:LagStationary_Feasible}
\end{align}
\end{subequations}

By its definition, a stationary solution of \eqref{eq:hard_cnst_RED_obj} satisfies the necessary first-order conditions for optimality.

To state Theorem \ref{thm:stat_pt_conv_thm}, we need to introduce some additional definitions pertaining to the denoiser $D$ used in RED.

\textbf{Definition 2 (Strong Passivity \cite{romano2017little}).} 
A function $D$ is strongly passive if $\|D(f)\| \leq \|f\|$.

Moreover, we assume the gradient rule of RED \eqref{eq:red_grad_rule}.
This is shown to hold true when the denoiser $D$ is assumed locally homogeneous (LH) \cite{romano2017little} and Jacobian symmetric (JS) \cite{reehorst2018regularization}. {It was shown \cite{romano2017little}} that several popular denoisers are practically LH, and that for \eqref{eq:red_grad_rule} to hold for the explicit regularizer $\rho$ in \eqref{eq:RedReg}, JS is necessary {\cite{reehorst2018regularization}}.

In our analysis, we only require the gradient rule \eqref{eq:red_grad_rule} to hold, rather than the explicit form \eqref{eq:RedReg} of the regularizer. Even for denoisers without LH and JS, the gradient rule can be used to define the fixed point iteration in the algorithm \cite{reehorst2018regularization}, and, 
more recent works on using denoisers for regularization 
\cite{cohen_it_2021,hurault_gradient_2021, PnP-Reg-Fermanian2023, berk_deep_2023, chand_multi-scale_2023} 
(in particular \cite{PnP-Reg-Fermanian2023}) shows how to learn a denoiser jointly with the gradient so that the gradient rule is satisfied.

Finally, we will also assume Lipschitz continuity of the denoiser with some Lipschitz constant $L_D$, i,e., $\| D_\phi(f_1) - D_\phi(f_2) \|_2 \leq L_D \| f_1 - f_2 \|_2 $ for all images $f_1, f_2 \in \mathbb{R}^{N^2}$. This means that the denoiser has some finite gain $L_D$, which is satisfied by any reasonable denoiser.
\smallskip

Our main convergence result is the following.
\begin{theorem}
\label{thm:stat_pt_conv_thm}
    Suppose that the denoiser $D_\phi$ {satisfies the gradient rule in \eqref{eq:red_grad_rule},} {and is} Lipschitz continuous with some Lipschitz constant $L_D$, and strongly passive.

    Then, if $\beta > 2L$ where $L \triangleq\lambda(1+L_D)$, Algorithm \ref{alg:admm_psm} converges globally (i.e., regardless of initialization) to a stationary solution $(f^*, \Lambda^*, \Psi^*)$ of \eqref{eq:hard_cnst_RED_obj}
    in the following sense:\\
    (i) The sequence of values of the objective $H$ converges to a limit $H^*=H(f^*, \Lambda^*, \Psi^*)$;
    and \\
    (ii) The iterates generated by Algorithm \ref{alg:admm_psm} converge subsequentially to a stationary solution $(f^*, \Lambda^*, \Psi^*, \gamma^*)$ of \eqref{eq:hard_cnst_RED_obj}, that is, any accumulation point (of which there is at least one) of the sequence $(f^{(i)}, \Lambda^{(i)}, \Psi^{(i)},  \gamma^{(i)})$ is a stationary solution of \eqref{eq:hard_cnst_RED_obj}.
\end{theorem}

The proof of Theorem \ref{thm:stat_pt_conv_thm} is provided in the supplementary material Section 
\hyperref[appendix:proof_conv_thm]{D}.
\smallskip

Regarding convergence to the globally optimal solution, more can be said. Problem \eqref{eq:variational_PSM_reg} is non-convex even if $\bar{\rho}$ is convex, because of the product between unknowns $\Lambda$ and $\Psi$. However, similar to {Theorem 7.1 of} \cite{udell2016generalized} {(which  generalizes arguments in \cite{rennie2005fast},\cite{recht2010guaranteed})}, thanks to the inclusion of the Frobenius norms of the factors $\Lambda$ and $\Psi$ in the bilinear form, {when $\bar{\rho}$ is convex (i.e., {when the gradient rule \eqref{eq:red_grad_rule} is assumed to hold}),}
the global minimum $\hat{f}=\hat{\Lambda}\hat{\Psi}^T$ of \eqref{eq:variational_PSM_reg} can be shown to coincide with the global minimum in
\begin{align}
    \hat{f} = \arg\min_{f} \sum_{t=0}^{P-1} &\| R_{\theta_t}f_t - g_t \|_2^2 + \lambda \bar{\rho}(f) + 2\xi \| f \|_* \nonumber \\ 
    &\text{s.t. } \mathrm{rank}(f) \leq K,
    \label{eq:f_nuc_norm_obj}
\end{align}
where the penalty weight $\xi>0$ is identical to the Frobenius norm penalty in our RED-PSM objective \eqref{eq:hard_cnst_RED_obj}.

{Importantly,} Problem \eqref{eq:f_nuc_norm_obj} is non-convex too, because of the rank constraint. {{In other words, it is \emph{not} equivalent to a conventional nuclear norm penalized problem. Hence, although the Frobenius norm penalties on the PSM factors in our problem formulation \eqref{eq:variational_PSM_reg} lead to a connection to nuclear norm, Problem \eqref{eq:variational_PSM_reg} is \emph{not equivalent} to a formulation encouraging a low-rank solution by a nuclear norm penalty.}}

However, returning to the question of determining global optimality of a candidate solution to \eqref{eq:f_nuc_norm_obj}, consider the convex version of Problem \eqref{eq:f_nuc_norm_obj}, without the rank constraint, and denote its solution by $\hat{f}^{\mathrm{Convex}}$. {If it so happens} that $\mathrm{rank}(\hat{f}^{\mathrm{Convex}}) \leq K$, then the low-rank constraint in Problem \eqref{eq:f_nuc_norm_obj} is not active, and $\hat{f}^{\mathrm{Convex}} = \hat{f}$. It then follows
\cite{udell2016generalized} that the optimality conditions of the convex problem {(without the rank constraint)} can be used to assess the global optimality of {{a candidate solution}} $\hat{f}~=~\hat{\Lambda}\hat{\Psi}^T$ produced by the algorithm solving \eqref{eq:variational_PSM_reg}.

\vspace*{-1ex}
\section{Experiments}
\label{sec:experiments}
\subsection{Datasets}
{Three categories of data sets are used in this work.}

    \textit{Walnut Dataset:}
    We use the CT reconstructions of two different (static) walnut objects from the publicly available 3D walnut CT dataset \cite{der2019cone}. We create a dynamic test object by synthetically warping the central axial slice of one of the walnut objects using a sinusoidal piecewise-affine time-varying warp \cite{pwise_affine}. To be precise, the image is divided into a $N \times N$ uniformly spaced rectangular grid, and the following vertical displacement is applied on each row separately to drive the temporally varying warp $  \Delta_{n,t} = -C(t) \sin(3 \pi n/N), \,\, n \in \{0, \ldots, N-1\},$ where $C(t)$ is a linearly increasing function of $t$ and $C(0) = 0$.
   \textit{Static} axial, coronal, and sagittal slices of the other walnut object are used to train the denoiser $D_\phi$.

   \textit{Compressed Object Dataset:} The compressed object data set is obtained from a materials science experiment \cite{personal_comm} with a sequence of nine increasing compression (loading) steps applied to an object, with a full set of radiographic projections collected (using Carl Zeiss Xradia 520 Versa) and reconstructed by the instrument's software at each step. Using this quasi-static data set, a fixed axial slice is extracted from each {of the 9} reconstructions. %
   {These nine extracted slices corresponding to nine time points} are interpolated to $P$ time frames using a recent deep learning-based video interpolation algorithm \cite{reda2022film}. The denoiser {$D_\phi$} for our experiments on this data set was trained using the axial slices of the \textit{static} pre-compression and post-compression versions of the object, which would be available in actual dynamic loading experiments of this type.

In a conference submission \cite{iskender2023dynamic} citing this work, the method has been applied to additional tomographic scenarios of materials compressed under load, further confirming the advantage of the proposed method and algorithm over the compared methods.

We note that all algorithms compared in this paper are agnostic to both the synthetic warp applied on the static walnut slice and to the data-driven interpolation method used for the compressed object. 

Spatio-temporal projection data for each dataset is simulated by a 
parallel-beam projection with $N$=128 detector bins of the dynamic phantoms, a single projection at each of the $P$ time instants. The sequence of projection angles $\{\theta_t\}_{t=0}^{P-1}$ (a free experimental design parameter) was chosen to follow the bit-reversed view angle sampling scheme, which has been shown \cite{iskender2022dynamic} to provide an especially favorable conditioning of the recovery problem.
The simulated measurements are corrupted using AWGN with standard deviation $\sigma$ = $5\cdot10^{-3}$. This noise level leads to the FBP (with Ram-Lak filter) of the full set of \textit{$P$=512 projections at each time instant} having a PSNR of approximately 46 dB. When, in the actual experiments with sequentially sampled data, only $1/P$ of this data is used, the PSNR of the reconstruction may be expected to be lower.

Ground-truth frames for {$P=4$} are shown in Figure \ref{fig:ground_truth_frames}.

{\textit{Cardiac dMRI Dataset:} 
{For a more direct comparison with the setting and data used in dMRI works,} we also test RED-PSM on the ``retrospective" cardiac dMRI data in \cite{yoo2021time}. The data includes 23 distinct time frames for one cardiac cycle. {Details of the data and experiments are} in Section \ref{sec:cardiac_mri_exp}.}
\vspace{-0.25cm}

\begin{figure}[hbtp!]
    \centering
    \includegraphics[width=\linewidth]{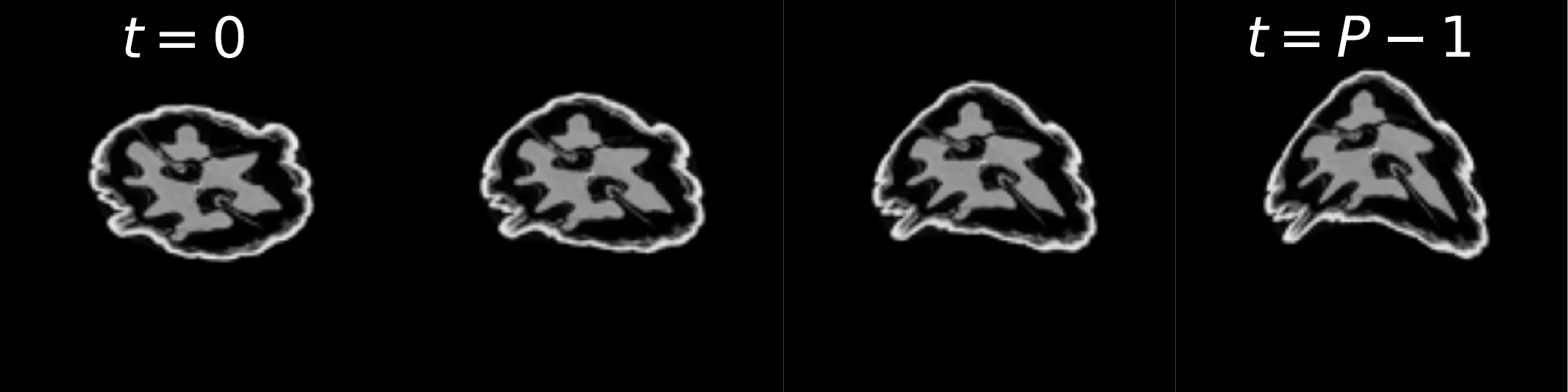} \\
    \vspace{-0.65cm}
    \includegraphics[width=\linewidth, trim={0 0.75cm 0 0}, clip]{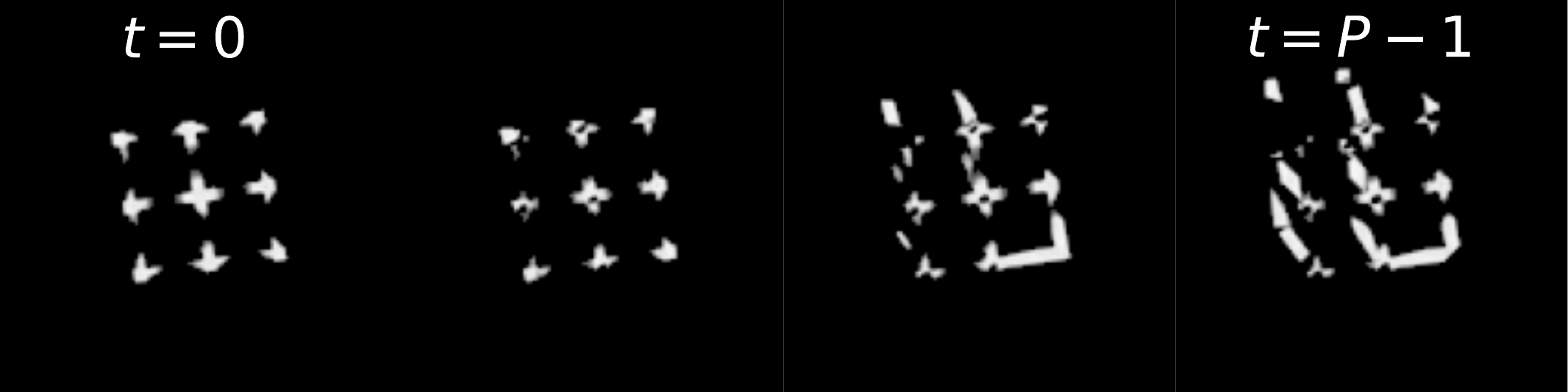} \\
    \setlength{\abovecaptionskip}{2pt}
    \setlength{\belowcaptionskip}{-16pt}
    \caption{\small Ground-truth frames %
    {uniformly sampled in time for $P=4$,} for the time-varying walnut (top) and compressed object (bottom). }
    \label{fig:ground_truth_frames}
\end{figure}

\subsection{{{Comparison Benchmarks}}}
\textbf{PSM-TV:} 
Similar to the proposed approach, this algorithm also uses a PSM to represent the object, but instead of the RED regularizer, the regularization penalizes the discrete 2D total variation of the temporal frames of $f$ at each time instant. To this end, the constraint $f=\Lambda\Psi^T$ is implemented by substitution into the objective in \eqref{eq:hard_cnst_RED_obj}, and the definition of $\rho$ is changed to 
{{$\rho(\cdot) = \text{TV}(\cdot)$},}
and the rest of the objective is kept the same. 
{{We consider spatial (PSM-TV-S) and spatiotemporal (PSM-TV-ST) alternatives of TV. The former computes TV only spatially in a single frame at each $t$ whereas the latter also computes differences between temporally adjacent pixels at $t-1$ and $t+1$.}}
The unconstrained problem is then solved for $\{\hat{\Lambda}, \hat{\Psi}\}$ (using Adam optimizer in Pytorch). Finally, the estimated object is obtained as $\hat{f} = \hat{\Lambda}\hat{\Psi}^T$.

\textbf{TD-DIP \cite{yoo2021time}:}
TD-DIP is a recent method based on the Deep Image Prior (DIP) approach.
\footnote{It would be interesting to include yet another {comparison benchmark} 
(also developed for dMRI) -- the DIP-based PSM approach \cite{ahmed2022dynamic}. However, as an implementation of this method was unavailable, and due to potential issues with replicating its performance and adapting to our CT problem {(specific initialization and framework that use other MRI algorithms)}, we were unable to do so.}
It uses a mapping network $M_\alpha$ and a generative model $N_\beta$ {in cascade} to obtain the estimated object at each time instant $f_t$ from fixed and handcrafted latent representations $\chi_t$. Because TD-DIP was originally proposed for dynamic MRI, we modified the objective minimally for dynamic tomography as
\begin{align}
    \min_{\alpha, \beta} \sum_t \big\| g_t - R_{\theta(t)}((N_\beta \circ M_\alpha) (\chi_t))  \big\|^2.
\end{align}
For the comparisons in this work, {the} mapping network and generator architectures, latent representation dimensionality, optimizer, learning rate, and decay schemes are {are identical to those} in the available implementation \cite{TDDIPGithub2021}. The original work focuses on the beating heart problem and thus proposes a helix-shaped manifold for $\chi_t$ with number of cycles equal to the number of heartbeats during measurement acquisition. Since we do not have a repetition assumption for the motions of the walnut and compressed object, we use a linear manifold as explained in the original paper \cite{yoo2021time}. Thus, {for clarity, in Section \ref{sec:results} the method is sometimes denoted as ``TD-DIP (L)".}
\vspace{-0.25cm}
\subsection{Experimental Settings}
\label{sec:framework}
All methods are run on a workstation with an Intel(R) Xeon(R) Gold 5320 CPU and NVIDIA RTX A6000 GPU. 
In practice, we used a minor variation of Algorithm \ref{alg:admm_psm_red}, where we combined the subproblems for $\Lambda$ and $\Psi$, and minimized with respect to both basis functions simultaneously using gradient descent with Adam \cite{kingma2014adam} optimizer.

\textbf{Denoiser {and} training.} 
{Each convolutional layer in the denoiser network is followed by a ReLU nonlinearity except for the final single-channel output layer. 
We tested both direct and residual DnCNN denoisers, where the former predicts the denoised image and the latter estimates the noise from the input.} We use the denoiser type that performs better for each object, {but} the differences are minor. {Further architectural details for the denoisers in our experiments are in Table \ref{tab:denoiser_arch} in the Supplementary Material. We use three pre-trained denoisers, one for each of the three object types. In each case, the same pre-trained denoiser was used  for all values of $P$.}

The upper limit for noise level used in training the denoiser was set to $\sigma_{\max} = 5 \cdot 10^{-2}$. For the dynamic walnut object, the denoiser $D_\phi$ is trained on the central 200 axial, {200} sagittal, and {200} coronal slices of another {static} walnut CT reconstruction downsampled to size 128$\times$128. For the compressed object, axial slices of pre-compression and post-compression \textit{static} {versions of the object}, containing 462 slices in total, are used to train $D_\phi$. {For the cardiac MRI setting, the denoiser was trained on the \textit{static} MRI training slices of the ACDC dataset \cite{bernard2018deep}.} For all datasets, $D_\phi$ is trained for 500 epochs using the Adam optimizer with a learning rate of $5\cdot10^{-3}$. {As mentioned above, we evaluate both direct and residual DnCNN denoisers.}

\textbf{Temporal Basis.} In {{compressed material and cardiac dMRI data experiments}}, we use the parametrization $\Psi=UZ$ with a fixed basis $U$ that corresponds to a cubic spline interpolator, {{and for the warped walnut we use DCT-II}}, to interpolate the low-dimensional temporal representation $Z$ to $\Psi$.

\textbf{Initialization.} Unless stated otherwise, the spatial and temporal basis functions are initialized using the SVD truncated to {the rank of} the dynamic object estimate produced by a recent projection-domain PSM-based method ``ProSep" \cite{iskender2022dynamic_long}. {If the ProSep estimate has rank smaller than $K$, the remaining basis functions are initialized as 0.} Otherwise, all spatial basis functions are initialized as 0 and the temporal latent representations $z_k$ are initialized randomly as $z_k \sim \mathcal{N}(0, I)$.

\textbf{Tomographic Acquisition Scheme.} All methods mentioned in this paper use the bit-reversed angular sampling scheme, over the range $[0, \pi]$. {For time-sequential acquisition, the bit-reversed scheme was shown \cite{iskender2022dynamic, iskender2022dynamic_long} to provide favorable results via better conditioning of the forward model in comparison to alternatives. In a standard CT scanner, the speed of rotation might have to be significantly increased to implement this scheme, possibly leading to greater motion blur. However, for several scenarios this would not be much of an issue. These include radial acquisition in MRI; a CT scanner with electronic beam deflection \cite{kulkarni_electron_2021}; and settings where the acquisition time is dominated by the time to acquire each view rather than the rotation time, e.g., micro-CT, or imaging of dense materials.}

{To help address the challenge in implementing the bit-reversed scheme in other, more physically constrained settings, the number of distinct view angles $\hat{P}$ can be reduced and these views can be repeated periodically without a performance drop as also shown in Figure \ref{fig:walnut_metrics_PSNR_vs_period}. With a reduced $\hat{P}$, the bit-reversed scheme can be implemented more conveniently, e.g., by multiple source-detector pairs, or by carbon nanotube sources \cite{spronk_evaluation_2021, xu_volumetric_2023}.}

\textbf{Run Times.} For $P=256$ and using the specified computational resources and parameter settings, {to reach the peak PSNR during optimization,} RED-PSM with ProSep initialization requires 50 $<$ iterations $<$ 150 taking about 2 to 6 minutes whereas TD-DIP with batch size $P$ typically requires $>$ 30k steps, taking about 3.5 hours to complete. Hence, RED-PSM provides a speedup over TD-DIP by a factor of 35 to 105. {Depending} on the parameter configuration, the speedup factor may vary. However, the proposed method provides a significant run time reduction in all cases.

{\textbf{Evaluation Metrics.}} Four quantitative metrics were implemented for comparing different method performances: (i) the \textit{peak signal-to-noise ratio} (PSNR) in dB; (ii) the \textit{structural similarity index} (SSIM) \cite{wang2004image}; (iii) the \textit{mean absolute error} (MAE); and (iv) the \textit{high-frequency error norm} (HFEN) \cite{ravishankar2010mr}. The latter is defined as $\text{HFEN}(f, f_{r})~=~\| \text{LoG}(f) - \text{LoG}(f_r) \|_2 \nonumber$ where $\text{LoG}$ is a rotationally symmetric Laplacian of Gaussian filter with a standard deviation of $1.5$ pixels.
\vspace{-0.25cm}

\subsection{Results}
\label{sec:results}

\subsubsection{Reconstruction accuracies for different $P$}
\label{sec:comparison_wrt_P}
    \begin{table*}[hbtp!]
    \small
    \setlength{\tabcolsep}{-2.0pt}
    \renewcommand{\arraystretch}{0.55}
    \centering
    \begin{tabular}{ccccc}
    \shifttext{-0.3cm}{\raisebox{1.65cm}{\rotatebox[origin=c]{90}{\underline{Walnut}}}} & \includegraphics[width=0.25\linewidth]{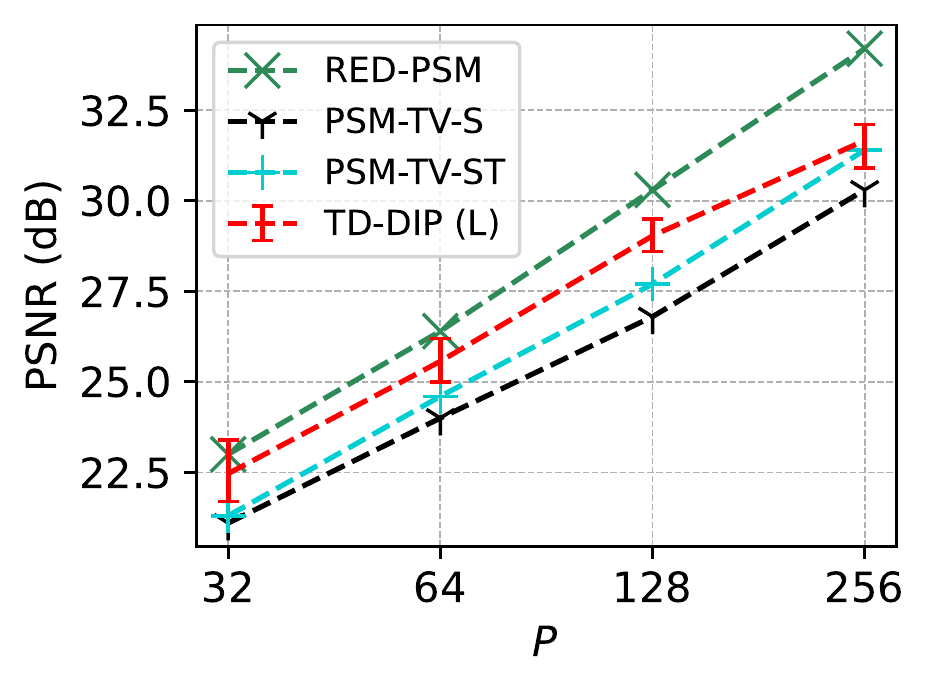} &
    \includegraphics[width=0.25\linewidth]{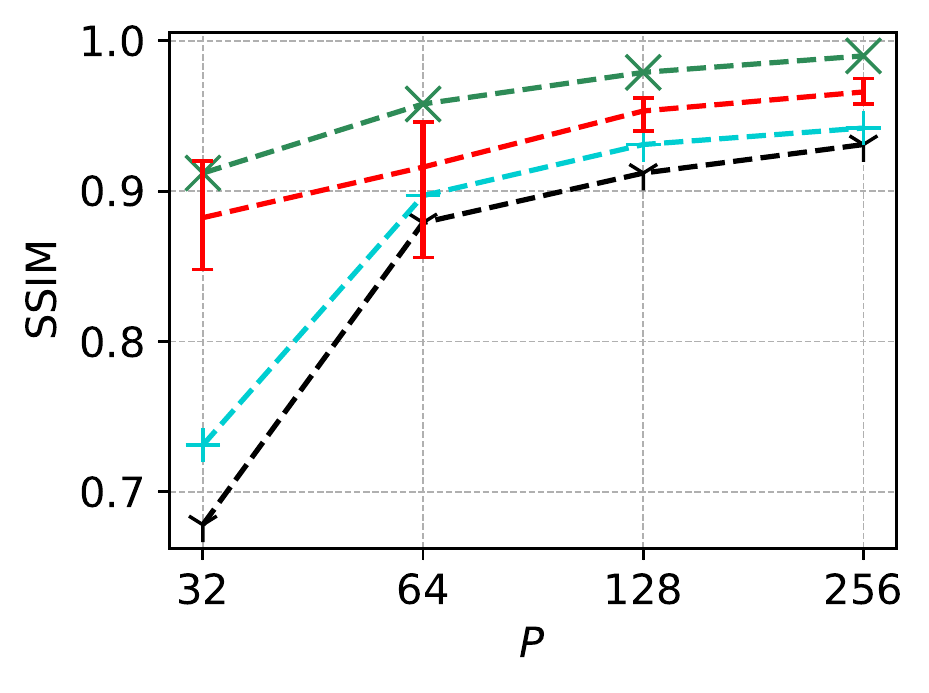} &
    \includegraphics[width=0.25\linewidth]{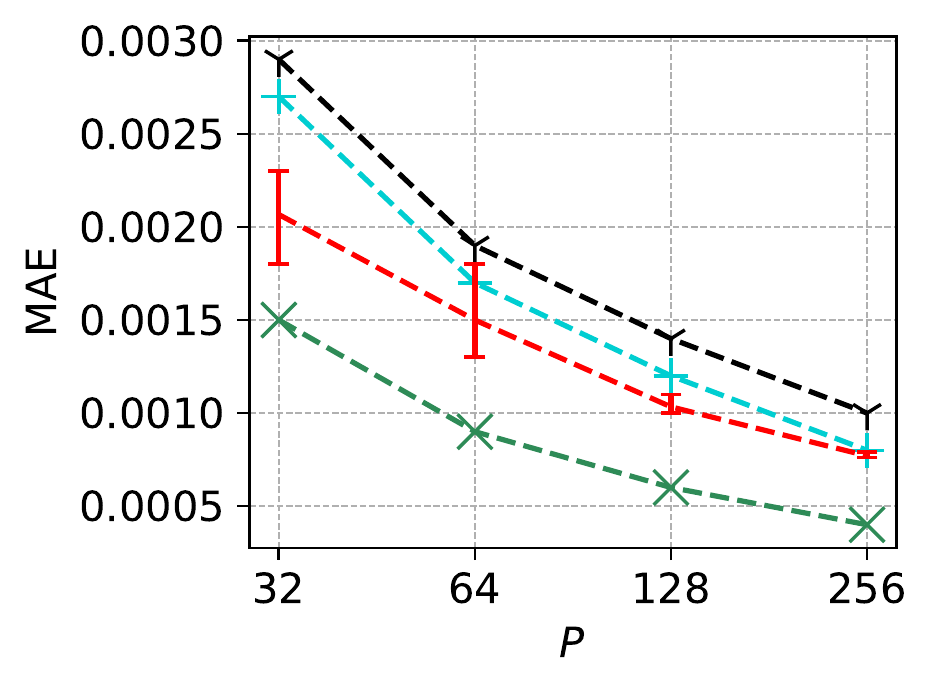} &
    \includegraphics[width=0.25\linewidth]{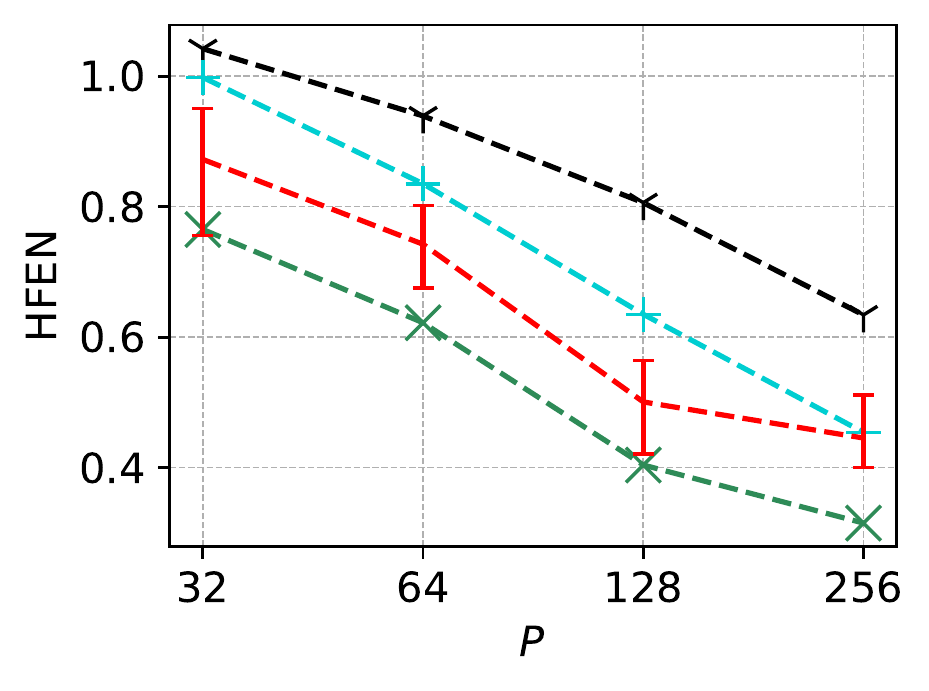} \vspace{-0.48cm} \\
    \shifttext{-0.3cm}{\raisebox{1.65cm}{\rotatebox[origin=c]{90}{\underline{Comp. Material}}}} & \includegraphics[width=0.25\linewidth]{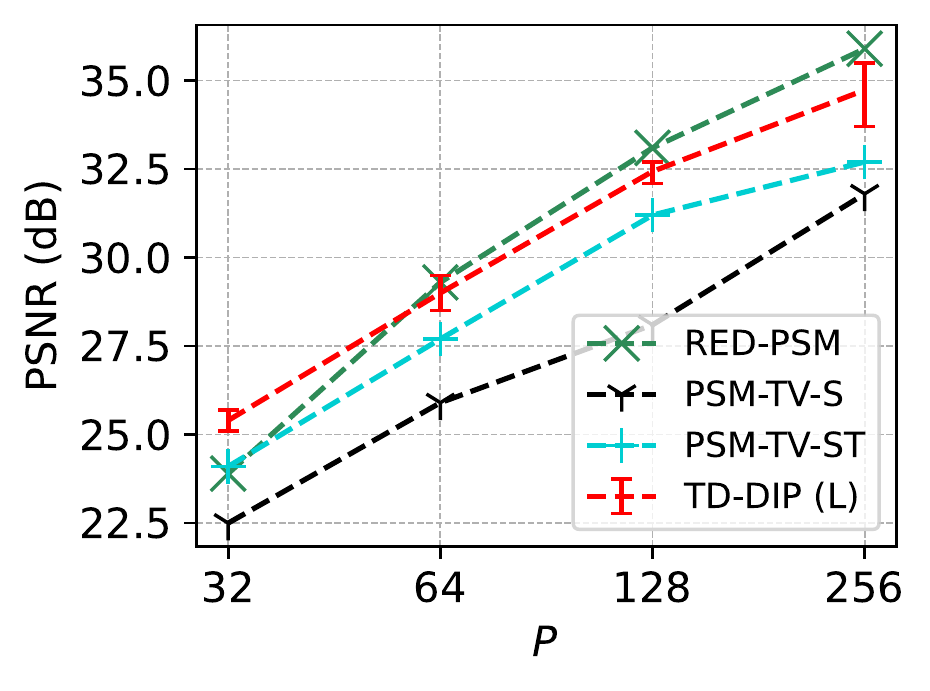} &
    \includegraphics[width=0.25\linewidth]{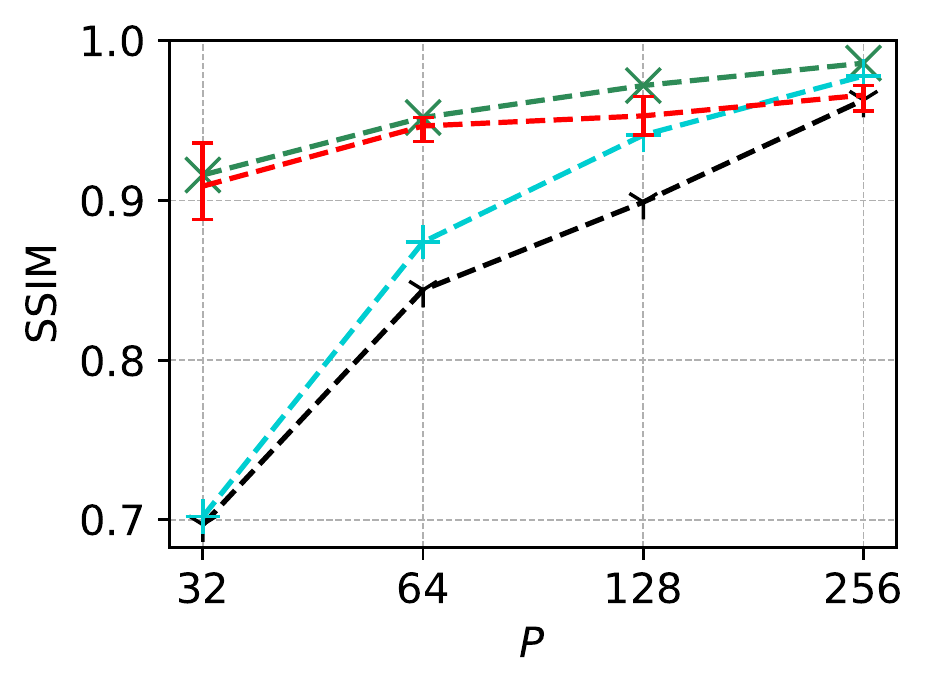} &
    \includegraphics[width=0.25\linewidth]{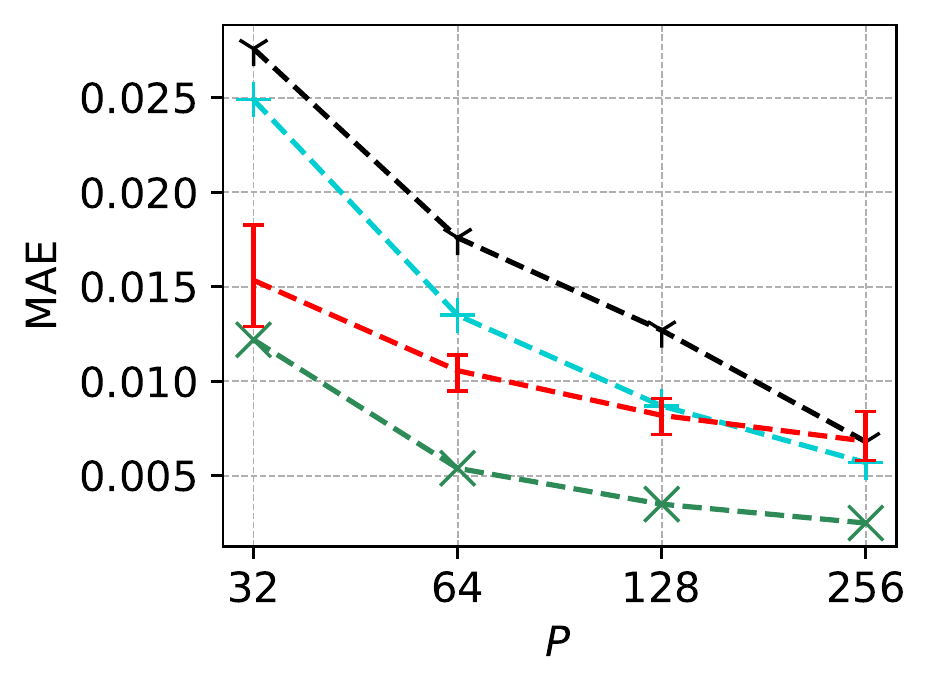} &
    \includegraphics[width=0.25\linewidth]{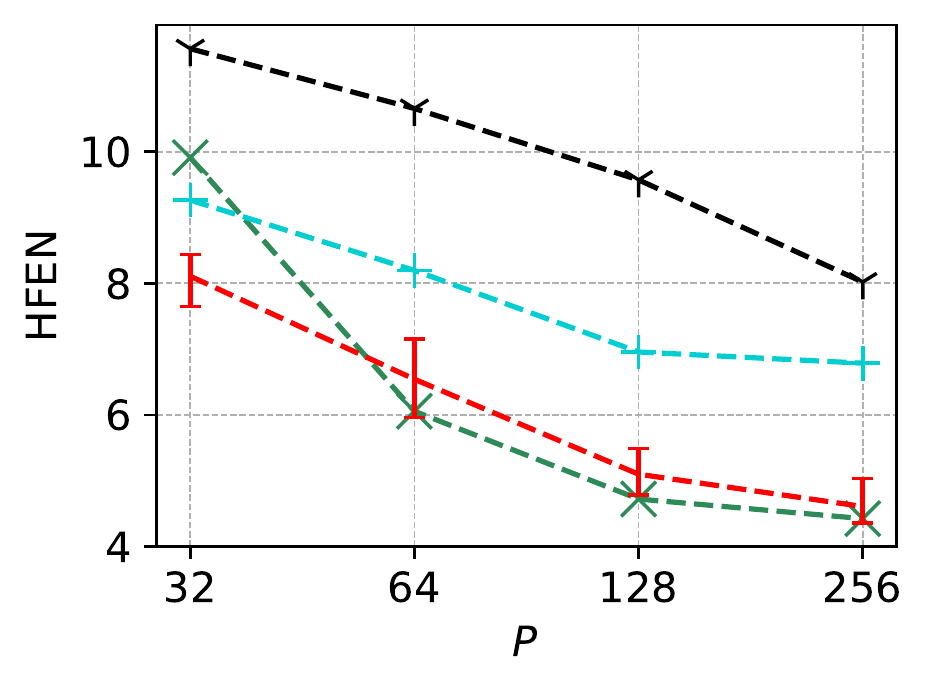}
    \end{tabular}
    \setlength{\abovecaptionskip}{-5pt}
    \setlength{\belowcaptionskip}{-14pt}
    \captionof{figure}{\small{Reconstruction metrics for the time-varying walnut and compressed material vs. $P$ using different methods. For TD-DIP, the metrics reported are assuming a ``stopping oracle" that stops the iterations at the best PSNR reconstruction.
    The minimum and maximum values over three different runs with different random initial values are shown with bars, and the mean values are connected by dashed lines}}
    \label{fig:walnut_metrics_errbar_vs_P}
    \end{table*}
RED-PSM is compared in Figure \ref{fig:walnut_metrics_errbar_vs_P} and Table \ref{tab:avg_rec_acc_comp} with {{PSM-TV-S, PSM-TV-ST}} and TD-DIP (using, for both a-periodic objects, a linear latent representation). Parameter configurations for the experiments are provided in Table \ref{tab:avg_rec_acc_comp_params} in Supplementary Material 
\hyperref[supp:exp_config]{A}.
While Figure \ref{fig:walnut_metrics_errbar_vs_P} facilitates the assessment and comparison of trends in the metrics for different $P$ {for the various} methods, Table \ref{tab:avg_rec_acc_comp} provides the detail for a more precise quantitative comparison. In fact, since the plotted range of metrics such as PSNR in Figure \ref{fig:walnut_metrics_errbar_vs_P} is very large due to varying $P$, Table \ref{tab:avg_rec_acc_comp} emphasizes important differences between methods for the same $P$.
\bigskip
\bigskip
\bigskip

{\textbf{Remarks:}

1) The scale differences in the MAE and HFEN between the two objects are due to working with un-normalized densities.

2) In all TD-DIP experiments, optimization is stopped early to achieve the best PSNR reconstructions, assuming such a ``stopping oracle" is available. In a more realistic setting, in the absence of this oracle, TD-DIP optimization with continued iterations suffers from overfitting, and produces in these experiments degraded results. {{For instance, for the warped walnut slice with 40k iterations, PSNR, SSIM, and MAE degrade significantly to 19.5 dB, 0.822, and $3.1 \cdot 10^{-3}$ for $P=32$, and 24.1 dB, 0.824, and $2.1 \cdot 10^{-3}$ for $P=64$.} {Thanks to the accurate spatial prior and the convergence properties, we do not encounter such a problem for RED-PSM.}}

3) {Because the performance of TD-DIP varies with initialization, in each experiment we ran it three times, with different random initialization, each time using a ``stopping oracle" to obtain the best PSNR for the given initialization.} Table \ref{tab:avg_rec_acc_comp} reports the 
average of \textit{the best PSNR reconstruction accuracies} for TD-DIP in  these three runs. Figure \ref{fig:walnut_metrics_errbar_vs_P}, {complements this information, by showing} in addition to the average results, also the best and worst of these runs {(still using the stopping oracle to get the best PSNR per run).}

As expected, {for all methods,} the estimates improve with increasing $P$. In terms of PSNR, the proposed algorithm performs on par with or {usually} better than TD-DIP {(for all but the lowest $P$ on the compressed material object)}, and consistently better than {{PSM-TV-S and PSM-TV-ST}}. 
The PSNR improvement of RED-PSM over TD-DIP enhances with increasing $P$, {reaching 2.4 dB for the time-varying walnut with $P=256$}. Moreover, for the other three metrics, SSIM, MAE, and HFEN, the improvement of RED-PSM {{over other algorithms}} is more significant. {{Specifically, the reduction in MAE reaches up to and exceeds $\%50$ for both objects for $P=256$.}} These observations are valid for both objects, however, RED-PSM provides slightly greater improvement over TD-DIP in the walnut case.

{{We would like to emphasize the significance of these reconstruction quality improvements by comparing them to some representative examples. TD-DIP reports up to 3 dB PSNR and 0.005 SSIM improvement with respect to an \textit{older} dMRI method in a \textit{single} scenario. While providing comparable PSNR improvements (2.4 dB), we are able to provide four times larger SSIM (0.02) improvements \textit{relative to TD-DIP itself}. Other recent dMRI \cite{zou2021dynamic} and dynamic photo-acoustic tomography \cite{lozenski2022memory} methods improve 1.5 dB, and at most 1 dB and often no improvement over their respective benchmarks, again in single imaging scenario for each method.}

{Also, based on the run-time improvement ($\sim$3x) of \cite{zou2021dynamic} with respect to TD-DIP, RED-PSM is still much ($\sim$10x-30x) faster, than \cite{zou2021dynamic}, and more interpretable.}}

    \begin{table}[hbtp!]
    \footnotesize
    \setlength{\tabcolsep}{1pt}
    \renewcommand{\arraystretch}{0.45}
    \centering
    \begin{tabular}{@{}cl|cccc|cccc@{}}
    \toprule
     \multicolumn{2}{c}{} & \multicolumn{4}{c}{(a) Walnut} & \multicolumn{4}{c}{(b) Compressed Material} \\
     \midrule
    \multicolumn{1}{c}{$P$} & \multicolumn{1}{c}{Method} 
    &  \multicolumn{1}{c}{\makecell{PSNR\\(dB)}} & \multicolumn{1}{c}{SSIM} & \multicolumn{1}{c}{\makecell{MAE\\(1e-3)}} & \multicolumn{1}{c}{HFEN} 
    & \multicolumn{1}{c}{\makecell{PSNR\\(dB)}} & \multicolumn{1}{c}{SSIM} & \multicolumn{1}{c}{\makecell{MAE\\(1e-2)}} & \multicolumn{1}{c}{HFEN} \\
    \midrule
      32 & PSM-TV-S (R) 
      & 21.1 & 0.678 & 2.9 & 1.04 
      & 22.5 & 0.697 & 2.9 & 11.6 \\
      & {PSM-TV-ST (R)}
      & 21.3 & 0.731 & 2.6 &  0.99
      & 24.1 & 0.702 & 2.5 & 9.3 \\
       & TD-DIP (L) 
      & 22.5 & 0.882 & 2.1 & 0.87 
      & \textbf{25.4} & 0.909 & 1.5 & \textbf{8.1} \\
       & RED-PSM (Pr) 
      & \textbf{22.8} & \textbf{0.911} & \textbf{1.5} & \textbf{0.78}
      & 23.9 & \textbf{0.916} & \textbf{1.2} & 9.9 \\
     \midrule
     64 & PSM-TV-S (R) 
     & 24.0 & 0.879 & 1.9 & 0.94 
     & 25.6 & 0.845 & 1.8 & 10.7 \\
     & {PSM-TV-ST (R)}
      & 24.6 & 0.897 & 1.7 & 0.83
      & 27.7 & 0.874 & 1.3 & 8.2 \\
       & TD-DIP (L) 
      & 25.6 & 0.916 & 1.5 & 0.74 
      & 29.0 & 0.947 & 1.1 & 6.5 \\
       & RED-PSM (Pr) 
      & \textbf{26.4} & \textbf{0.958} & \textbf{0.9} & \textbf{0.57} 
      & \textbf{29.3} & \textbf{0.952} & \textbf{0.5} & \textbf{6.0} \\
     \midrule
      128 & PSM-TV-S (R) 
      & 26.8 & 0.912 & 1.4 & 0.78 
      & 29.2 & 0.907 & 1.1 & 9.6 \\
      & {PSM-TV-ST (R)}
      & 27.7 & 0.931 & 1.2 & 0.63 
      & 31.2 & 0.942 & 0.9 & 7.0 \\
       & TD-DIP (L) 
      & 29.0 & 0.953 & 1.0 & 0.50 
      & 32.4 & 0.953 & 0.8 & 5.1 \\
       & RED-PSM (Pr) 
      & \textbf{30.3} & \textbf{0.979} & \textbf{0.6} & \textbf{0.40} 
      & \textbf{33.1} & \textbf{0.972} & \textbf{0.4} & \textbf{4.7} \\
     \midrule
     256 & PSM-TV-S (R) 
     & 30.1 & 0.934 & 1.0 & 0.63 
     & 31.9 & 0.963 & 0.7 & 8.0 \\
     & {PSM-TV-ST (R)}
      & 31.4 & 0.942 & 0.8 &  0.45
      & 32.7 & 0.978 & 0.6 & 6.8 \\
      & TD-DIP (L) 
     & 31.7 & 0.966 & 0.8 & 0.45 
     & 34.7 & 0.966 & 0.7 & 4.6 \\
      & RED-PSM (Pr) 
     & \textbf{34.2} & \textbf{0.989}& \textbf{0.4} & \textbf{0.32} 
     & \textbf{35.9} & \textbf{0.986} & \textbf{0.3} & \textbf{4.4} \\
     \bottomrule
    \end{tabular}
    \setlength{\abovecaptionskip}{2pt}
    \setlength{\belowcaptionskip}{-8pt}
    \caption{\small Reconstruction accuracies for for different $P$.
    Random (R) and ProSep (Pr) initialization for the PSM methods. 
    For TD-DIP, the reported accuracies are for the best PSNR using a ``stopping oracle", averaged over three runs with random initial conditions.}
    \label{tab:avg_rec_acc_comp}
    \end{table}

Figure \ref{fig:recon_est} compares the reconstructions for both objects at two different {values of} $t$ for $P=256$. As expected, {PSM-TV-S} performs the worst among the compared methods and provides, for both objects, blurry reconstructions lacking finer details. The TD-DIP reconstructions improve somewhat over {PSM-TV-S and PSM-TV-ST}, but contain visible noise-like artifacts on the piecewise constant regions of the walnut object, which are alleviated by RED-PSM. This is manifested also in the absolute difference figures, with error for TD-DIP distributed throughout the interior regions of the walnut. Also, around the shell of the walnut, RED-PSM is further able to preserve sharper details. For the compressed material, in comparison to TD-DIP, RED-PSM shows reduced absolute error almost uniformly over the object. This difference is more prominent around the highly dynamic regions of the object, emphasizing the advantage of the proposed method. {For a better understanding of RED-PSM results, we also display the reconstructed spatial and temporal basis $\Lambda$ and $\Psi$ for {the time-varying walnut scenario} in Supplementary Material 
\hyperref[sec:supp_sample_basis]{C}.}

\begin{table*}[hbtp!]
    \small
    \setlength{\tabcolsep}{0.45pt}
    \renewcommand{\arraystretch}{0.55}
    \centering
    \begin{tabular}{cccc}
    & \small{Reconstructed slices} & \small{Absolute reconstruction errors} & \vspace{-0.03cm} \\
    {\begin{sideways} $\quad\quad\quad\quad\quad$Walnut\end{sideways}} &
    \includegraphics[width=0.525\linewidth]{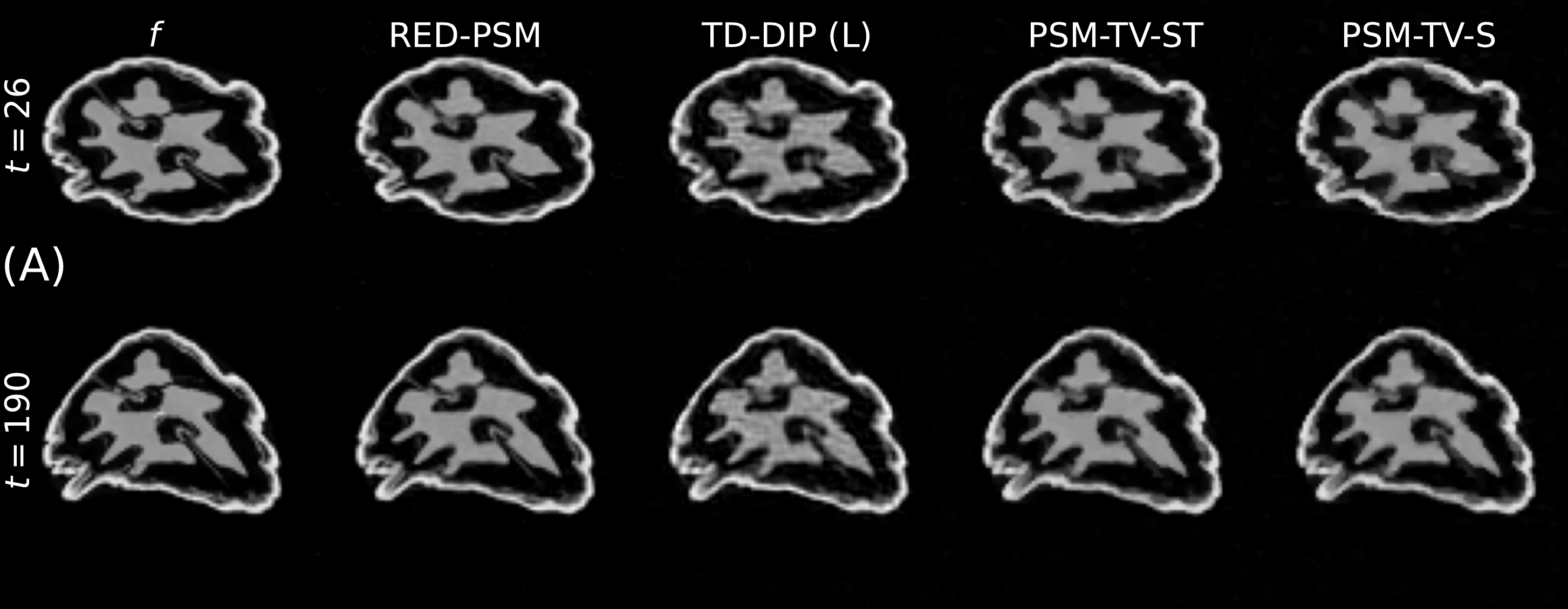} \hspace{-0.15cm} &
    \includegraphics[width=0.42\linewidth]{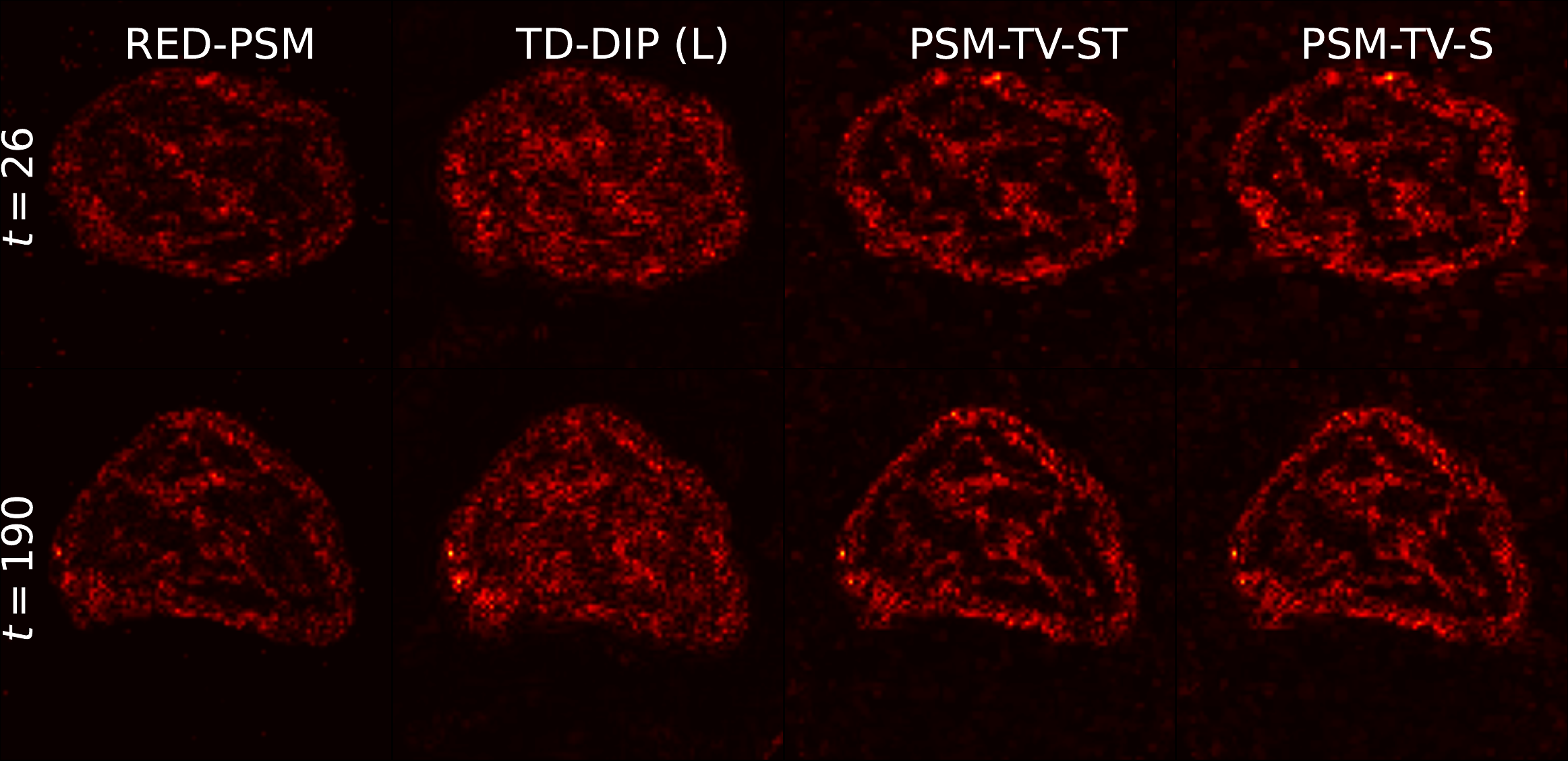} &
    \includegraphics[width=0.045\linewidth]{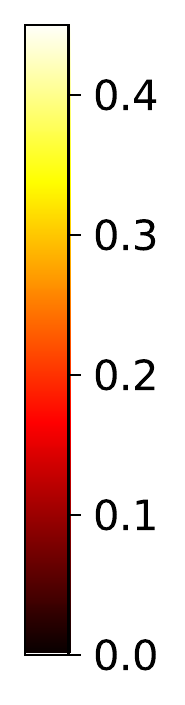}
    \vspace{-0.5cm} \\
    {\begin{sideways} $\quad\quad\quad\quad\quad$Material\end{sideways}} &
    \includegraphics[width=0.525\linewidth]{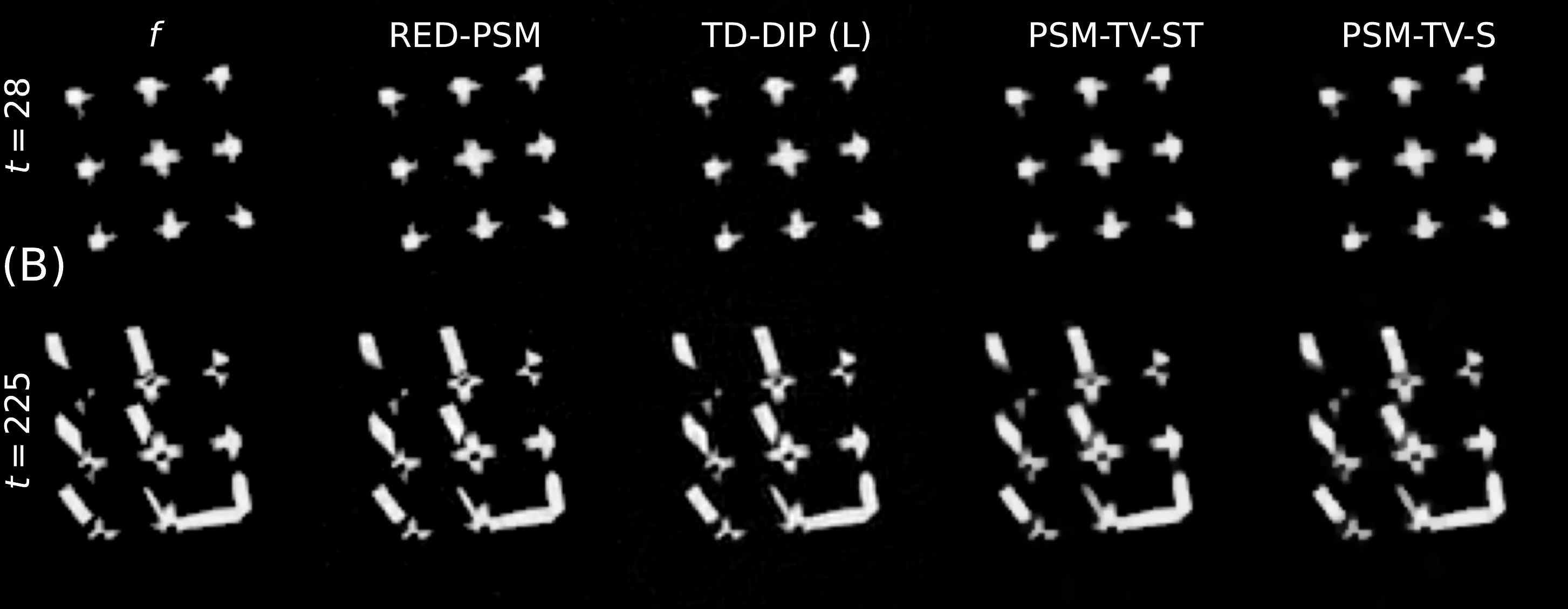} \hspace{-0.15cm} & 
    \includegraphics[width=0.42\linewidth]{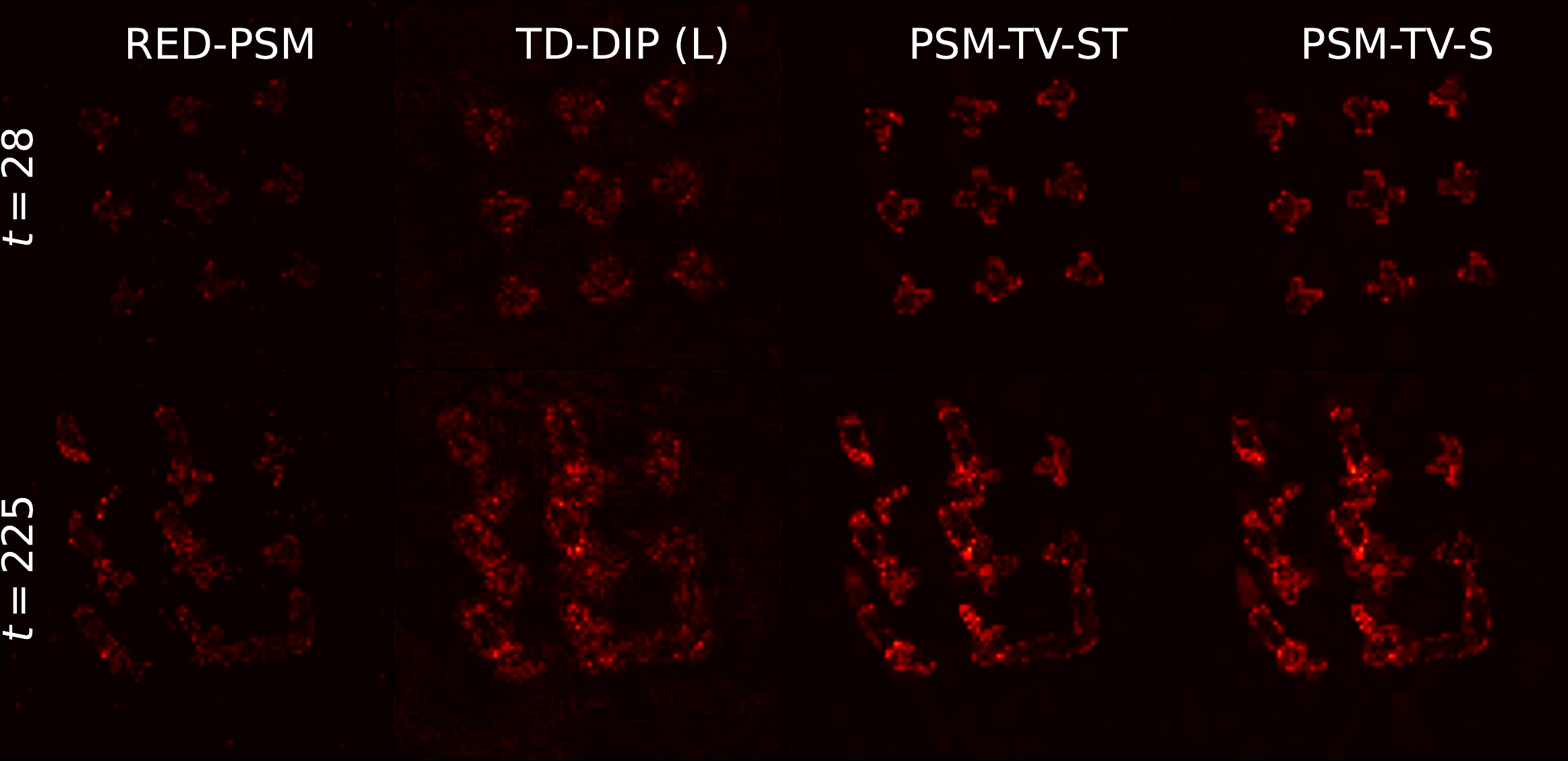} &
    \includegraphics[width=0.045\linewidth, trim=-0cm 0 0.0cm -0.25cm]{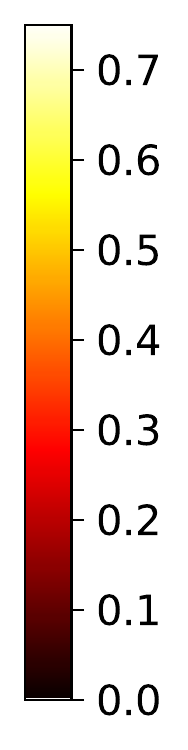}
    \end{tabular}
    \setlength{\abovecaptionskip}{0pt}
    \setlength{\belowcaptionskip}{-10pt}
    \captionof{figure}{\small{Comparison of {reconstructed object frames at two time instants using different methods for $P$=256, and the corresponding normalized absolute reconstruction errors for (A) the time-varying walnut, and (B) compressed object.}}}
    \label{fig:recon_est}
\end{table*}

{Reconstructed $x$-$t$ ``slices" through the dynamic walnut are compared in Figure \ref{fig:walnut_recon_slice}.} The location of the $x$-$t$ slice} is highlighted on the $t=0$ static $x$-$y$ frame by a yellow line. Consistent with the comparison in Figure \ref{fig:recon_est}, RED-PSM {provides} reduced absolute error values throughout the respective {$x$-$t$ slice.} Also, as more apparent on the error figures, TD-DIP leads to higher background errors.

{Finally, a zoomed-in comparison of the time-varying walnut object for another time instant for $P=256$ is provided in Figure {\ref{fig:zoomed_in_comp_manuscript}}. The comparison shows the better performance of RED-PSM at recovering the finer details clearly.}

\begin{table*}[hbtp!]
\small
\setlength{\tabcolsep}{0.05pt}
\renewcommand{\arraystretch}{0.55}
\centering
\begin{tabular}{ccc}
\includegraphics[width=0.16\linewidth]{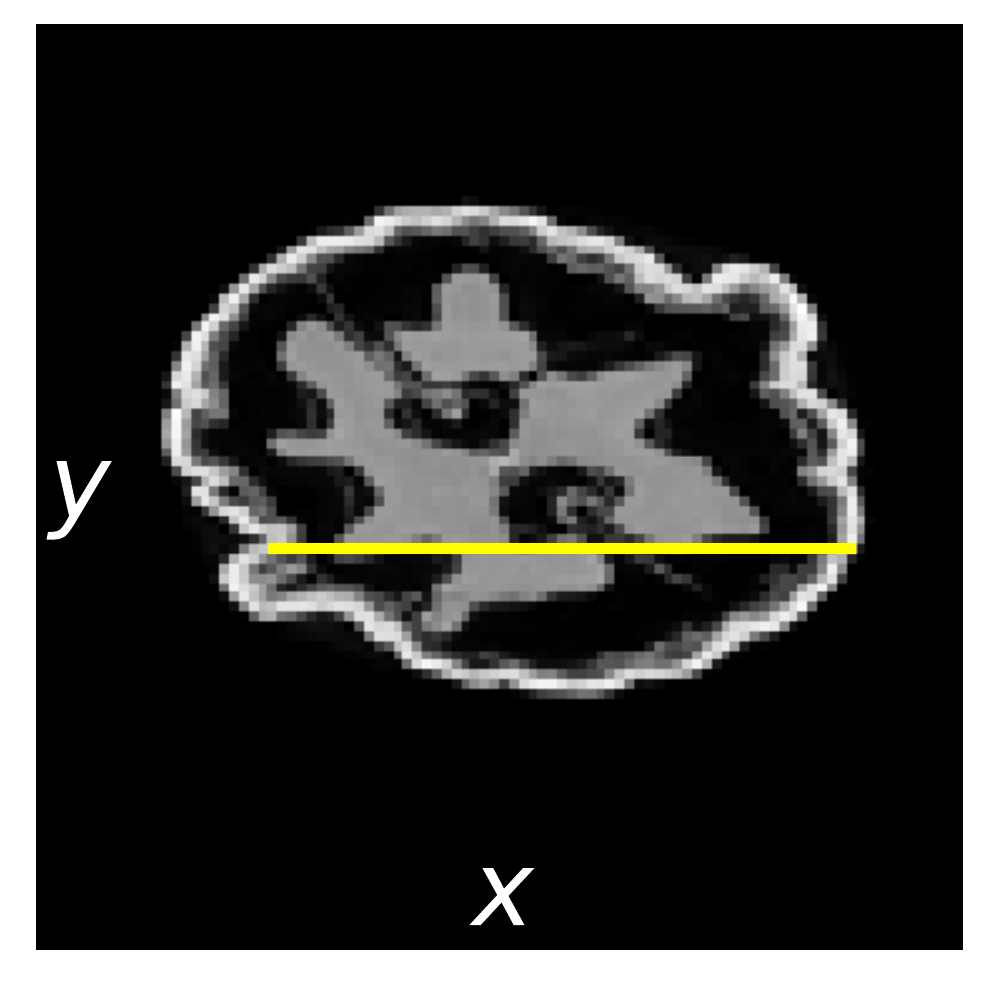} & 
\includegraphics[width=0.80\linewidth]{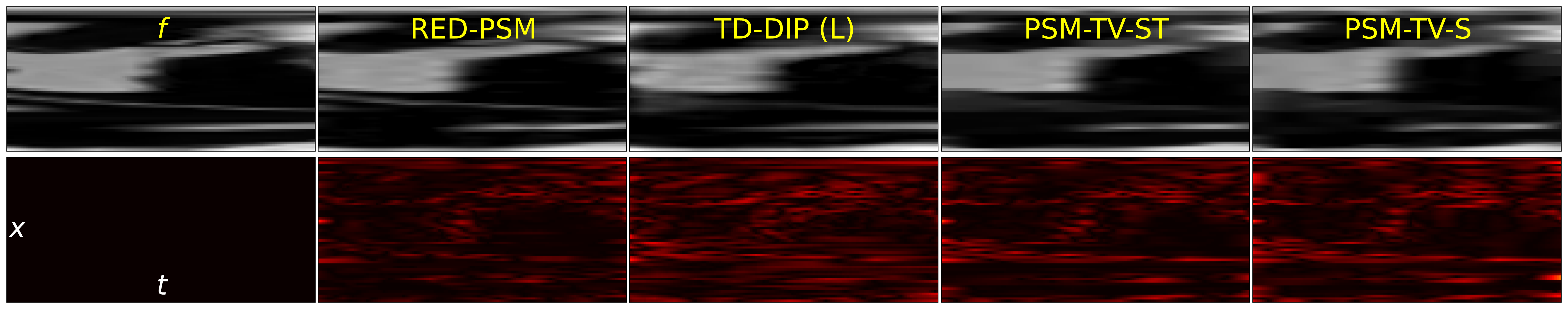}
& \hspace{-0.1cm}
\includegraphics[width=0.052\linewidth]{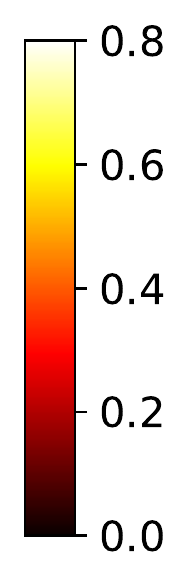}
\end{tabular}
\setlength{\abovecaptionskip}{0pt}
\setlength{\belowcaptionskip}{-12pt}
\captionof{figure}{\small {Comparison of reconstructed $x$-$t$ %
{slices (top) and corresponding normalized absolute error (bottom)}
using different methods for $P$=256. The cross-section location is indicated {on the static $t=0$ object with a yellow line.}
The $x$-$y$-$t$ {coordinates are indicated in white text on the static object and bottom left absolute error figure.}}
}
\label{fig:walnut_recon_slice}
\end{table*}

    \begin{table}[hbtp!]
    \small
    \setlength{\tabcolsep}{0.05pt}
    \renewcommand{\arraystretch}{0.55}
    \centering
    \begin{tabular}{c}
    \includegraphics[width=0.95\linewidth]{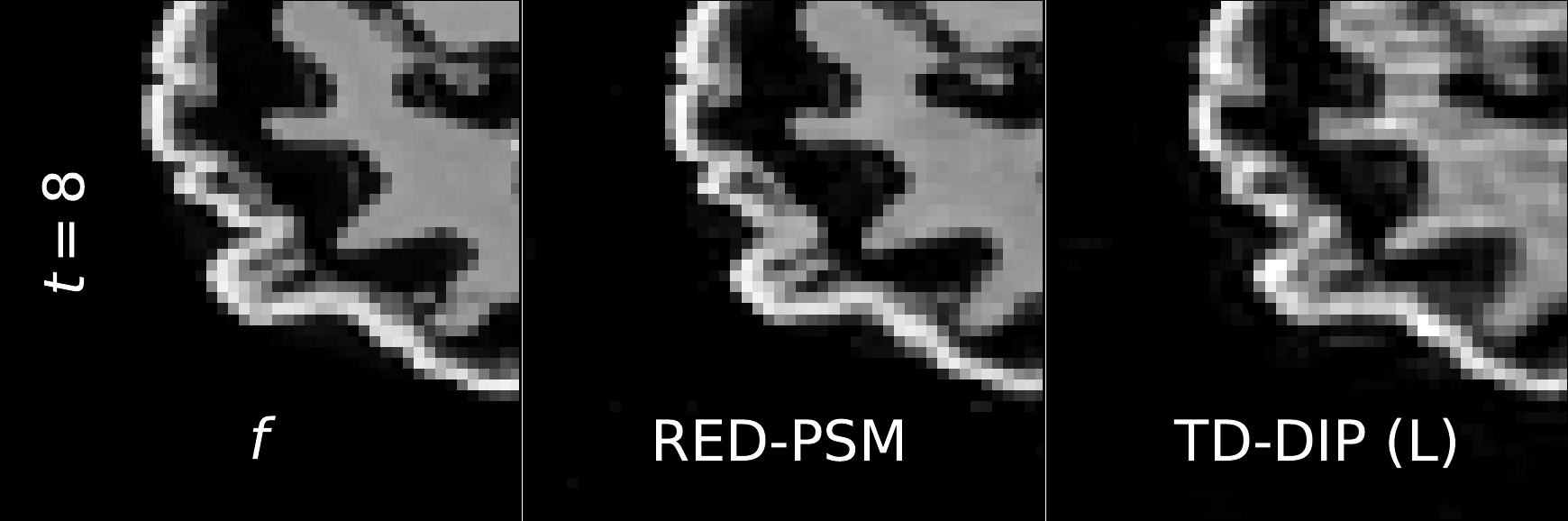} 
    \end{tabular}
    \setlength{\abovecaptionskip}{0pt}
    \setlength{\belowcaptionskip}{-10pt}
    \captionof{figure}{\small {Zoomed-in walnut reconstruction comparison for $P$=256.}}
    \label{fig:zoomed_in_comp_manuscript}
    \end{table}

\subsubsection{PSNR vs. $t$ Comparisons}
\label{sec:psnr_vs_t}
{To complement the cumulative metrics in  Figure \ref{fig:walnut_metrics_errbar_vs_P} and Table \ref{tab:avg_rec_acc_comp} and the ``snapshot" qualitative comparisons in Figure \ref{fig:recon_est}, we study how the reconstructed image frame PSNRs vary over the reconstructed time interval.} {The per frame PSNRs (in dB) of the walnut and compressed object reconstructed with $P=256$ by the different methods are shown in Figure~\ref{fig:walnut_metrics_PSNR_vs_t_supp} as a function of $t$.} {For TD-DIP, the best PSNR obtained using a ``stopping oracle" is reported, with the red shading indicating, for each $t$, the interval between the highest and lowest PSNR in three runs with different random initial conditions.}
{For the warped walnut object,} RED-PSM provides consistently better PSNR {than the best-case TD-DIP for all $t$.}
{For the compressed object, the same is true at about $70\%$ of $t$ points.} {Figure~\ref{fig:walnut_metrics_PSNR_vs_t_supp} also shows transient effects at the beginning and the end for both objects and all methods. In scenarios such as the object compression experiment, in which the initial and final state are static and could be measured using multiple projections, such transients could be eliminated. Similarly, in quasi-periodic scenarios such as cardiac imaging the effect of such transients would be minimal.}

    \begin{table}[hbtp!]
    \small
    \setlength{\tabcolsep}{-0pt}
    \renewcommand{\arraystretch}{0.55}
    \centering
    \begin{tabular}{cccc}
    {\begin{sideways} $\quad\quad\quad$(a) Walnut\end{sideways}}&\includegraphics[width=0.47\linewidth]{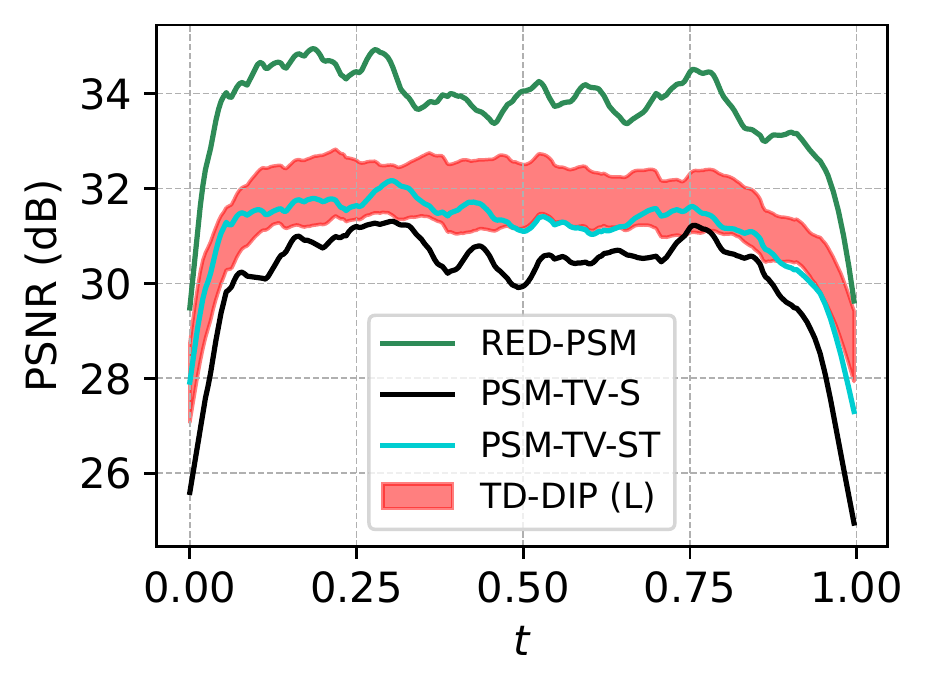} &
    {\begin{sideways} $\quad\quad\quad$(b) Material\end{sideways}} &
    \includegraphics[width=0.47\linewidth]{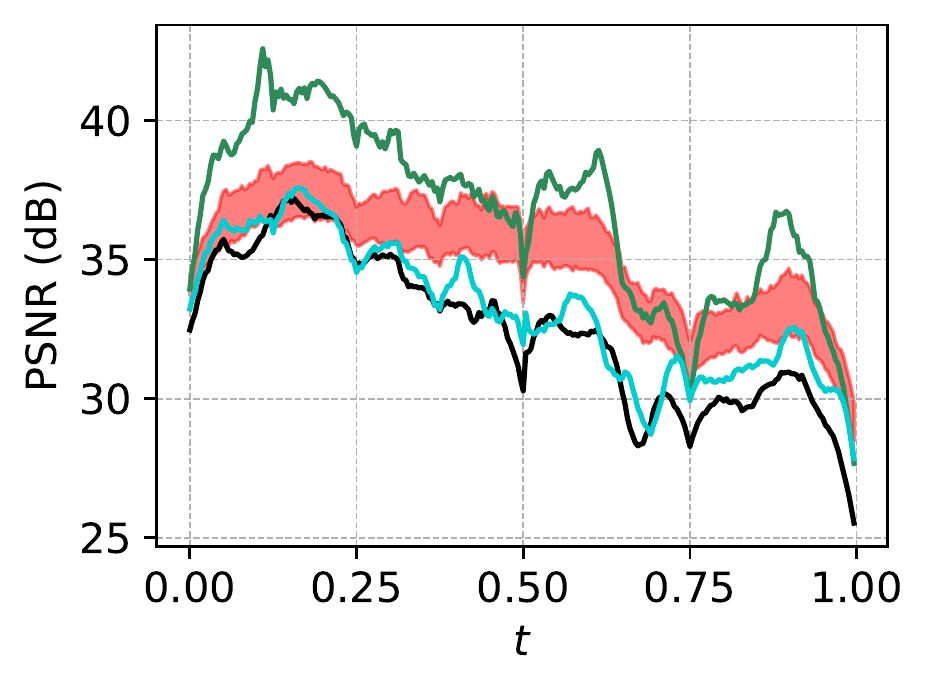}
    \end{tabular}
    \setlength{\abovecaptionskip}{0pt}
    \setlength{\belowcaptionskip}{-12pt}
    \captionof{figure}{\small {Reconstruction PSNR vs. $t$ for the (a) time-varying walnut, and (b) compressed material for $P$=256. {The red shading for TD-DIP indicates, for each $t$, the interval between the best and worst PSNR with the early stopping oracle explained in Section \ref{sec:comparison_wrt_P} in three runs with different random initial conditions.
    }}}\label{fig:walnut_metrics_PSNR_vs_t_supp}
    \end{table}

\subsubsection{Effect of Initialization}
\label{sec:prosep_init_exp}
The initialization of $\Lambda$, $\Psi$, and $f$ plays an important role in the performance and convergence speed of RED-PSM. We observe significant speed-up when {rather than a random initialization,} we initialize the algorithm with ProSep \cite{iskender2022dynamic} estimated reconstruction. Figure \ref{fig:random_vs_prosep_init} shows PSNR vs. iterations comparison for different initialization techniques for the dynamic walnut object with $P=256$. The rest of the parameters were selected identical to those indicated in Supplementary Material Table \ref{tab:avg_rec_acc_comp_params}. This experiment highlights the advantages of initializing with ProSep estimated basis functions: eliminating the need for multiple runs for a best-case result; and speeding up convergence considerably.

{Combined also with the convergent algorithm eliminating the need for an unrealistic stopping oracle and the theoretical analysis, RED-PSM provides improved reliability which is of practical significance.}

\begin{figure}[hbtp!]
    \centering
    \includegraphics[width=0.62\linewidth]{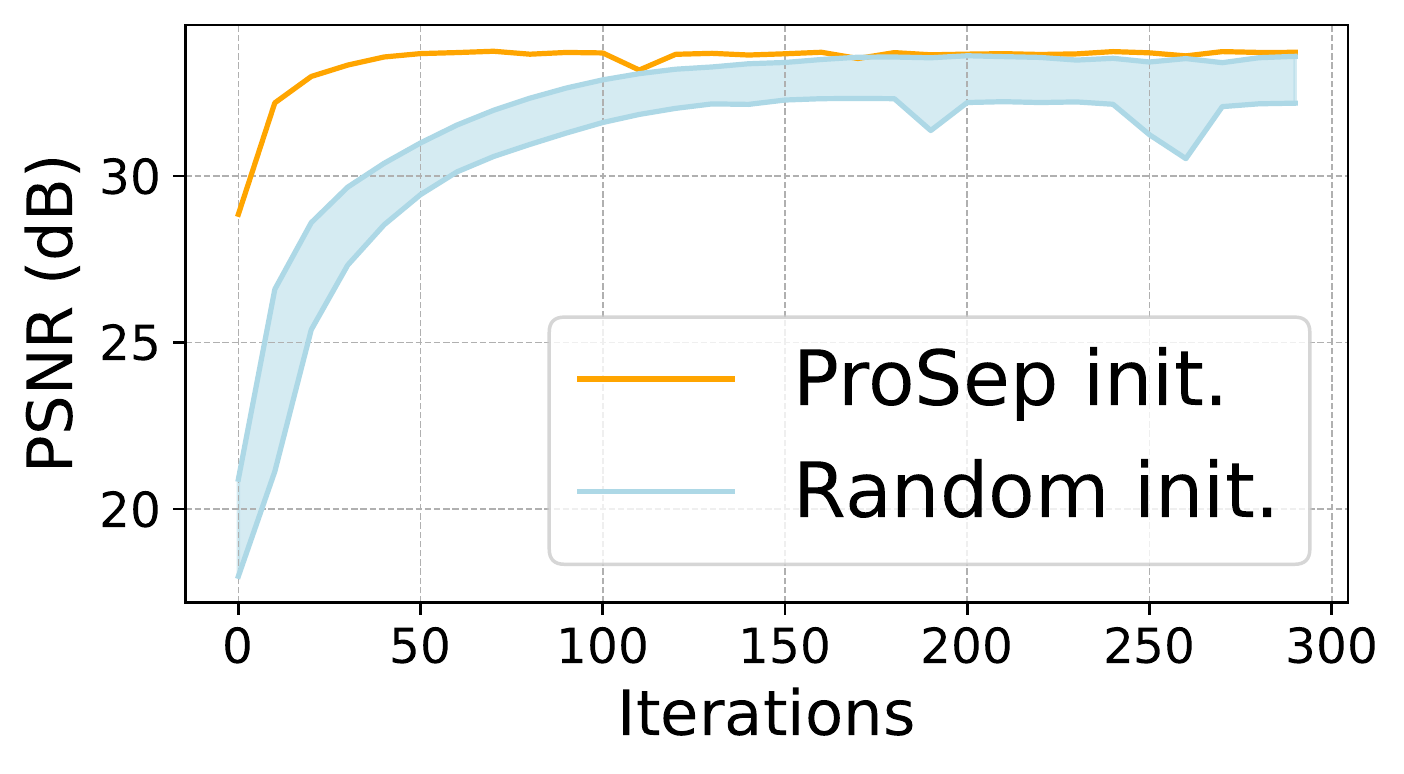}
    \setlength{\abovecaptionskip}{0pt}
    \setlength{\belowcaptionskip}{-10pt}
    \caption{\small 
    {{Advantage of ProSep-based vs. random initialization of RED-PSM} {(see Sec.~\ref{sec:framework})}. 
    {For random initialization,}
    the area {in blue} between the best and the worst PSNR for each iteration {highlights} the varying performances of five different runs with {different} random initializations.}}
\label{fig:random_vs_prosep_init}
\end{figure}
\subsubsection{Cardiac dMRI data experiments}
\label{sec:cardiac_mri_exp}
{In this setting, different to previous experiments, we used 4 $k$-space radial lines (``spokes") per frame at the bit-reversed angles. We used 1.4 and 2.8 cardiac cycles, with 23 $\times$ 4 = 92 spokes/cycle, for a total of $P$=$128$ and $P$=$256$ spokes. The problem is still severely undersampled compared to the experiment in \cite{yoo2021time} where 13 spokes are used per frame for 13 cycles, for a total of $13 \times 23 \times 13 = 3,887$ spokes.

Since the data is periodic, we also tested the helix latent scheme (H) for TD-DIP.

The metrics in Table \ref{tab:acc_mri} and the qualitative comparison in Figure \ref{fig:comp_mri} with zoomed-in reconstructions and absolute error maps, show that RED-PSM performs better than both versions of TD-DIP.

\begin{table}[hbtp!]
\footnotesize
\setlength{\tabcolsep}{2.5pt}
\renewcommand{\arraystretch}{0.45}
\centering
\begin{tabular}{@{}clcccc@{}}
\toprule
\multicolumn{1}{c}{$P$}&\multicolumn{1}{c}{Method}
& \multicolumn{1}{c}{PSNR (dB)} & \multicolumn{1}{c}{SSIM} & \multicolumn{1}{c}{MAE (1e-2)} & \multicolumn{1}{c}{HFEN}\\
\midrule
  128 & TD-DIP (L) & 34.6 & 0.923 & 1.7 & 3.38 \\
  & TD-DIP (H) & 34.6 & 0.928 & 1.7 & 3.32 \\
   & \textbf{RED-PSM} & \textbf{36.6} & \textbf{0.939} & \textbf{1.4} & \textbf{2.73} \\
   \midrule
  256 & TD-DIP (L) & 36.2 & 0.948 & 1.4 & 3.47 \\
  & TD-DIP (H) & 36.4 & 0.947 & 1.4 & 3.44 \\
   & \textbf{RED-PSM} & \textbf{38.4} & \textbf{0.962} & \textbf{1.1} & \textbf{2.99} \\
 \bottomrule
\end{tabular}
\setlength{\abovecaptionskip}{2pt}
\setlength{\belowcaptionskip}{-8pt}
\caption{\small Reconstruction accuracies for RED-PSM and TD-DIP for the retrospective dMRI data \cite{yoo2021time}.}
\label{tab:acc_mri}
\end{table}

\begin{table}[hbtp!]
\small
\setlength{\tabcolsep}{0.05pt}
\renewcommand{\arraystretch}{0.55}
\centering
\begin{tabular}{c}
 \vspace{-0.075cm} \includegraphics[width=0.95\linewidth]{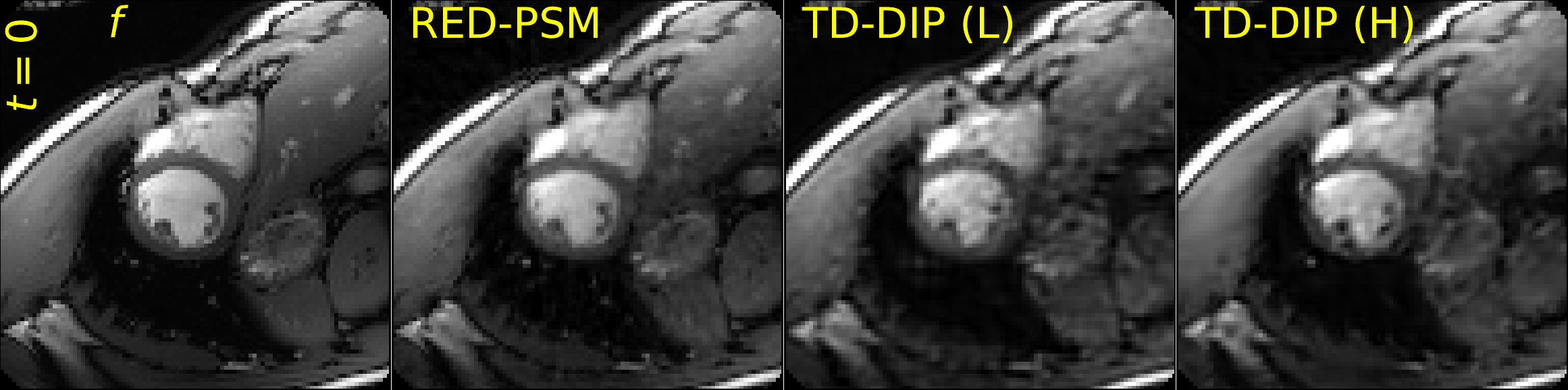} 
\\
 \vspace{-0.075cm}
\includegraphics[width=0.95\linewidth]{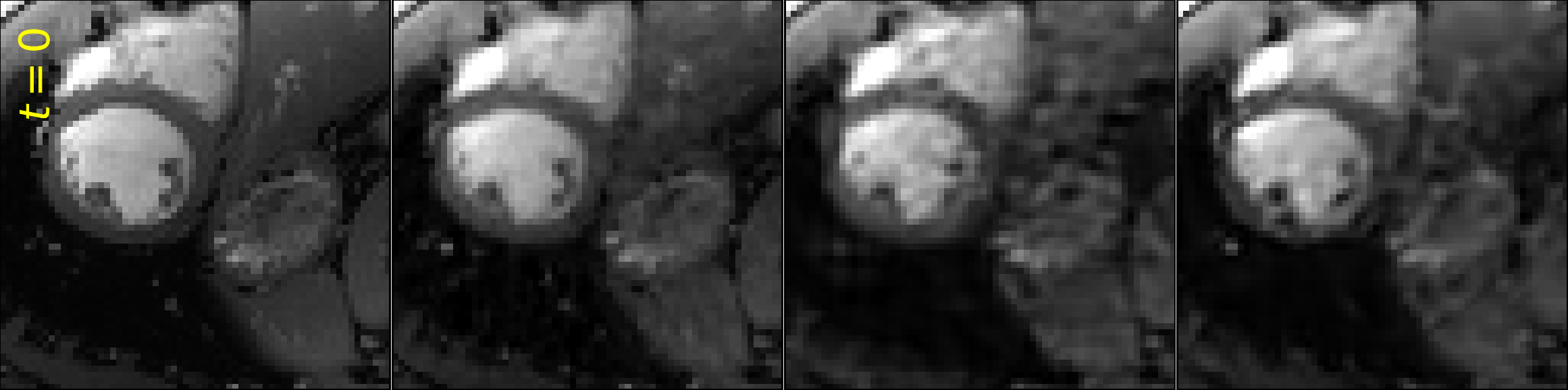}
\\
\includegraphics[width=0.95\linewidth]{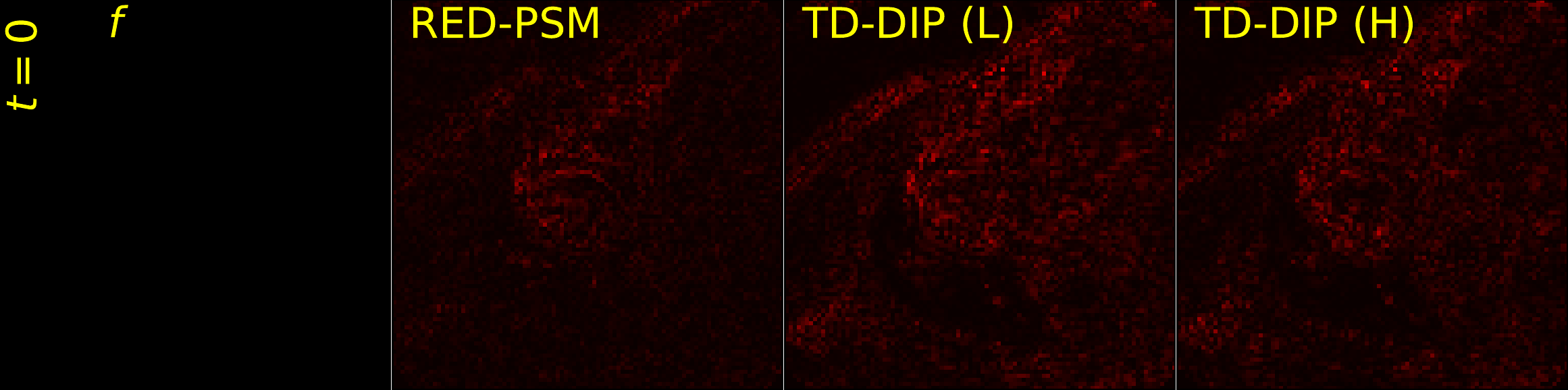}
\end{tabular}
\setlength{\abovecaptionskip}{0pt}
\setlength{\belowcaptionskip}{-14pt}
\captionof{figure}{
\small
Reconstructed frames for $P$=256 for retrospective dMRI data \cite{yoo2021time} with zoomed-in frames (middle row), and absolute reconstruction errors (last row).
}
\label{fig:comp_mri}
\end{table}
}

{In \cite{zou2021dynamic}, the authors compare their method to an \textit{older} version of TD-DIP without the improved latent representation prior scheme for cardiac dMRI and report 1.5 dB PSNR improvement. In RED-PSM, we exceed this improvement on the latest version of TD-DIP in multiple scenarios.}

\subsubsection{Acquisition with smaller number of distinct view angles}
\label{sec:acquisition_period}
{Since obtaining time-sequential projections from different angles in a sufficiently short time period can be physically challenging, we also test the performance of RED-PSM with an acquisition scheme that may be easier to realize physically: keeping the same total number of $P$ projections, but taken at a smaller number $\hat{P}$ (called ``period") of \emph{distinct} view angles, which are also obtained using the bit-reversed angular sampling scheme. The comparison in Figure \ref{fig:walnut_metrics_PSNR_vs_period} shows that it is possible to use up to 1/8-th of the distinct view angles without performance loss for RED-PSM.}

    \begin{table}[hbtp!]
    \small
    \setlength{\tabcolsep}{-0pt}
    \renewcommand{\arraystretch}{0.55}
    \centering
    \begin{tabular}{cc}
    \includegraphics[width=0.47\linewidth]{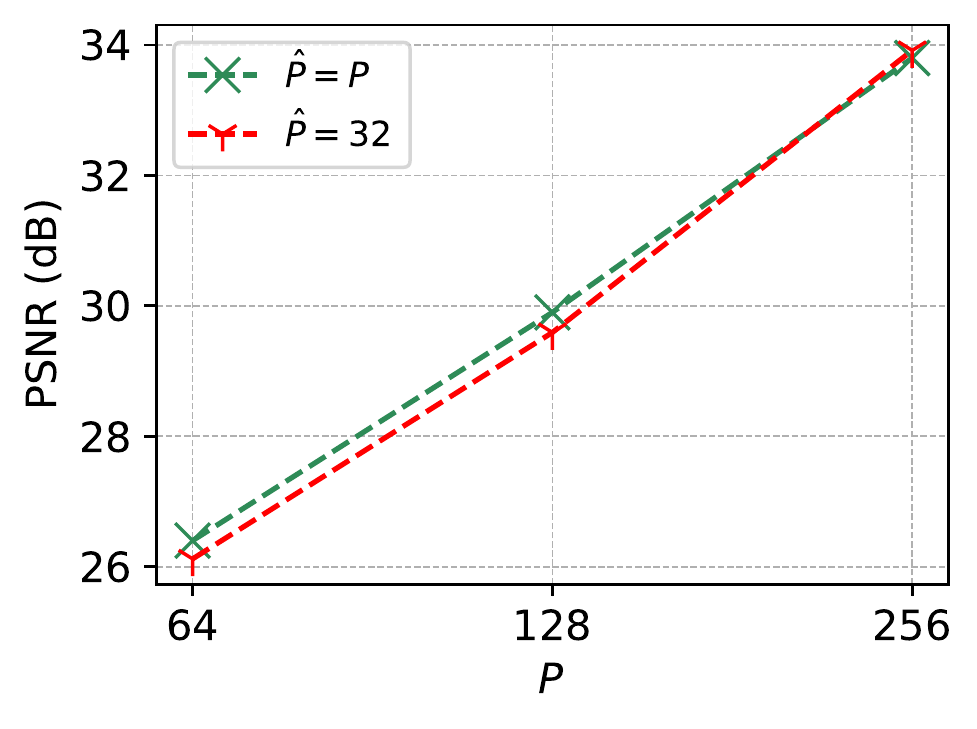} &
    \includegraphics[width=0.47\linewidth]{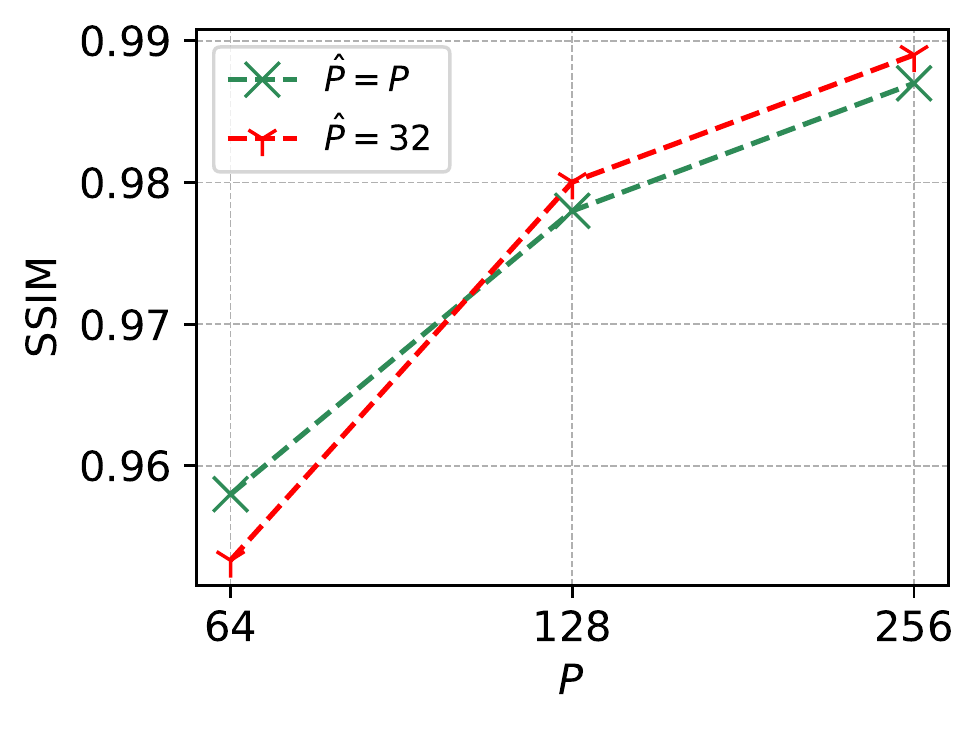}
    \end{tabular}
    \setlength{\abovecaptionskip}{0pt}
    \setlength{\belowcaptionskip}{-6pt}
    \captionof{figure}{
    \small{Reconstruction PSNR and SSIM for the time-varying walnut vs. the number of distinct view angles $\hat{P}$ using RED-PSM with different total number of views $P$.}
    }
    \label{fig:walnut_metrics_PSNR_vs_period}
    \end{table}

\subsubsection{Patch-based RED denoiser}
\label{sec:patch-based_denoiser}
To improve the scalability of the method to higher resolution and/or 3D dynamic object, conveniently, the objective in \eqref{eq:hard_cnst_RED_obj} can be manipulated to operate on the patches of temporal {image frames} of the time-varying object. {This circumvents the need to store the complete image frame at a given time, and also enables the denoiser $D_\phi$ to be both trained and to operate on patches of image frames.} To showcase the potential of the suggested scheme, we replace the full-size $D_\phi$ with a patch-based counterpart in the RED step and compare the performance with the originally proposed method for 2D dynamic objects.

The patch-based denoiser for RED updates is trained using
\begin{equation}
\label{eq:patch_based_denoiser_train}
\min_{\phi} \sum_i \sum_l \| B_lf_i - D_\phi (B_l\tilde{f}_i) \|_F^2 \,\, \text{s.t.} \,\, \tilde{f}_i = f_i + \eta_i, \, \forall i, \nonumber
\end{equation}
where $\eta_i, \sigma_i$ are set as in \eqref{eq:red_denoiser_train}, and $B_l$, with $l \in \{0, \ldots, L-1\}$, is the operator to extract the $l$-th patch of the image. The denoiser $D_\phi$ operates separately on each patch.

To train the patch-based $D_\phi$, uniformly random rotations (multiples of $\frac{\pi}{2}$), and random horizontal/vertical flips, each with $\frac{1}{2}$ probability, were used for data augmentation. {The patch size was chosen 8$\times$8 with a stride of 2.}

Table \ref{tab:avg_rec_acc_red_admm_patchbased} compares the results for the two denoiser types for both objects, using the same denoiser training policy as in Section \ref{sec:framework}, and experimental configurations of Table \ref{tab:avg_rec_acc_red_admm_patchbased_params} in the Supplementary Material Section~\hyperref[supp:exp_config]{A}.
The results show little difference {{between the patch-based denoiser and full frame-based denoiser RED-PSM variants,}} thus verifying the effectiveness of the patch-based version. The analysis of the computational requirements of the patch-based RED-PSM variant in Section~\ref{sec:complexity_analysis} shows its potential for a highly scalable implementation. 

We note that in a divergent beam scenario, the contribution of a patch to a projection will be position-dependent and this would require accurate bookkeeping. {Details for doing so can be found in} tile-based {methods for} fan-beam \cite{xiao2002n} and cone-beam \cite{bresler2003fast} tomography. However, we leave such analysis for future studies.

\begin{table}[hbtp!]
    \footnotesize
    \setlength{\tabcolsep}{2.5pt}
    \renewcommand{\arraystretch}{0.4}
    \centering
    \begin{tabular}{@{}cclcccc@{}}
    \toprule
    \multicolumn{1}{c}{Object} & \multicolumn{1}{c}{$P$} & \multicolumn{1}{c}{Denoiser} & \multicolumn{1}{c}{PSNR(dB)} & \multicolumn{1}{c}{SSIM} & \multicolumn{1}{c}{MAE(1e-3)} & \multicolumn{1}{c}{HFEN}\\
    \midrule
      Walnut & 256 & Full-image & 33.8 & 0.987 & 0.4 & 0.32 \\
      & & Patch-based & 33.7 & 0.989 & 0.4 & 0.35 \\
     \midrule
      Material & 256 & Full-image & 35.9 & 0.986 & 2.5 & 4.43 \\
      & & Patch-based & 35.6 & 0.981 & 2.7 & 4.56 \\
     \bottomrule
    \end{tabular}
    \setlength{\abovecaptionskip}{2pt}
    \setlength{\belowcaptionskip}{-16pt}
    \caption{\small Performance comparison for different denoiser types for RED-PSM. Patch-based denoisers also use DnCNN and have the same configuration and training policy as the full image denoiser.}
    \label{tab:avg_rec_acc_red_admm_patchbased}
    \end{table}
    
 \section{Conclusions}
We proposed RED-PSM, the first PSM-based approach to dynamic imaging using a pre-trained and learned (RED-based) spatial prior. The objective in the proposed variational formulation is optimized using a novel and effective bi-convex ADMM algorithm, which enforces the PSM as a hard constraint. Unlike existing PSM-based techniques, RED-PSM is supported by theoretical analysis, with a convergence guarantee to a stationary point of the objective. The results of the numerical experiments show better reconstruction accuracy and considerably faster run times compared to a recent DIP-based algorithm. A patch-based regularizer version of RED-PSM provides almost equivalent performance with a massive reduction of storage requirements, indicating the potential of our framework for dynamic high-resolution 2D or 3D settings.

Possible directions for future work include the application of RED-PSM to different imaging scenarios other than tomography and MRI, and robust denoiser training for RED framework since the deep denoisers encounter {varying} artifact distributions during optimization. This could also improve the generalizability of the framework to different input types.

\smallskip

{\textbf{Acknowledgements.} We thank Dr. Brian M. Patterson for providing the compressed material experiment data.}

{
\scriptsize
\bibliographystyle{IEEEtran}
\bibliography{preambley, egbib, dMRI_NoDup_egbib, dCT}
}

\clearpage
\section*{Supplementary Material}
\subsection{Experimental configurations}
\label{supp:exp_config}
The PSM-TV and PSM-RED parameter selections for the experiments listed in Table \ref{tab:avg_rec_acc_comp} are provided in Table \ref{tab:avg_rec_acc_comp_params}. Likewise, Table \ref{tab:avg_rec_acc_red_admm_patchbased_params} shows the parameter configurations for the denoiser type comparison experiments for RED in Table \ref{tab:avg_rec_acc_red_admm_patchbased}. 
Finally, the architectural information for the DnCNN denoisers used throughout this work is in Table \ref{tab:denoiser_arch}.

    \begin{table}[hbtp!]
    \footnotesize
    \setlength{\tabcolsep}{1.75pt}
    \renewcommand{\arraystretch}{0.4}
    \centering
    \begin{tabular}{@{}cl|ccccc|ccccc@{}}
    \toprule
     \multicolumn{2}{c}{} & \multicolumn{5}{c}{(a) Walnut} & \multicolumn{5}{c}{(b) Comp. Material} \\
     \midrule
    \multicolumn{1}{c}{$P$} & \multicolumn{1}{c}{Method} & \multicolumn{1}{c}{$K$} & \multicolumn{1}{c}{$d$} & \multicolumn{1}{c}{$\lambda$} & \multicolumn{1}{c}{$\tilde{\lambda}$} & \multicolumn{1}{c}{$\beta$} & \multicolumn{1}{c}{$K$} & \multicolumn{1}{c}{$d$} & \multicolumn{1}{c}{$\lambda$} & \multicolumn{1}{c}{$\tilde{\lambda}$} & \multicolumn{1}{c}{$\beta$}\\
    \midrule
      32 & PSM-TV-S (R) & 3 & 4 & 5e-2 & - & - & 3 & 4 & 10 & - & - \\
      32 & {PSM-TV-ST (R)} & 4 & 5 & 5e-2 & 5e-2 & - & 3 & 4 & 1e1 & 1e2 & - \\
      32 & PSM-RED (P) & 3 & 7 & 1e-4 & - & 1e-4 & 3 & 9 & 5.12e-2 & - & 1.6e-2 \\
     \midrule
     64 & PSM-TV-S (R) & 3 & 4 & 5e-2 & - & - & 5 & 6 & 10 & - & - \\
     64 & {PSM-TV-ST (R)} & 4 & 5 & 5e-2 & 5e-2 & - & 5 & 6 & 1e1 & 1e2 & - \\
      64 & PSM-RED (P) & 4 & 7 & 1e-4 & - & 5e-4 & 5 & 9 & 32e-4 & - & 4e-3 \\
     \midrule
      128 & PSM-TV-S (R) & 5 & 7 & 5e-2 & - & - & 8 & 9 & 5 & - &  - \\
      128 & {PSM-TV-ST (R)} & 6 & 7 & 5e-2 & 1e-1 & - & 8 & 9 & 5 & 5e1 & - \\
      128 & PSM-RED (P) & 6 & 9 & 2e-4 & - & 2e-4 & 8 & 11 & 4e-4 & - & 2e-3  \\
     \midrule
     256 & PSM-TV-S (R) & 10 & 11 & 5e-2 & - & - & 11 & 13 & 10 & - & -  \\
     256 & {PSM-TV-ST (R)} & 10 & 11 & 5e-2 & 5e-2 & - & 11 & 13 & 1e1 & 1e2 & - \\
     256 & PSM-RED (P) & 10 & 11 & 5e-5 & - & 1e-4 & 11 & 13 & 1e-4 & - & 1e-3 \\
     \bottomrule
    \end{tabular}
    \caption{\small The parameter selections for the reconstructions in Table \ref{tab:avg_rec_acc_comp}. The latent penalty weight was selected as $\xi$=$10^{-1}$ for the walnut, and $\xi$=$10^{-3}$ for the compressed object experiments.}
    \label{tab:avg_rec_acc_comp_params}
    \end{table}

    \begin{table}[hbtp!]
    \footnotesize
    \setlength{\tabcolsep}{2.5pt}
    \renewcommand{\arraystretch}{0.5}
    \centering
    \begin{tabular}{@{}cl|cccc|cccc@{}}
    \toprule
    \multicolumn{1}{c}{}&\multicolumn{1}{c}{}&\multicolumn{4}{c}{(a) Walnut}&\multicolumn{4}{c}{(b) Comp. Material}\\
    \midrule
    \multicolumn{1}{c}{$P$}&\multicolumn{1}{c}{Denoiser}&\multicolumn{1}{c}{$K$}&\multicolumn{1}{c}{$d$}&\multicolumn{1}{c}{$\lambda$}&\multicolumn{1}{c}{$\beta$}&\multicolumn{1}{c}{$K$}&\multicolumn{1}{c}{$d$}&\multicolumn{1}{c}{$\lambda$}&\multicolumn{1}{c}{$\beta$}\\
    \midrule
      256 & Full-image & 10 & 11 & 5e-5 & 1e-4 & 11 & 13 & 1e-4 & 1e-3 \\
      256 & Patch-based & 10 & 11 & 5e-5 & 1e-4 & 11 & 13 & 1e-4 & 1e-3 \\
     \bottomrule
    \end{tabular}
    \caption{\small Parameter configurations for the denoiser type study experiments in Table \ref{tab:avg_rec_acc_red_admm_patchbased}.}
    \label{tab:avg_rec_acc_red_admm_patchbased_params}
    \end{table}

    \begin{table}[hbtp!]
    \footnotesize
    \setlength{\tabcolsep}{2.5pt}
    \renewcommand{\arraystretch}{0.5}
    \centering
    \begin{tabular}{@{}lccc@{}}
    \toprule
    \multicolumn{1}{c}{Dataset} & 
    \multicolumn{1}{c}{$\#$ of layers} & 
    \multicolumn{1}{c}{$\#$ of channels} & 
    \multicolumn{1}{c}{Denoising}\\
    \midrule
       Walnut & 6 & 64 & Direct \\
       Compressed material & 3 & 32 & Residual \\
       Cardiac dMRI & 6 & 64 & Residual \\
     \bottomrule
    \end{tabular}
    \caption{\small Denoiser DnCNN configurations for different datasets.}
    \label{tab:denoiser_arch}
    \end{table}

\subsection{Time and space complexity analyses}
\label{sec:time_space_comp}
In this section, we analyze the operation count and storage requirements for the proposed algorithm and its patch-based version. 

In this analysis, $M_i$ represents the number of inner iterations for iteratively solved subproblems, assumed common to the different subproblems, and we assume $K \ll N$. Also, $R_{\theta(t)}$ and $R_{\theta(t)}^T$ are implemented as operators, each requiring $O(N^2)$ when applied {to a $N \times N$ image} or a projection of size $N$.

\subsubsection{Storage requirements for the input quantities}
When calculations are performed sequentially in $t$ as in the following sections, the storage requirements for the primal and dual variables are calculated for $f_t$, $\gamma e_t$, $\Psi^T e_t$ and $\Lambda$. The {storage of the} spatial basis functions $\Lambda$ dominates the overall {storage cost at} $O(KN^2)$. 
The space complexity analyses of different subproblems in the following sections consider terms other than these input quantities.

\subsubsection{RED-PSM}
Expanding the $\Lambda$ and $\Psi$ subproblems in Algorithm \ref{alg:admm_psm_red} in $t$, we have
\begin{align}
\label{eq:lambda_RED_obj_over_t}
    \min_{\Lambda} &\sum_t \left( \| R_{\theta(t)}\Lambda \Psi^Te_t - g_t \|_2^2 + {\beta}\| (\Lambda\Psi^T - f + \gamma)e_t \|_2^2 \right) \nonumber \\ 
    &+ \xi \| \Lambda \|_{F}^2,
\end{align}
and
\begin{align}
\label{eq:psi_RED_obj_over_t}
    \min_{\Psi} &\sum_t \| R_{\theta(t)}\Lambda \Psi^Te_t - g_t \|_2^2 + \xi\| \Psi^Te_t \|_2^2 \nonumber \\ 
    &+ {\beta}\| (\Lambda\Psi^T - f + \gamma)e_t \|_2^2.
\end{align}
{We first} analyze the operation counts for gradients with respect to $\Lambda$ and $\Psi^T e_t$, respectively. 

We note that in these gradient computations, {the term $R_{\theta_t}^T g_t$ is pre-computed and stored before the start of the algorithm, and $(\gamma - f)e_t$ and $(R_{\theta(t)}^T g_t + (\gamma - f)e_t)$ are computed and stored at each bilinear ADMM outer iteration} so that they do not contribute to the respective costs {in each inner iteration, i.e, their cost does not scale with $M_i$.}

Starting with the gradient of \eqref{eq:psi_RED_obj_over_t} with respect to $\Psi^Te_t$, we have
\begin{align}
\label{eq:grad_psi_psm_fidelity}
    &2((R_{\theta(t)} \Lambda)^T(R_{\theta(t)} \Lambda) + \Lambda^T\Lambda + \xi I) \Psi^Te_t - 2(R_{\theta(t)} \Lambda)^Tg_t \nonumber \\
    &+ 2\Lambda^T (\gamma - f)e_t.
\end{align}
The overall {operation count} for this term is $O(K N^2 P)$ due to the computation of $R_{\theta(t)}\Lambda^T$ and $\Lambda^T(\gamma - f)e_t$ at each $t$. These terms do not need to be computed for each inner iteration.

When $\Psi = UZ$ is used, the gradient for $Z$ requires the multiplication of \eqref{eq:grad_psi_psm_fidelity} with $U^T e_t$, adding a minor operation count $O(dK)$.

The most efficient implementation in terms of space complexity is the sequential computation of the gradients for each $t$ without increasing the operation count. It is also possible to parallelize the operations along $t$ with the trade-off between the run time and the storage requirements.

Considering sequential computation for each $t$, the additional storage requirement is $O(NK)$ due to $R_{\theta(t)}\Lambda$.

{Next,} we consider the gradient of \eqref{eq:lambda_RED_obj_over_t} with respect to $\Lambda$ 
\begin{align}
\label{eq:grad_lambda_psm_fidelity}
    2\xi\Lambda + &\sum_t 2(R_{\theta(t)}^TR_{\theta(t)} + I) \Lambda (\Psi^T e_t)(e_t^T \Psi) \nonumber \\
    &+ (R_{\theta(t)}^T g_t + (\gamma - f)e_t)(e_t^T \Psi).
\end{align}
The operation count for this gradient term is $O(K N^2 P M_i)$ due to $R_{\theta(t)}^TR_{\theta(t)} \Lambda (\Psi^T e_t)(e_t^T \Psi)$.

Assuming sequential computations in $t$, the additional storage requirement is {determined} by the storage of the terms $R_{\theta(t)}^T R_{\theta(t)} \Lambda$, $R_{\theta(t)}^T R_{\theta(t)} \Lambda(\Psi^T e_t)(e_t^T \Psi)$, and $(R_{\theta(t)}^T g_t  + (\gamma - f)e_t)(e_t^T \Psi)$, each requiring $O(KN^2)$.

The operation count for the efficient implementation of the $f$-step in Line \ref{line:f_step_psm_fidelity_alg2} of Algorithm \ref{alg:admm_psm_red} is dominated by the single forward computation of the denoiser for each $t$, leading to $O(P C_D)$ where $C_D$ is the operation count of the denoiser $D_\phi$ for a single frame. In addition to the input variables, the only storage requirement is due to the term $\Lambda\Psi^Te_t$ which is $O(N^2)$.

\subsubsection{Patch-based RED-PSM}
\label{sec:patch_based_red_psm_complexity}
The patch-based {variant} of RED-PSM uses a reduced spatial size $N_B$ for its inputs where $N_B \ll N$. Compared to the originally proposed RED-PSM formulation in \eqref{eq:hard_cnst_RED_obj}, the operation count is also dependent on the stride $s$ of the patch extraction operators and the depth of the deep denoiser architecture. 
Assuming equal-depth or shallower deep convolutional denoisers for patch-based inputs, the total operation count due to the efficient denoiser step stays the same or decreases. 
Thus, the operation count is again {determined} by the gradient of the data fidelity term with respect to the patch-based spatial basis functions, and the result approximately scales by the increase in the total stored number of pixels as $O(K N^2 P M_i (\frac{N_B}{s})^2)$ where the stride $s \leq N_B$. 

The gradients with respect to the data fidelity term require the computation of projections of multiple spatial patches for a given view angle at time $t$. However, these projections can be computed separately and accumulated. Thus, the space complexity reduces to $O(KN_B^2)$ since we only need to store a single spatial patch of $\Lambda$ at a given time.

{A summary comparison of the two variants of RED-PSM of this section is provided in Section \ref{sec:complexity_analysis} of the manuscript.}

\subsection{Sample reconstructed spatial and temporal basis functions}
\label{sec:supp_sample_basis}
{To help interpret the operation of the proposed RED-PSM algorithm, we display the reconstructed spatial and temporal basis functions $\Lambda$ and $\Psi$ for the time-varying walnut {scenario} with $P=256$ and DCT-II latent temporal basis $U$ in Figure \ref{fig:sample_basis_fcts}. The energy corresponding to each basis pair $E_k = \| \Lambda_k \Psi_k^T \|_F^2$ is shown in Figure \ref{fig:sample_engs}.}

\begin{table*}[hbtp!]
\small
\setlength{\tabcolsep}{0.05pt}
\renewcommand{\arraystretch}{0.55}
\centering
\begin{tabular}{c}
\includegraphics[width=0.99\linewidth]{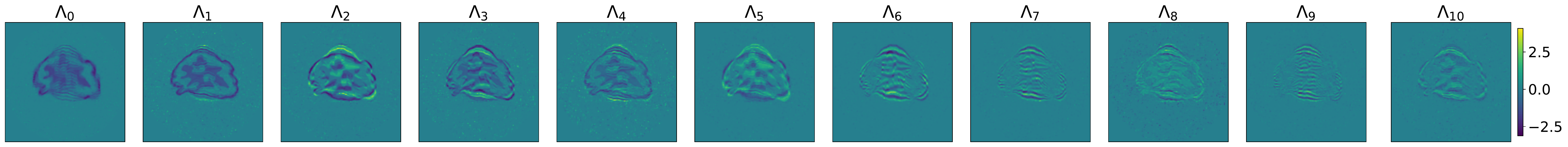} \\ 
\hspace{-0.75cm}
\includegraphics[width=0.99\linewidth]{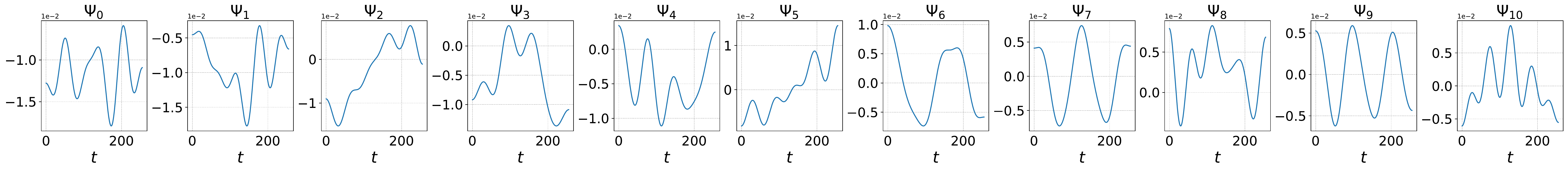}
\end{tabular}
\captionof{figure}{\small {Reconstructed $K=11$ spatial ($\Lambda$) and temporal ($\Psi$) basis functions for the time-varying walnut with $P=256$ and $U$ as DCT-II basis.}}
\label{fig:sample_basis_fcts}
\end{table*}

\begin{figure}
    \centering
    \includegraphics[width=0.75\linewidth]{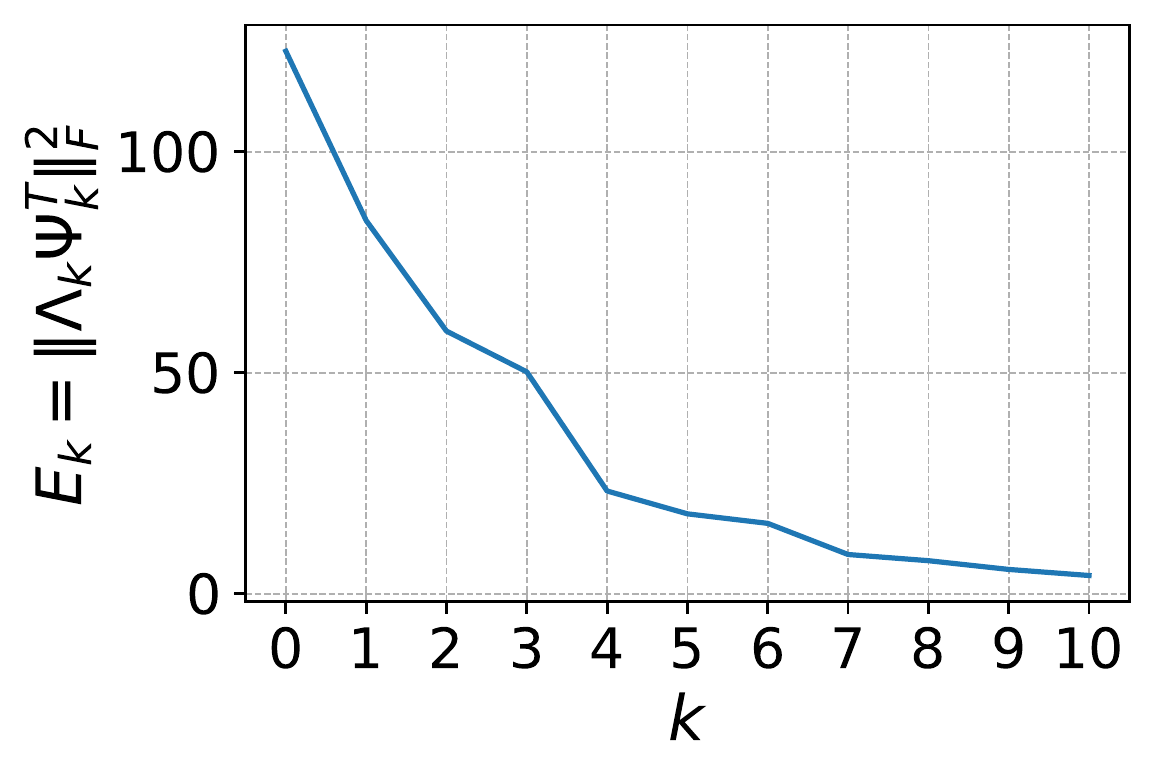}
    \caption{\small{$E_k = \| \Lambda_k \Psi_k^T \|_F^2$ for each spatiotemporal basis pair for the setting described in Section \ref{sec:supp_sample_basis}.}}
    \label{fig:sample_engs}
\end{figure}


\subsection{{Proof of Theorem \ref{thm:stat_pt_conv_thm}}}
\label{appendix:proof_conv_thm}
In this section, we provide the detailed steps for the proof of Theorem \ref{thm:stat_pt_conv_thm}.

We begin by establishing some simple consequences of the assumed properties of the denoiser. 
First, the objective $H(f, \Lambda, \Psi)$ \eqref{eq:obj_fidelity_psm} is lower bounded by $\underbar{H} \in \mathbb{R}$ because the RED regularizer {$\bar{\rho}(f)$} is non-negative, as follows from the strong passivity assumption on {${D}_\phi$.}
\begin{align}
    \bar{\rho}(f) & = \langle f, (f - \bar{D}_\phi(f)) \rangle = \|f\|_F^2 - \langle f, \bar{D}_\phi(f) \rangle \nonumber \\
    & = \sum_t \|f_t\|_2^2 - f_t^T D_\phi(f_t) 
    \\ 
    & \geq \sum_t \|f_t\|_2^2 - \|f_t\|_2 \|D_\phi(f_t)\|_2 \geq 0.
    \label{eq:red_lower_bnd}
\end{align}

Second, we have the following result.
\begin{lemma}
\label{lemma:Lipschitz_Gradient_D}
If the denoiser ${D}_\phi$ 
satisfies the gradient rule in \eqref{eq:red_grad_rule} and
$L_D$ Lipschitz continuous, then the regularizer $\bar{\rho}(f)$ is gradient Lipschitz continuous with gradient Lipschitz constant $L_{\nabla \rho} = 1 + L_D$.
\end{lemma}
\begin{proof}
Let $f_{(1)}$ and $f_{(2)}$ denote two different spatio-temporal objects. 
By the gradient rule of ${\rho}$, \eqref{eq:red_grad_rule}, and the assumed $L_D$ Lipschitz continuity of the denoiser ${D_\phi}$ we have
\begin{align*}
    &\| \nabla \bar{\rho}(f_{(1)}) - \nabla \bar{\rho}(f_{(2)}) \|_F \nonumber \\ &= \| f_{(1)} - f_{(2)} - (\bar{D}_\phi(f_{(1)}) - \bar{D}_\phi(f_{(2)})) \|_F \nonumber \\
    &\leq \| f_{(1)} - f_{(2)} \|_F + \| \bar{D}_\phi(f_{(1)}) - \bar{D}_\phi(f_{(2)}) \|_F \nonumber \\
    &\leq (1 + L_D)\| f_{(1)} - f_{(2)} \|_F.
\end{align*}
\vspace{-0.25cm}
\end{proof}

Third, we establish strong convexity of the objectives in the different subproblems in Algorithm \ref{alg:admm_psm_red}.
Clearly, for any choice of positive constant $\xi$ the objectives $S_\Lambda$ {\eqref{eq:SLambda}} and $S_\Psi$ {\eqref{eq:SPsi}} in the subproblems for $\Lambda$ and $\Psi$ are strongly convex with {moduli $\alpha_\Lambda \geq \xi$ and $\alpha_\Psi \geq \xi$,} respectively.

As we show next, with the assumed $\beta>L$, 
the objective $S_f$ {in \eqref{eq:Sf}} 
is strongly convex too with modulus $\alpha_f \geq \beta-L >0$. 

To prove this, we rewrite $S_f=\omega(f) + \frac{1}{2} (\beta - L) {\|f\|_F^2}$
where 
\begin{align*}
\omega(f) &\triangleq
\lambda \bar{\rho}(f) + \frac{L}{2}\| f\|_F^2- \beta\langle f, \Lambda^{(i)} \Psi^{(i)T} + \gamma^{(i-1)} \rangle \\
&+
\frac{\beta}{2}\|\Lambda^{(i)} \Psi^{(i)T} + \gamma^{(i-1)} \|_F^2
\end{align*}
and show that $\omega(f)$ is a convex function. 
Since $\frac{1}{2} (\beta - L){\|f\|_F^2}$ is convex quadratic for $\beta > L$, it then follows (by one of the alternative definitions of strong convexity) that $S_f$ is a strongly convex function with modulus $\alpha_f \geq \beta - L$.

All that remains, is to show that $\omega(f)$ is a convex function. This follows immediately by recalling Lemma~\ref{lemma:Lipschitz_Gradient_D}, that {$\lambda \bar{\rho}(f)$} is gradient Lipschitz continuous with gradient Lipschitz constant $L = \lambda(1 + L_D)$, and applying Lemma~\ref{lemma:thm21} below to the first two terms of $\omega(f)$.

\begin{lemma}{\normalfont{(Thm. 2.1 of \cite{zlobec2005liu})}}
\label{lemma:thm21}
    If {$\bar{\rho}(f): \mathbb{R}^{n} \rightarrow \mathbb{R}$} is Lipschitz continuously differentiable on a convex set $C$ with some Lipschitz constant $L$, then {$\phi(f) = \bar{\rho}(f) + \frac{1}{2}\beta {\|f\|_F^2}$} is a convex function on $C$ for every $\beta \geq L$.
\end{lemma}

With these preliminaries established, our proof of Theorem~\ref{thm:stat_pt_conv_thm} follows steps similar to \cite{hajinezhad2018alternating}(Sec. 4.2), but for our own Algorithm \ref{alg:admm_psm}:\\
\indent 1. Bounding the size of successive differences of the dual variables by those of the primal ones.
    
    2. Showing that $\mathcal{L}_\beta[f^i, \Lambda^i, \Psi^i; \gamma^i]$, the augmented Lagrangian, is a lower-bounded decreasing function.
     
    3. Combining the first two steps and showing convergence to a stationary solution.

\smallskip

The following result upper bounds the successive differences of the dual variable by those of the primal.

\begin{lemma}
\label{lemma:dual_primal_bnd_psm_fidelity}
Using the update rules in Algorithm \ref{alg:admm_psm}, the following holds:
\begin{align*}
\| \gamma^{(i)} - \gamma^{(i-1)} \|_F
\leq {\frac{L}{\beta}} 
\| f^{(i)} - f^{(i-1)} \|_F, 
\,\, \forall i \geq 1,
\end{align*}
where $L = \lambda (1+L_D)$.
\end{lemma}

\begin{proof}
Since $f^{(i)}$ is the optimal solution for Line~\ref{line:f_step_psm_fidelity_alg1} of Algorithm \ref{alg:admm_psm}, it should satisfy the optimality condition 
\begin{align}
    &\lambda \nabla \bar{\rho}(f^{(i)}) - \beta (\gamma^{(i-1)} - \Lambda\Psi^{(i)T} + f^{(i)}) = 
    0 \nonumber \\
    \label{eq:dual_grad_d_rln}
    &\lambda \nabla \bar{\rho}(f^{(i)}) - \beta\gamma^{(i)} = 0
\end{align}
where \eqref{eq:dual_grad_d_rln} is due to the dual variable update Step~\ref{line:dual_step_psm_fidelity_alg1} in Algorithm \ref{alg:admm_psm}.
Then, using Lemma~\ref{lemma:Lipschitz_Gradient_D}, we have
\begin{align}
    \| \gamma^{(i)}& - \gamma^{(i-1)} \|_F = \frac{\lambda}{\beta} \| \nabla \bar{\rho}(f^{(i)}) - \nabla \bar{\rho}(f^{(i-1)}) \|_F \nonumber \\
    &\leq \frac{\lambda}{\beta}(1 + L_D)\| f^{(i)} - f^{(i-1)} \|_F = \frac{L}{\beta}\| f^{(i)} - f^{(i-1)} \|_F. \nonumber
\end{align}
\vspace{-0.25cm}
\end{proof}

Next, using Lemma \ref{lemma:dual_primal_bnd_psm_fidelity}, we show that the augmented Lagrangian is decreasing and bounded below.
\begin{lemma}
\label{lemma:aug_lang_upper_bnd_psm_fidelity}
    In Algorithm~\ref{alg:admm_psm}, if $\beta > 2L$,
    then the following two propositions hold true. \\
    1. Successive differences of the augmented Lagrangian function \eqref{eq:aug_lag_1} are bounded above by
    \begin{align}
    \label{eq:aug_lang_upper_bnd_psm_fidelity}
        \mathcal{L}_\beta&[f^{(i)}, \Lambda^{(i)}, \Psi^{(i)}; \gamma^{(i)}] - \mathcal{L}_\beta[f^{(i-1)}, \Lambda^{(i-1)}, \Psi^{(i-1)}; \gamma^{(i-1)}] \nonumber \\ 
        &\leq -C_f \| f^{(i)} - f^{(i-1)} \|_F^2 - C_\Psi \| \Psi^{(i)} - \Psi^{(i-1)} \|_F^2 \nonumber \\&\quad\, - C_\Lambda \| \Lambda^{(i)} - \Lambda^{(i-1)} \|_F^2,
    \end{align}
    where $C_f = \frac{\alpha_f}{2} - \frac{L^2}{\beta}$, $C_\Lambda = \frac{\alpha_\Lambda}{2}$, and $C_\Psi = \frac{\alpha_\Psi}{2}$  are positive constants. That is, the augmented Lagrangian is monotone decreasing.
    
    2. There exists an $\underline{\mathcal{L}_\beta}$ such that
    \begin{align*}
        \mathcal{L}_\beta&[f^{(i)}, \Lambda^{(i)}, \Psi^{(i)}; \gamma^{(i)}] \geq \underline{\mathcal{L}_\beta}.
    \end{align*}
    That is, the augmented Lagrangian is lower bounded.
\end{lemma}

\begin{proof}
\textbf{Part 1.}
 
For conciseness define $W^{(i)}=(\Psi^{(i)}, \Lambda^{(i)}, f^{(i)}; \gamma^{(i)})$. Then the successive difference of the augmented Lagrangian can be expressed by adding and subtracting the term $\mathcal{L}_\beta [\Psi^{(i)}, \Lambda^{(i)}, f^{(i)}; \gamma^{(i-1)}]$, 
\begin{align}
\label{eq:aug_lang_split_psm_fid}
    &\mathcal{L}_\beta[W^{(i)}] - \mathcal{L}_\beta[W^{(i-1)}] \nonumber \\
    &\quad= \mathcal{L}_\beta[W^{(i)}] -\mathcal{L}_\beta[\Psi^{(i)}, \Lambda^{(i)}, f^{(i)}; \gamma^{(i-1)}] \nonumber \\ 
    &\quad+ \mathcal{L}_\beta[\Psi^{(i)}, \Lambda^{(i)}, f^{(i)}; \gamma^{(i-1)}] - \mathcal{L}_\beta[W^{(i-1)}].
\end{align}
For the first two terms on the RHS, using Lemma \ref{lemma:dual_primal_bnd_psm_fidelity} 
we have 
\begin{align}
    \mathcal{L}_\beta&[W^{(i)}] -\mathcal{L}_\beta[\Psi^{(i)}, \Lambda^{(i)}, f^{(i)}; \gamma^{(i-1)}] \nonumber \\ 
    &= {\beta}\langle\gamma^{(i)} - \gamma^{(i-1)}, (\Lambda \Psi^T)^{(i)} - f^{(i)}\rangle \nonumber \\
    &= \| \gamma^{(i)} - \gamma^{(i-1)} \|_F^2 \leq {L^2}\|f^{(i)} - f^{(i-1)}\|_F^2. \label{eq:dual_pr_bnd}
\end{align}
where the first equality in \eqref{eq:dual_pr_bnd} follows from the  dual variable update in Step \eqref{line:dual_step_psm_fidelity_alg1} of Algorithm \ref{alg:admm_psm}.
The last two terms on the RHS of \eqref{eq:aug_lang_split_psm_fid} are further split into
\begin{align}
\label{eq:second_split_psm_fidelity}
    &\mathcal{L}_\beta[\Psi^{(i)}, \Lambda^{(i)}, f^{(i)}; \gamma^{(i-1)}] - \mathcal{L}_\beta[W^{(i-1)}] \nonumber \\ &= \mathcal{L}_\beta[\Psi^{(i)}, \Lambda^{(i)}, f^{(i)}; \gamma^{(i-1)}] 
    - \mathcal{L}_\beta[\Psi^{(i)}, \Lambda^{(i)}, f^{(i-1)}; \gamma^{(i-1)}] \nonumber \\
    &\quad + \mathcal{L}_\beta[\Psi^{(i)}, \Lambda^{(i)}, f^{(i-1)}; \gamma^{(i-1)}] \nonumber \\
    &\quad- \mathcal{L}_\beta[\Psi^{(i-1)}, \Lambda^{(i)}, f^{(i-1)}; \gamma^{(i-1)}]\nonumber \\
    &\quad + \mathcal{L}_\beta[\Psi^{(i-1)}, \Lambda^{(i)}, f^{(i-1)}; \gamma^{(i-1)}] - \mathcal{L}_\beta[W^{(i-1)}]
\end{align}
The first two terms on the RHS of \eqref{eq:second_split_psm_fidelity} are upper bounded using the strong convexity of $S_f$
with modulus $\alpha_f$:
\begin{align}
\label{eq:f_step_bnd}
    \mathcal{L}&_\beta[\Psi^{(i)}, \Lambda^{(i)}, f^{(i)}; \gamma^{(i-1)}] - \mathcal{L}_\beta[\Psi^{(i)}, \Lambda^{(i)}, f^{(i-1)}; \gamma^{(i-1)}] \nonumber \\
    &\leq \langle \nabla_f S_f[\Psi^{(i)}, \Lambda^{(i)}, f^{(i)}; \gamma^{(i-1)}], f^{(i)} - f^{(i-1)} \rangle \nonumber \\ 
    &\quad\, - \frac{\alpha_f}{2} \| f^{(i-1)} - f^{(i)} \|_F^2 = - \frac{\alpha_f}{2} \| f^{(i-1)} - f^{(i)} \|_F^2,
\end{align}
where the last line of \eqref{eq:f_step_bnd} is due to the optimality condition of Line \ref{line:f_step_psm_fidelity_alg1} of Algorithm \ref{alg:admm_psm} for $f^{(i)}$.

The next two terms on the RHS of \eqref{eq:second_split_psm_fidelity} are upper bounded using the strong convexity of $S_\Psi$
with modulus $\alpha_\Psi$ as
\begin{align}
    &\mathcal{L}_\beta[\Psi^{(i)}, \Lambda^{(i)}, f^{(i-1)}; \gamma^{(i-1)}] - \mathcal{L}_\beta[\Psi^{(i-1)}, \Lambda^{(i)}, f^{(i-1)}; \gamma^{(i-1)}] \nonumber \\
    &= S_\Psi(\Lambda^{(i)}, \Psi^{(i)}, f^{(i-1)}; \gamma^{(i-1)}) \nonumber \\
    &\quad - S_\Psi(\Lambda^{(i)}, \Psi^{(i-1)}, f^{(i-1)}; \gamma^{(i-1)}) \nonumber
    \\
    &\leq \langle \nabla_\Psi S_\Psi[\Psi^{(i)}, \Lambda^{(i)}, f^{(i-1)}; \gamma^{(i-1)}], \Psi^{(i)} - \Psi^{(i-1)} \rangle \nonumber \\ 
    &\quad - \frac{\alpha_\Psi}{2}\| \Psi^{(i)} - \Psi^{(i-1)}\|_F^2 = - \frac{\alpha_\Psi}{2} 
    \| \Psi^{(i)} - \Psi^{(i-1)}\|_F^2.
    \nonumber
\end{align}
Repeating similar steps for the last two terms on the RHS of \eqref{eq:second_split_psm_fidelity} using the objective $S_\Lambda$ leads to the upper bound of $- \frac{\alpha_\Lambda}{2} \| \Lambda^{(i)} - \Lambda^{(i-1)}\|_F^2$.

Finally, we establish the positivity of the constants in \eqref{eq:aug_lang_upper_bnd_psm_fidelity}. We have $c_\Psi = \alpha_\Psi \geq \xi >0$ and $c_\Lambda = \alpha_\Lambda \geq \xi >0$. Next, because $\alpha_f \geq \beta-L$, we have $C_f = \frac{\alpha_f}{2} - {L^2} \geq \frac{\beta-L}{2} - {L^2} >0$, where the last inequality follows by the assumption $\beta > 2L$.

Combining these results, we achieve the inequality in \eqref{eq:aug_lang_upper_bnd_psm_fidelity}, which concludes Part 1.

\textbf{Part 2.}
For Part 2, we need to show that the augmented Lagrangian is lower bounded.
\begin{align}
    &L_\beta[W^{(i)}] = H(f^{(i)}, \Lambda^{(i)}, \Psi^{(i)}) 
    \nonumber\\
    & + {\beta}\langle \gamma^{(i)}, (\Lambda\Psi^T)^{(i)} - f^{(i)} \rangle 
    + \frac{\beta}{2}\| (\Lambda\Psi^T)^{(i)} - f^{(i)} \|^2_F \label{eq:AugLag} \\
    \label{eq:dualvar_lips_rlt}
    &\geq H(f^{(i)}, \Lambda^{(i)}, \Psi^{(i)}) \nonumber 
    + \langle \lambda\nabla {\bar{\rho}}(f^{(i)}), (\Lambda\Psi^T)^{(i)} - f^{(i)} \rangle 
    \nonumber \\ 
    & 
    + \frac{L}{2}\| (\Lambda\Psi^T)^{(i)} - f^{(i)} \|^2_F \\
    \label{eq:D_lip_cont}
    &\geq H(f^{(i)}, \Lambda^{(i)}, \Psi^{(i)}) 
    + \lambda {\bar{\rho}}((\Lambda\Psi^{T})^{(i)}) - \lambda {\bar{\rho}}(f^{(i)}) \\ 
    &= H(\Lambda^{(i)}\Psi^{(i)T}, \Lambda^{(i)}, \Psi^{(i)}). 
     \label{eq:lower_bnd_aug_lang}
\end{align}
where \eqref{eq:dualvar_lips_rlt} follows by using \eqref{eq:dual_grad_d_rln} to
 replace the second term of \eqref{eq:AugLag} by that of \eqref{eq:dualvar_lips_rlt} and assuming $\beta > L$, and \eqref{eq:D_lip_cont} is due to the gradient Lipschitz continuity of the regularizer {$\bar{\rho}(\cdot)$},
\begin{align*}
    \lambda {\bar{\rho}}(\Lambda\Psi^T) &\leq \lambda {\bar{\rho}}(f) + \langle \lambda \nabla_f {\bar{\rho}}(f), \Lambda\Psi^T - f \rangle \nonumber 
    + \frac{L}{2}\| \Lambda\Psi^T - f \|_F^2.
\end{align*}

Then, since we initially showed that our minimization objective in \eqref{eq:obj_fidelity_psm} is lower bounded by $\underbar{H}$, we have 
\begin{align*}
H(\Lambda^{(i)}\Psi^{(i)T}, \Lambda^{(i)}, \Psi^{(i)}) \geq \underbar{H},
\end{align*}
which leads, by combining with \eqref{eq:lower_bnd_aug_lang}, to $L_\beta[W^{(i)}] \geq \underline{\mathcal{L}_\beta} = \underbar{H}$ and completes the proof of Lemma~\ref{lemma:aug_lang_upper_bnd_psm_fidelity}.
\end{proof}

By Lemma \ref{lemma:aug_lang_upper_bnd_psm_fidelity}, the sequence $L_\beta[W^{(i)}]$ of augmented Lagrangian values converges. What remains is to address the convergence of the objective $H$ and the sequence of \emph{iterates} $(f^{(i)}, \Lambda^{(i)}, \Psi^{(i)}, \gamma^{(i)})$ to a stationary solution of \eqref{eq:hard_cnst_RED_obj}, which we do in the third and final step of the proof of Theorem~\ref{thm:stat_pt_conv_thm}. We start by showing that the duality gap shrinks to zero.
\begin{lemma}
\label{lemma:duality_gap}
Consider Algorithm \ref{alg:admm_psm} 
for solving problem \eqref{eq:hard_cnst_RED_obj}. If $\beta > 2L$, then the duality gap goes to zero, i.e.
\begin{align*}
    \lim_{i \rightarrow \infty} \| (\Lambda\Psi^T)^{(i)} - f^{(i)} \|_F \rightarrow 0.
\end{align*}
\end{lemma}
\begin{proof}
    Lemma \ref{lemma:aug_lang_upper_bnd_psm_fidelity} implies that when $\beta > 2L$, the sequence of augmented Lagrangian values is convergent. Thus, the LHS of \eqref{eq:aug_lang_upper_bnd_psm_fidelity} converges to zero. This also means that each of the norms on the RHS converges to zero to satisfy the inequality. Hence, the following successive differences converge to zero,
    \begin{align}
        f^{(i)} - f^{(i-1)} \rightarrow 0; \Lambda^{(i)} - \Lambda^{(i-1)} \rightarrow 0; \Psi^{(i)} - \Psi^{(i-1)} \rightarrow 0. \nonumber
    \end{align}
    Furthermore, applying Lemma \ref{lemma:dual_primal_bnd_psm_fidelity}, {we have $\gamma^{(i)} - \gamma^{(i-1)} \rightarrow 0.$} 
    Finally, using the update equation for the dual variable $\gamma$, we obtain
    \begin{align}
        (\Lambda\Psi^T)^{(i)} - f^{(i)} \rightarrow 0,
        \label{eq:DualGapToZero}
    \end{align}
    which completes the proof of Lemma~\ref{lemma:duality_gap}.
\end{proof}
\vspace{-0.25cm}

It follows from Lemma~\ref{lemma:duality_gap} that the gap $L_\beta[W^{(i)}] - H(\Lambda^{(i)}\Psi^{(i)T}, \Lambda^{(i)}, \Psi^{(i)})$ between the augmented Lagrangian and the objective converges to zero, and therefore thanks to the convergence of $L_\beta[W^{(i)}]$, the objective $H$ converges too. This establishes Part (i) of Theorem~\ref{thm:stat_pt_conv_thm}.

Turning to Part (ii) of Theorem~\ref{thm:stat_pt_conv_thm}, we first show that the iterates are bounded.

\begin{lemma}
\label{lemma:iter_bnd}
The iterates $(f^{(i)}, \Lambda^{(i)}, \Psi^{(i)}, \gamma^{(i)})$ of Algorithm \ref{alg:admm_psm} for solving \eqref{eq:hard_cnst_RED_obj} are bounded for $i > 0$.
\end{lemma}

\begin{proof}

We use the following facts:
\begin{itemize}
    \item[(F1)] $\mathcal{L}_\beta [W^{(i)}]$ is a non-increasing function as $i$ increases for all $i > 0$, as shown in Lemma \ref{lemma:aug_lang_upper_bnd_psm_fidelity}.
    \item[(F2)] The objective $H(f, \Lambda, \Psi)$ is coercive with respect to the second and the third variables.
\end{itemize}

It follows from (F1) that
$\mathcal{L}_\beta[W^{(i)}] \leq \mathcal{L}_\beta[W^{(0)}] , \,\, \forall i \geq 0$.
Combining this with \eqref{eq:lower_bnd_aug_lang}, yields
\begin{align}
    &H((\Lambda\Psi^{T})^{(i)}, \Lambda^{(i)}, \Psi^{(i)}) \leq \mathcal{L}_\beta[W^{(i)}] \leq
    \mathcal{L}_\beta[W^{(0)}]  \label{eq:bnd_iterate_aug_lang}.
\end{align}

Now, note  that by \eqref{eq:hard_cnst_RED_obj}, the first variable $(\Lambda\Psi^T)^{(i)}$  in $H((\Lambda\Psi^{T})^{(i)}, \Lambda^{(i)}, \Psi^{(i)})$ only appears as the argument of ${\bar{\rho}}(\cdot)$, and by \eqref{eq:red_lower_bnd}
${\bar{\rho}}(\cdot) \geq 0$. This implies by (F2) that  
$H((\Lambda\Psi^{T})^{(i)}, \Lambda^{(i)}, \Psi^{(i)})$ is coercive with respect to $\Lambda^{(i)}$ and $\Psi^{(i)}$. 
Combining this with the boundedness of $H$ in \eqref{eq:bnd_iterate_aug_lang} implies the boundedness of these iterates. 

Next, the boundedness of $f^{(i)}$ follows from Lemma \ref{lemma:duality_gap}, since we showed that $\Lambda^{(i)}$ and $\Psi^{(i)}$, thus $(\Lambda\Psi^T)^{(i)}$ are all bounded, and the duality gap shrinks to zero as the iterates progress.

Finally, since all the primary variables are bounded, the augmented Lagrangian is
bounded both from below (by Part 2 of Lemma \ref{lemma:aug_lang_upper_bnd_psm_fidelity}) and from above by \eqref{eq:bnd_iterate_aug_lang}, and the dual variable iterates $\gamma^{(i)}$ appear only in the linear Lagrangian term, they are also bounded.

This concludes the proof of Lemma~\ref{lemma:iter_bnd}.
\end{proof}
\vspace{-0.2cm}

It follows that the sequences of iterates $\{f^{(i)}, \Lambda^{(i)}, \Psi^{(i)}, \gamma^{(i)})\}$ has at least one accumulation point $(f^*, \Lambda^*, \Psi^*, \gamma^*)$. Next, we establish the properties of all such accumulation points.
    
Taking the limit over $i$ in \eqref{eq:dual_grad_d_rln} yields \eqref{eq:LagStationary_f}.

Next, since $\Lambda^{(i)}$ is optimal for the subproblem in Line \ref{line:lambda_step_psm_fidelity_alg1} of Algorithm \ref{alg:admm_psm}, we have
\begin{align}
    &\nabla_\Lambda S_\Lambda(\Psi^{(i-1)}, \Lambda^{(i)}, f^{(i-1)}; \gamma^{(i-1)}) = 0 \nonumber \\ 
    &2\bar{R}^T(\bar{R}\Lambda^{(i)}\Psi^{(i-1)T}\Psi^{(i-1)} - g\Psi^{(i-1)}) \nonumber \\
    &\,\, + \beta (\Lambda^{(i)}\Psi^{(i-1)T} - f^{(i-1)} + {\gamma^{(i-1)}})\Psi^{(i-1)} + \xi \Lambda^{(i)}
    = 0 \nonumber
\end{align}
Taking the limit over $i$ yields \eqref{eq:LagStationary_Lambda}. Similarly, taking the limit in the optimality condition with respect to $\Psi^{(i)}$, $\nabla_\Psi S_\Psi(\Psi^{(i)},\Lambda^{(i-1)},f^{(i-1)};\gamma^{(i-1)}) =0$ yields \eqref{eq:LagStationary_Psi}. Finally, taking the limit with respect to $i$ in \eqref{eq:DualGapToZero} verifies \eqref{eq:LagStationary_Feasible}. These results complete the proof of Theorem \ref{thm:stat_pt_conv_thm}.
\qed
\vspace{-0.25cm}

\end{document}